\documentclass[journal,onecolumn,letterpaper,10pt]{IEEEtran}
\usepackage[utf8]{inputenc}
\usepackage[T1]{fontenc}
\usepackage{url}
\usepackage{ifthen}
\usepackage{cite}
\usepackage[cmex10]{amsmath}
\usepackage{array}
\usepackage{booktabs,multirow,mathtools,enumitem}
\usepackage{amssymb,graphicx}
\usepackage[dvipsnames]{xcolor}
\newtheorem{Theo}{\bf Theorem}
\newtheorem{Lem}{\bf Lemma}
\newtheorem{Cor}{\bf Corollary}
\newtheorem{Propo}{\bf Proposition}
\newtheorem{Def}{Definition}
\newtheorem{Remark}{Remark}
\newtheorem{Example}{Example}
\usepackage{arydshln}
\usepackage{placeins}
\usepackage{needspace}
\usepackage[hidelinks,unicode]{hyperref}
\hypersetup{pdftitle={Secure Multi-Access Coded Caching: A Lifting Approach},pdfauthor={Han Fang, Qinyi Lu, Nan Liu, and Wei Kang}}
\allowdisplaybreaks

\begin{document}
\bstctlcite{IEEEtranBSTcontrol}
\title{Secure Multi-Access Coded Caching:\\
A Lifting Approach}

\author{Han Fang, Qinyi Lu, Nan Liu, and Wei Kang%
\thanks{This paper was presented in part at the 2024 IEEE Information Theory Workshop (ITW), Shenzhen, China, Nov. 2024. The DOI of the conference paper is 10.1109/ITW61385.2024.10807027.}%
\thanks{Han Fang is with the School of Cyber Science and Engineering, Southeast University, Nanjing, China (e-mail: h\_fang@seu.edu.cn).}%
\thanks{Qinyi Lu is with the School of Information Science and Engineering, Southeast University, Nanjing, China (e-mail: qylu@seu.edu.cn).}%
\thanks{Nan Liu is with the National Mobile Communications Research Laboratory and the School of Cyber Science and Engineering, Southeast University, Nanjing, China (e-mail: nanliu@seu.edu.cn).}%
\thanks{Wei Kang is with the School of Information Science and Engineering and the School of Cyber Science and Engineering, Southeast University, Nanjing, China (e-mail: wkang@seu.edu.cn).}}
\maketitle

\begin{abstract}
We construct secure cyclic multi-access coded caching schemes by re-encoding the
caches of a secure single-access scheme and reusing its multicast message
unchanged.  For a library of $N$ files, each of $K$ users reads $L<K$
consecutive shared caches and must
recover its requested file while learning no information about the other
files.  Starting from a parent scheme with a cyclic representation of its
cached symbols, the transformation lets each access window recover the
corresponding parent-cache symbols and, conditioned on them, learn nothing
about the remaining symbols.  A deterministic Pascal-window encoder handles
symbols whose consecutive appearances span at least $L$ users.  A randomized
complement-dual encoder handles shorter spans.  Over a suitable finite field,
these encoders attain the exact minimum per-block storage for every input length in the
block-separable, parent-preserving class.  Applying the transformation to two
secure single-access schemes gives an achievable memory-rate family for
$N,K\ge2$ and $1\le L<K$, including the exact rate-one cache size
$(N-1)(K-1)/L$.  Using converse bounds for arbitrary secure multi-access
schemes, we prove approximation guarantees for the achievable family and
give a proof of the exact $L=K-1$ tradeoff.
\end{abstract}


\section{Introduction}

Secure single-access coded caching lets each user recover its requested file
from one cache and a common multicast message while keeping the other files
secret \cite{PrivateCodedCaching,FangSmallRate}.  We show how to retain this
delivery message when each user instead reads an overlapping window of shared
caches.  The network has $N$ files, $K$ users, and $K$ caches on a cycle, with
user $k$ reading $L<K$ consecutive caches
\cite{CodedCachingForMultiLevelPopularityAndAccess}.  The construction starts
from a secure single-access scheme, called the \emph{parent}, together with a
cyclic representation of which users' caches contain each parent symbol.

Cyclic multi-access transformations already provide ways to recover logical
cache contents from overlapping physical windows
\cite{ANovelTransformationApproachOfSharedLinkCodedCachingSchemesForMultiaccessNetworks}.
To reuse a secure scheme's multicast message, each cache window must also
preserve the parent's secrecy guarantee.
Our construction meets both requirements: each user's window reconstructs its
parent-cache symbols and, conditioned on those symbols, reveals no information
about the remaining parent variables.  The user can therefore apply the
parent decoder to the original multicast message.  The conditional secrecy property carries
the parent's file-secrecy guarantee to the physical window, so the delivery
message and rate remain unchanged. We call this transformation \emph{parent-preserving lifting}.  To construct it,
we list the consecutive groups of users whose parent caches contain a given
symbol.  Each maximal group is an \emph{appearance run}, and its length is the
run's span $b$.  We arrange runs of the same span into cyclic blocks and encode
each block across the physical caches.  Each block is designed to satisfy
the local recovery and conditional secrecy requirements.  Combining the
encoded blocks gives the complete physical-cache placement. Two encoders handle these blocks.  For $b\ge L$, a deterministic encoder uses
invertible submatrices of a Pascal matrix to recover the required symbols.
For $b<L$, we use the complementary parameters $(K,K-L,K-b)$ to build a
system of linear equations, then choose a solution using independent
randomness.  The random choice hides the symbols outside each user's target set while
preserving exact recovery of the target symbols.

Our main contributions are as follows.
\begin{itemize}[leftmargin=*]
\item \emph{Secure reuse of the parent broadcast.}  We give local recovery
and secrecy conditions under which a physical cache window can replace its
parent cache.  For a secure parent supplied with a finite cyclic-incidence
representation, the lifting preserves correct decoding and user-level file
secrecy while leaving the multicast message unchanged.

\item \emph{Explicit encoders and achievable memory-rate points.}  We construct
a deterministic Pascal-window encoder for $b\ge L$ and a randomized
complement-dual encoder for $b<L$.  Applying them to the
subset-indexed parent of \cite{PrivateCodedCaching} gives achievable
memory-rate points for arbitrary $N$, $K$, and $1\le L<K$.  Lifting the
rate-one parent of \cite{FangSmallRate} gives the globally exact cache size
$(N-1)(K-1)/L$ at rate one.

\item \emph{Exact local storage and global memory-rate bounds.}  The two
encoders attain the minimum per-block allocation for every appearance span in
the block-separable, parent-preserving design class.  For arbitrary secure
multi-access schemes, we combine a secrecy-aware cyclic-window converse with
a necessary memory lower bound to prove approximation guarantees for the
achievable family.  The cyclic-window converse used in the conference analysis
\cite{FangLargestAccess} was established in
\cite[Theorem~6]{ImprovedLowerBoundsForMultiAccessCodedCaching}.
\end{itemize}

The conference paper \cite{FangLargestAccess} established the exact $L=K-1$
memory-rate tradeoff.  This paper provides the lifting construction for all
access degrees, establishes approximation guarantees, and includes a complete
proof of the conference result.

\subsection{Related work}
Coded caching creates multicast opportunities by placing data in the caches
before demands are known.  During delivery, the server uses the cached data
as side information to serve several users with a common message
\cite{FundamentalLimitsofCaching,DecentralizedCodedCachingAttainsOrderOptimalMemoryRateTradeoff}.
User-level file secrecy additionally requires that each user learn no
information about the files it does not request.  The scheme in
\cite{PrivateCodedCaching} encodes files into secret shares and masks delivery
messages.  Its subset-indexed cache contents supply the parent for our
general achievable family.  Placement-delivery arrays also provide secretive
constructions with reduced subpacketization \cite{SecretiveCodedCachingFromPDAs}.
Our second parent is the unit-rate scheme of \cite{FangSmallRate}, whose
lifting gives the exact rate-one cache size.

In cyclic multi-access caching, each user's side information is the content
of a consecutive cache window.  Overlapping windows give users shared side
information, which affects both cache placement and multicast delivery.
Schemes with uncoded placement and their memory-rate
bounds are studied in
\cite{RateMemoryTradeoffForMultiaccessCodedCachingWithUncodedPlacement},
while \cite{StructuredIndexCodingProblemAndMultiAccessCodedCaching} uses
structured index-coding problems to obtain achievable rates and bounds.
The window geometry also supports caching gains beyond cache redundancy
\cite{MultiaccessCodedCachingGainsBeyondCacheRedundancy}.
The transformation in
\cite{ANovelTransformationApproachOfSharedLinkCodedCachingSchemesForMultiaccessNetworks}
builds multi-access schemes from shared-link placement-delivery arrays.
Coded placement and cyclic placement strategies provide further ways to
exploit the shared cache contents
\cite{MultiAccessCodedCachingWithCodedPlacement,CyclicPlacementStrategyForMultiAccessCodedCaching}.
Converse bounds constrain the achievable memory-rate region
\cite{ImprovedLowerBoundsForMultiAccessCodedCaching}, and exact tradeoffs have
been established in large-cache regimes
\cite{TheOptimalRateMemoryTradeoffInMultiAccessCodedCachingLargeCacheSize}.

The number of pieces into which each file is divided is another design
consideration.  Index-coding methods and placement-delivery arrays give
cyclic multi-access constructions with linear subpacketization
\cite{ACodedCachingSchemeWithLinearSubpacketizationAndItsApplicationToMultiAccessCodedCaching,MultiAccessCodedCachingSchemeWithLinearSubPacketizationUsingPDAs,MultiaccessCodedCachingwithLinearSubpacketization}.
Cross resolvable designs provide constructions for other cache-access
patterns, with the user-cache associations determined by the design
\cite{MultiAccessCodedCachingSchemesFromCrossResolvableDesigns}.

Security and privacy requirements are determined by the protected
information and the observer.  Secure delivery keeps the library secret
from an external observer of the broadcast
\cite{FundamentalLimitsOfCachingWithSecureDelivery}.  Demand-private caching
keeps users' requested indices secret
\cite{DemandPrivateCodedCachingKamath,OnCodedCachingWithPrivateDemands,FundamentalLimitsOfDemandPrivateCodedCaching}.
Other formulations protect demand indices against colluding users
\cite{FundamentalLimitsOfCachingForDemandPrivacyAgainstColludingUsers}
or combine file secrecy with demand privacy
\cite{CodedCachingWithFileAndDemandPrivacy}.
The survey in \cite{SystematicReviewSecureCodedCaching} organizes these
security and privacy models.

These protection goals have also been studied in multi-access networks.
Secure-delivery schemes protect against an external broadcast observer
\cite{MultiAccessCodedCachingWithSecureDelivery}, and demand-private schemes
protect requested indices when users access several caches
\cite{MultiaccessCodedCachingWithPrivateDemands,MultiAccessCodedCachingWithDemandPrivacy}.
Cyclic multi-access networks have also been studied with private per-user
caches in addition to the shared caches
\cite{CyclicWrapAroundMultiAccessCodedCachingWithPrivateCaches}.

In combinatorial multi-access caching, each user is assigned to one cache
subset of a specified size, and the users cover all such subsets
\cite{MaddahAliNiesenSchemeForMultiaccessCodedCaching,FundamentalLimitsofCombinatorialMultiAccessCaching,CombinatorialMultiAccessCodedCachingImprovedRateMemoryTradeOffWithCodedPlacement}.
Security, privacy, and linear-function retrieval have been studied on this
topology
\cite{SecurityAndPrivacyInCacheAidedLinearFunctionRetrievalForMultiAccessCodedCaching,SecurityPrivacyLinearFunctionRetrievalCombinatorialPrivateCaches,SecrecyPrivacyMultiAccessCombinatorialTopology}.
Our construction addresses user-level file secrecy for cyclic cache windows.
It re-encodes the symbols of a secure single-access parent so that each
window recovers its parent-cache contents with conditional secrecy, allowing
the parent multicast message to be reused unchanged.

\section{System Model}

Fix integers
\begin{align*}
N\ge2,\qquad K\ge2,\qquad 1\le L<K.
\end{align*}
For a positive integer $n$, let $[n]\triangleq\{1,2,\ldots,n\}$.
For integers $a,b$, write $[a:b]\triangleq\{j\in\mathbb Z:a\le j\le b\}$
for the integer interval, which is empty when $a>b$.
For an ordered index set $\mathcal B$, write $Y_{\mathcal B}$ for the corresponding
ordered tuple of random variables.  User and cache indices lie in $[K]$ and are
interpreted cyclically through
\begin{align*}
\langle i\rangle_K
\triangleq 1+((i-1)\bmod K).
\end{align*}
The complement of an index set is taken relative to its stated ground set,
such as $[K]$ or $[N]$.

For a positive integer blocklength $F$, a code may choose any finite source
alphabet $\mathcal Q$ of size $q\ge2$.  All logarithms, entropies, and mutual
informations are measured in base-$q$ units.  The server stores mutually
independent files
\begin{align*}
W_n\sim\operatorname{Unif}(\mathcal Q^F),\qquad n\in[N],
\end{align*}
so $H(W_n)=F$ and $H(W_{[N]})=NF$.  It also has a private random seed
$S^{(F)}$, generated independently of the library and the demands before
placement.  Randomized placement and delivery maps are understood as
deterministic maps once this seed is included.

There are $K$ physical caches and $K$ users.  User $k$ accesses the cyclic
window
\begin{align}
\mathcal A_k
\triangleq
\{k,\langle k+1\rangle_K,\ldots,\langle k+L-1\rangle_K\}.
\end{align}
Let $\mathcal Z_k^{(F)}$ be the prescribed finite storage alphabet for cache
$k$.  During placement, before any demand is known, the placement map
$\phi_k^{(F)}$ assigns the cache content
\begin{align}
Z_k^{(F)}
=\phi_k^{(F)}(W_{[N]},S^{(F)})
\in\mathcal Z_k^{(F)},
\qquad k\in[K].
\label{eq:placement-map}
\end{align}

The demand vector $\mathbf D=(D_1,\ldots,D_K)$ has full support on $[N]^K$, is
independent of $(W_{[N]},S^{(F)})$, and is revealed publicly to the server and
all users.  Every vector $\mathbf d\in[N]^K$ is allowed, including vectors with
repeated entries.  For each $\mathbf d$, let $\mathcal X_{\mathbf d}^{(F)}$
be the prescribed finite broadcast alphabet.  When $\mathbf D=\mathbf d$ is
revealed, the server uses the delivery map $\psi_{\mathbf d}^{(F)}$ to transmit
\begin{align}
X_{\mathbf d}^{(F)}
=\psi_{\mathbf d}^{(F)}(W_{[N]},S^{(F)})
\in\mathcal X_{\mathbf d}^{(F)}.
\end{align}
User $k$ applies a public decoder
\begin{align*}
g_{k,\mathbf d}^{(F)}:
\left(\prod_{i\in\mathcal A_k}\mathcal Z_i^{(F)}\right)
\times\mathcal X_{\mathbf d}^{(F)}
\longrightarrow \mathcal Q^F.
\end{align*}

The finite-blocklength memory and worst-case delivery rate are defined by the
sizes of the prescribed alphabets:
\begin{align}
M_F
&\triangleq
\max_{k\in[K]}\frac{1}{F}\log_q|\mathcal Z_k^{(F)}|,
\\
R_F
&\triangleq
\max_{\mathbf d\in[N]^K}\frac{1}{F}
\log_q|\mathcal X_{\mathbf d}^{(F)}|.
\end{align}
These alphabet-size constraints imply
$H(Z_k^{(F)})\le M_FF$ and
$H(X_{\mathbf d}^{(F)})\le R_FF$ for every cache and demand vector,
respectively.  The converse proofs use these entropy bounds.

For every demand vector and every user, a valid code must provide exact
recovery:
\begin{align}
g_{k,\mathbf d}^{(F)}
\bigl(Z_{\mathcal A_k}^{(F)},X_{\mathbf d}^{(F)}\bigr)
=W_{d_k}\quad\text{with probability one},
\qquad
(\mathbf d,k)\in[N]^K\times[K].
\label{eq:correctness}
\end{align}
Equivalently,
$H(W_{d_k}\mid Z_{\mathcal A_k}^{(F)},X_{\mathbf d}^{(F)},
\mathbf D=\mathbf d)=0$.
The user-level file-secrecy requirement is
\begin{align}
I\!\left(
W_{[N]\setminus\{d_k\}};
Z_{\mathcal A_k}^{(F)},X_{\mathbf d}^{(F)}
\,\middle|\,
W_{d_k},\mathbf D=\mathbf d
\right)=0,
\quad
(\mathbf d,k)\in[N]^K\times[K].
\label{eq:security}
\end{align}
Equation~\eqref{eq:security} requires perfect secrecy against each user
individually in one delivery round.  Even after recovering its requested file,
the user learns no information about the other files from its accessible
caches and the broadcast.

\begin{Remark}[Equivalent secrecy identity]
By file independence and exact recovery, \eqref{eq:security} is equivalent to
\[
I\!\left(
W_{[N]\setminus\{d_k\}};
Z_{\mathcal A_k}^{(F)},X_{\mathbf d}^{(F)}
\,\middle|\,
\mathbf D=\mathbf d
\right)=0.
\]
\end{Remark}

A pair $(M,R)$ is \emph{securely achievable} for $(K,L,N)$ if there is a
sequence of codes with file lengths $F_j\to\infty$ such that every code in the
sequence satisfies exact recovery \eqref{eq:correctness} and perfect secrecy
\eqref{eq:security}, and
\begin{align*}
\limsup_{j\to\infty}M_{F_j}\le M,
\qquad
\limsup_{j\to\infty}R_{F_j}\le R.
\end{align*}
\begin{Remark}[Normalization along a code sequence]
The source alphabet may vary along the sequence.  For each code, memory and
rate are measured using the logarithm of that code's alphabet size.
\end{Remark}
The optimal rate-memory function is
\begin{align}
R^*_{K,L,N}(M)
\triangleq
\inf\bigl\{R\ge0:(M,R)\text{ is securely achievable for }(K,L,N)\bigr\},
\end{align}
with $\inf\varnothing=+\infty$.  When $(K,L,N)$ is fixed, we abbreviate this
quantity as $R^*(M)$.  A cache size $M$ is \emph{feasible} if
\mbox{$R^*(M)<\infty$}.

For $N,K\ge2$, exact recovery and file secrecy require the access range
$L<K$.
\begin{Remark}[Full cache access]
At $L=K$, choose a demand vector containing two distinct requests.  Such a
vector exists because $N,K\ge2$.  Every user then observes all caches together
with $X_{\mathbf d}^{(F)}$, so the corresponding public decoders recover both
requested files from the common observation
$(Z_{[K]}^{(F)},X_{\mathbf d}^{(F)})$.  Each of these users can consequently
recover an unrequested file, contradicting~\eqref{eq:security}.
\end{Remark}

\begin{figure}[t]
\centering
\includegraphics[width=0.78\textwidth]{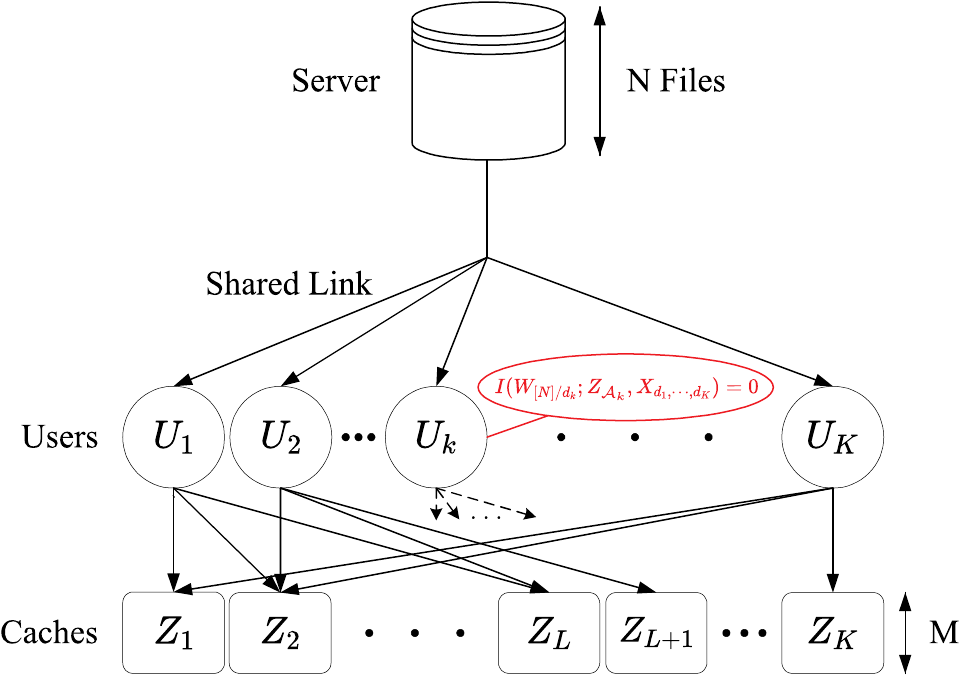}
\caption{The cyclic secure multi-access caching network.  Each user
accesses $L$ consecutive caches and receives the server's broadcast.}
\label{fig:system-model}
\end{figure}

\section{Parent-Preserving Secure Lifting}
\label{sec:secure-lifting}

We call the starting secure single-access scheme the \emph{parent}.
In that scheme, user \(u\) decodes its requested file from its cache
\(\widetilde Z_u\) and the multicast message \(X_{\mathbf d}\), with
user-level file secrecy.  Our lifting re-encodes the parent cache contents
across the physical caches.  User
\(u\)'s \(L\)-cache window recovers \(\widetilde Z_u\).
Conditioned on \(\widetilde Z_u\), the window reveals no information about
parent variables outside \(\widetilde Z_u\).  Each user can therefore run the parent
decoder, and the parent secrecy guarantee continues to hold.  Delivery uses
\(X_{\mathbf d}\) unchanged.

The construction groups parent symbols according to the users that store
them.  Symbols present in every parent cache form a common pool.  For every
other stored symbol, each maximal consecutive set of users storing it gives a
\emph{formal occurrence}.  Occurrences spanning the same number of users
are placed in \(K\)-position cyclic blocks, with one local encoder for each
block.  A deterministic Pascal-window encoder handles long spans.  A
randomized affine-coset encoder, built by complement-transpose duality,
handles short spans.  Together, they cover every proper appearance span
\(1\leq b\leq K-1\).

\subsection{Parent representation and local recovery and secrecy conditions}

Throughout this section, user, cache, and formal-position indices belong to
\(\mathbb Z_K\triangleq\{0,1,\ldots,K-1\}\), with additions to these indices
evaluated modulo \(K\).  Index-set complements are taken in \(\mathbb Z_K\).
Scalar-coordinate indices and polynomial exponents use ordinary integers.
The physical access window and the
length-\(b\) target window are
\begin{align}
\mathcal A_u
  &\triangleq\{u,u+1,\ldots,u+L-1\},
&
\mathcal T_u^{(b)}
  &\triangleq\{u,u+1,\ldots,u+b-1\}.
\end{align}
For a tuple \(V_{\mathbb Z_K}=(V_i:i\in\mathbb Z_K)\) and a set
\(\mathcal S\subseteq\mathbb Z_K\), we write
\(V_{\mathcal S}=(V_i:i\in\mathcal S)\).

For a named parent (sub)symbol \(Y\), its \emph{incidence set} is
\begin{align*}
\mathcal I(Y)
\triangleq
\{u\in\mathbb Z_K:Y\text{ is a named component of }\widetilde Z_u\}.
\end{align*}
A maximal nonempty consecutive component of \(\mathcal I(Y)\) on the cycle
is an \emph{appearance run}.  Its number of users is its \emph{appearance
span} \(b\).  Each run produces one \emph{formal occurrence}, placed at the
run's final user in cyclic order.  Occurrences of a fixed span \(b\) are packed by
endpoint into \emph{\(K\)-position formal blocks}
\(V_{\mathbb Z_K}=(V_0,\ldots,V_{K-1})\).
The parameter \(b\) specifies how many consecutive
users store each represented occurrence.  A parent symbol with several runs
has several formal occurrences, all referring to the same underlying random
variable.  We call such repeated references \emph{aliases}.

\begin{Def}[Admissible cyclic-incidence parent]
\label{def:admissible-parent}
A secure single-access parent is \emph{admissible} if its cache components have
a finite cyclic-incidence representation over a common construction field
\(\mathbb F_q\).
\begin{enumerate}
\item The parent has logical caches
\(\widetilde Z_0,\ldots,\widetilde Z_{K-1}\) and, for each public demand
vector \(\mathbf d\), a message \(X_{\mathbf d}\) satisfying the parent
correctness and user-level file-secrecy requirements.
\item The formal occurrences of all named parent (sub)symbols with nonempty
proper incidence sets form a finite list of \(K\)-position formal blocks.
Each block has an appearance span \(b\in[1:K-1]\), equal-length inputs,
and a named parent (sub)symbol assigned to every position.
\item Symbols common to all logical caches form a separate common-symbol
pool.
\item For every user \(u\), the target tuples
\(V_{\mathcal T_u^{(b)}}\) from all blocks, together with the common-symbol
pool, contain exactly the information in \(\widetilde Z_u\).
Fixed deduplication maps keep one copy of each named component to recover
\(\widetilde Z_u\).  Fixed expansion maps restore its copies at all target
references.  Both maps use the reference labels.
\end{enumerate}
\end{Def}

\begin{Def}[Valid local transformation]
Fix one \(K\)-position formal block of span \(b\).  A local transformation
maps \(V_{\mathbb Z_K}\) to a segment \(C_i\) in each physical cache \(i\).
A randomized transformation samples its seed at placement time,
independently of the block inputs and the demand vector.
The transformation is \emph{valid} if, for every user \(u\) and every joint
distribution of the formal inputs,
\begin{align}
H\!\left(V_{\mathcal T_u^{(b)}}\mid C_{\mathcal A_u}\right)&=0,
\label{eq:local-reconstruction}\\
I\!\left(
 V_{\mathbb Z_K\setminus\mathcal T_u^{(b)}};
 C_{\mathcal A_u}
 \middle|
 V_{\mathcal T_u^{(b)}}
\right)&=0.
\label{eq:local-conditional-secrecy}
\end{align}
The same encoder must satisfy both conditions for every input law, including
correlated inputs and aliases of one parent variable.
\end{Def}

\begin{Remark}[Distributional scope]
The lifting and fixed-segment allocation use encoders that are valid for every
joint input law.  The restricted local converse tests those encoders on
independent unit-entropy inputs.
\end{Remark}

\begin{Example}[Run-to-block packing]
Let \(K=4\), \(L=2\), and consider the subset-indexed parent with \(t=2\)
\cite{PrivateCodedCaching}.  A share named by \(\{0,1\}\) has one appearance
run of span \(2\) and one formal occurrence at endpoint \(1\).  A share named
by \(\{0,2\}\) has two singleton runs at endpoints \(0\) and \(2\).  The two
formal occurrences refer to the same parent share.  Each is grouped with
other occurrences of the same span into a four-position block.  For a span-two block, user \(0\) targets positions
\(\mathcal T_0^{(2)}=\{0,1\}\) and observes caches
\(\mathcal A_0=\{0,1\}\).  Equations~\eqref{eq:local-reconstruction} and
\eqref{eq:local-conditional-secrecy} require those two caches to recover
\((V_0,V_1)\) and, conditioned on that pair, to reveal nothing about
\((V_2,V_3)\).  Rotating the indices gives the same reconstruction and conditional-secrecy
requirements for the other three users.
\end{Example}

\subsection{Common scalarization and common symbols}

The direct and dual encoders require input vectors of lengths \(L\) and
\(K-L\), respectively.  To use both encoders with the same parent symbols,
express those symbols in coordinates over a common prime field.
Write \(Q=p^e\), where \(p\) is prime and \(e\) is a positive integer.
Fix a basis of \(\mathbb F_Q\) over its prime subfield \(\mathbb F_p\).
The direct encoder groups the basis coordinates into vectors of length
\(L\), and the dual encoder groups them into vectors of length \(K-L\).
Every reference to the same named symbol uses the same expansion.
We call this representation \emph{common scalarization}.

\begin{Remark}[Choice of construction field]
For an existence construction, choose a parent implementation over
\(\mathbb F_Q\), where \(Q=p^e\), \(p>K\), and
\(\operatorname{lcm}(L,K-L)\mid e\), satisfying that parent's field
requirements.  The lifting uses this implementation's symbols, reference
maps, and delivery message.  Local encoding and the common-symbol MDS layer
operate over its prime subfield.
\end{Remark}

\begin{Lem}[Common scalarization]
\label{lem:common-scalarization}
Fix \(1\leq L<K\), put
\begin{align}
\chi\triangleq\operatorname{lcm}(L,K-L),
\end{align}
and let an admissible parent be supplied over \(\mathbb F_Q\), where
\(Q=p^e\), \(p>K\), and \(\chi\mid e\).
Fix an \(\mathbb F_p\)-basis of \(\mathbb F_Q\), with coordinate bijection
\(\varphi:\mathbb F_Q\longrightarrow\mathbb F_p^e\).
If a named parent symbol \(Y\) has \(\ell_Y\) coordinates over
\(\mathbb F_Q\), its basis expansion \(\overline Y\) has length
\begin{align}
d_Y=e\ell_Y,
\qquad
L\mid d_Y,
\qquad
(K-L)\mid d_Y.
\label{eq:scalarized-length}
\end{align}
The original length \(\ell_Y\) is arbitrary.  Each reference to \(Y\)
can therefore be partitioned into exactly \(d_Y/L\) direct groups or
\(d_Y/(K-L)\) dual groups over \(\mathbb F_p\).
Both partitions use the same vector \(\overline Y\), and
\(\varphi^{-1}\) recovers the original parent symbols.

To make a common symbol recoverable from any \(L\) physical caches, split
\(\overline Y\) into \(L\) rows
\(\overline Y_1,\ldots,\overline Y_L\in\mathbb F_p^{d_Y/L}\).
Choose \(G^{\rm com}\in\mathbb F_p^{K\times L}\) whose every \(L\) rows
are nonsingular and give cache \(i\)
\begin{align}
C_i^{\rm com}(Y)
\triangleq\sum_{j=1}^{L}G^{\rm com}_{i,j}\overline Y_j.
\end{align}
Any \(L\) physical caches recover \(Y\).  A segment containing
\(a\) prime-field coordinates has alphabet \(\mathbb F_p^a\), so its storage is
\begin{align}
\log_Q\left|\mathbb F_p^a\right|=\frac{a}{e}
\end{align}
in parent-field units.  Thus each common-symbol share requires
\(\ell_Y/L\) units of storage.  Basis expansion preserves the original parent input and
broadcast lengths.

Component codes over a common parent field can be combined with any finite
set of rational memory-sharing weights.
\end{Lem}

\begin{IEEEproof}[Proof sketch]
Since \(L\mid e\) and \(K-L\mid e\), both divisibilities in
\eqref{eq:scalarized-length} hold for every \(\ell_Y\).
The basis map is a bijection, and its inverse reconstructs each parent
coordinate from its \(e\) prime-field coordinates.
The condition \(p>K\) supplies \(K\) distinct evaluation points for a
Vandermonde matrix, so every \(L\) shares of the common-symbol pool recover
its basis coordinates.  The identity
\(\log_Q(p^a)=a/e\) accounts for their exact storage.
The parent encoder and delivery message remain those of the supplied
\(\mathbb F_Q\) implementation.

A common denominator for the memory-sharing weights and a common multiple of
the component file lengths make each assigned component length an integer.
\end{IEEEproof}

\begin{Remark}[Fixed-field parent realization]
\label{rem:fixed-field-scalarization}
For a parent supplied over a fixed field \(\mathbb F_q\) of characteristic
greater than \(K\), independent parallel copies make the named symbol lengths
divisible by both local group sizes while preserving normalized memory and
rate.  The required finite repetition factor is determined by those lengths
and \(\operatorname{lcm}(L,K-L)\).
\end{Remark}

\subsection{Two complementary local encoders}

The two local encoders operate over \(\mathbb F_q\), with
\(p=\operatorname{char}(\mathbb F_q)>K\).  Under basis expansion,
\(q=p\).  The fixed-field implementation uses the supplied parent field.

The direct construction uses binomial-coefficient matrices, a family also
studied in connection with superregular matrices
\cite{StronglyMDSConvolutionalCodes,SuperregularMatricesConvolutionalCodes}.
Related coding approaches to secrecy include ramp secret sharing
\cite{RampSecretSharing} and the wiretap-II channel
\cite{WiretapChannelII}.

\begin{Theo}[Direct encoder, \(L\leq b\leq K-1\)]
\label{thm:direct-local-transform}
Represent every formal position by
\(\mathbf V_i=(v_{i,0},\ldots,v_{i,L-1})^{\mathsf T}\in\mathbb F_q^L\).
Put \(\omega=b-L\).  For each \(j\in[0:\omega]\), define a matrix
with row index \(m\in[0:b-1]\) and column index \(k\in[0:L-1]\) by
\begin{align}
G_j&\in\mathbb F_q^{b\times L},
&
(G_j)_{m,k}&=\binom{\omega}{j}\mathbf 1\{m=j+k\}.
\end{align}
Cache \(i\) stores
\begin{align}
\mathbf C_i
=\sum_{j=0}^{\omega}G_j\mathbf V_{i+j}
\in\mathbb F_q^b.
\label{eq:direct-local-encoder}
\end{align}
Define the global encoding matrix by stacking the inputs and outputs as
\(\mathbf C=\mathbf H\mathbf V\).  In a submatrix \(\mathbf H(\mathcal A,\mathcal T)\),
the index sets select cache-row blocks and input-column blocks, respectively.
Each cache contributes \(b\) rows, and each input position contributes
\(L\) columns.  For every user \(u\),
\begin{align}
\mathbf H(\mathcal A_u,
  \mathbb Z_K\setminus\mathcal T_u^{(b)})&=0,
\label{eq:direct-local-zero}\\
\mathbf H(\mathcal A_u,\mathcal T_u^{(b)})
&\in\mathbb F_q^{Lb\times Lb}
\quad\text{is nonsingular}.
\label{eq:direct-local-invertible}
\end{align}
The user therefore decodes by
\begin{align}
\mathbf V_{\mathcal T_u^{(b)}}
=\mathbf H(\mathcal A_u,\mathcal T_u^{(b)})^{-1}
 \mathbf C_{\mathcal A_u},
\label{eq:direct-local-decoder}
\end{align}
so the encoder is a deterministic valid local transformation for every
joint input law.  The encoder stores \(b\) scalars per cache for \(L\) source
scalars per formal position, so its normalized per-cache allocation is
\(b/L\).  The global matrix
\(\mathbf H\in\mathbb F_q^{Kb\times KL}\) has full column rank.
\end{Theo}

\begin{IEEEproof}[Proof Sketch]
Cache \(i\) depends only on positions \(i,\ldots,i+\omega\).  The union of
these supports over caches \(u,\ldots,u+L-1\) is exactly
\(\mathcal T_u^{(b)}\), proving \eqref{eq:direct-local-zero} and conditional
complement secrecy.
Rotate the user window to \(u=0\), with cache indices \(i\in[0:L-1]\)
and target indices \(n\in[0:b-1]\).
Write \(C_{i,m}\) for the \(m\)th scalar output of cache \(i\).
If \(v_{n,k}\) occurs in output coordinate \(C_{i,m}\),
then \(n=i+j\), \(m=j+k\), and hence
\begin{align}
n+k=i+m.
\end{align}
Group equations and unknowns by this integer antidiagonal index.
The local decoding matrix becomes block diagonal.  Each block has a positive
size \(r\leq L\) and is a Pascal-window matrix of the form
\begin{align}
B_{\omega,r,\delta}
=\left[\binom{\omega}{\delta+a-c}\right]_{a,c=0}^{r-1},
\qquad
0\leq\delta\leq\omega,
\quad r\leq L.
\end{align}
Each block is nonsingular when \(p>K\).  This gives
\eqref{eq:direct-local-invertible}.  A vector in the kernel of the global
matrix must therefore vanish on every target window.  These windows cover
all input positions, so the global matrix has full column rank.
\end{IEEEproof}

For \(b<L\), we hide the complementary inputs by choosing the cache contents
at random from the solutions of a linear system.  An auxiliary direct encoder
on complementary windows supplies the constraints.  Set
\begin{align}
L'&\triangleq K-L,
& b'&\triangleq K-b,
& \Delta&\triangleq L-b.
\end{align}
Since \(b'>L'\), the direct encoder applies to the auxiliary parameters
\((K,L',b')\).  Define the positive block shift by
\begin{align}
(P_\Delta\mathbf v)_{i+\Delta}=\mathbf v_i,
\qquad\text{equivalently}\qquad
(P_\Delta\mathbf v)_r=\mathbf v_{r-\Delta}.
\end{align}
For an index set \(\mathcal S\subseteq\mathbb Z_K\), write
\(\mathcal S\mathbin\oplus\Delta
\triangleq\{i+\Delta:i\in\mathcal S\}\).

\begin{Theo}[Complement-dual encoder, \(1\leq b<L\)]
\label{thm:short-local-transform}
Let \(\mathbf H\in\mathbb F_q^{Kb'\times KL'}\) be the global matrix of the
direct encoder with auxiliary parameters \((K,L',b')\), and set
\begin{align}
A&\triangleq\mathbf H^{\mathsf T}
  \in\mathbb F_q^{KL'\times Kb'},
&
d_{\rm ker}&\triangleq Kb'-KL'=K(L-b).
\end{align}
Since \(\mathbf H\) has full column rank, \(A\) has full row rank.
We construct a right inverse \(J\) to obtain one solution of the constraints
and a kernel basis \(B\) to randomize over all solutions.  Choose a
column-permutation matrix
\(P_{\rm col}\in\mathbb F_q^{Kb'\times Kb'}\) for which
\begin{align}
AP_{\rm col}=[A_{\rm P}\;A_{\rm F}],
\qquad
A_{\rm P}\in\mathbb F_q^{KL'\times KL'}
\quad\text{is nonsingular},
\qquad
A_{\rm F}\in\mathbb F_q^{KL'\times d_{\rm ker}},
\end{align}
and define
\begin{align}
J&\triangleq
P_{\rm col}\begin{pmatrix}A_{\rm P}^{-1}\\0\end{pmatrix},
&
B&\triangleq
P_{\rm col}\begin{pmatrix}-A_{\rm P}^{-1}A_{\rm F}\\I_{d_{\rm ker}}\end{pmatrix}.
\end{align}
Here
\(J\in\mathbb F_q^{Kb'\times KL'}\),
\(B\in\mathbb F_q^{Kb'\times d_{\rm ker}}\), and the block-permutation
matrix used in the encoder is
\(P_\Delta\in\mathbb F_q^{KL'\times KL'}\).
Thus \(AJ=I_{KL'}\), and the columns of \(B\) form a basis of \(\ker A\).

Represent each formal position by
\(\mathbf V_i\in\mathbb F_q^{L'}\), and stack these vectors into
\(\mathbf V\in\mathbb F_q^{KL'}\).  Stack cache segments
\(\mathbf C_i\in\mathbb F_q^{b'}\) into
\(\mathbf C\in\mathbb F_q^{Kb'}\).  For each block instance and scalar
group, draw a seed vector during placement:
\(\mathbf U\sim\operatorname{Unif}(\mathbb F_q^{d_{\rm ker}})\).
The seed vectors used by different block instances and scalar groups are
mutually independent.  Their collection is independent of the library, the
parent randomness, and the demand vector.  Place
\begin{align}
\mathbf C=JP_\Delta\mathbf V+B\mathbf U.
\label{eq:dual-affine-coset-encoder}
\end{align}
Every user \(u\) recovers its target tuple as
\begin{align}
\mathbf V_{\mathcal T_u^{(b)}}
=\mathbf H^{\mathsf T}
 (\mathcal T_u^{(b)}\mathbin\oplus\Delta,\mathcal A_u)
 \mathbf C_{\mathcal A_u}.
\label{eq:dual-local-decoder}
\end{align}
The accessible window also satisfies
\begin{align}
I\!\left(
 \mathbf V_{\mathbb Z_K\setminus\mathcal T_u^{(b)}};
 \mathbf C_{\mathcal A_u}
 \middle|
 \mathbf V_{\mathcal T_u^{(b)}}
\right)=0
\end{align}
for every joint input law.  Hence the encoder is valid, has normalized
per-cache allocation
\begin{align}
\frac{K-b}{K-L},
\end{align}
and uses exactly \(K(L-b)\) fresh field symbols per scalar group when the
affine solution is sampled uniformly.
\end{Theo}

\begin{IEEEproof}[Proof Sketch]
For \(a=u+L\), the auxiliary access window is
\(\mathcal A'_a=\mathcal A_u^{\rm c}\), and its target window is
\begin{align*}
\mathcal T'_a
=(\mathcal T_u^{(b)})^{\rm c}\mathbin\oplus\Delta,
\qquad
(\mathcal T'_a)^{\rm c}
=\mathcal T_u^{(b)}\mathbin\oplus\Delta.
\end{align*}
After transposing the auxiliary direct matrix and partitioning rows and
columns accordingly, the affine constraint
\(A\mathbf C=P_\Delta\mathbf V\) has the block form
\begin{align}
\begin{pmatrix}M_{11}&0\\M_{21}&M_{22}\end{pmatrix}
\begin{pmatrix}
\mathbf C_{\mathcal A_u}\\
\mathbf C_{\mathcal A_u^{\rm c}}
\end{pmatrix}
=
\begin{pmatrix}
\mathbf V_{\mathcal T_u^{(b)}}\\
\mathbf V_{(\mathcal T_u^{(b)})^{\rm c}}
\end{pmatrix},
\qquad M_{22}\text{ nonsingular}.
\end{align}
The first block row gives \eqref{eq:dual-local-decoder}.  For each source
value \(v\), \eqref{eq:dual-affine-coset-encoder} samples uniformly from
\(\{c:Ac=P_\Delta v\}\).  Projecting this solution set onto
\(c_{\mathcal A_u}\) gives an affine set determined by
\(v_{\mathcal T_u^{(b)}}\).  The first block row describes that set, and
invertibility of \(M_{22}\) gives each point a unique completion in the
complementary caches.  The accessible window is thus uniform on a set
determined by the target inputs.  Its conditional distribution given those
targets is independent of the complement.  Rank-nullity gives the stated
seed dimension.
\end{IEEEproof}

\begin{Example}[One-time-pad form of a short block]
Let \((K,L,b)=(4,3,2)\).  Then
\((L',b',\Delta)=(1,2,1)\).  Draw independent uniform placement-time pads
\(S_0,S_1,S_2,S_3\in\mathbb F_q\), and store, with indices modulo four,
\begin{align}
\mathbf C_i=(V_{i-1}-S_{i-1},\,S_i)^{\mathsf T}.
\end{align}
User \(u\) reads caches \(u,u+1,u+2\) and obtains
\begin{align*}
V_u&=S_u+(V_u-S_u),
&
V_{u+1}&=S_{u+1}+(V_{u+1}-S_{u+1}).
\end{align*}
The accessible cache expressions contain no \(V_{u+2}\).  The term
\(V_{u-1}-S_{u-1}\) is masked by the uniform pad
\(S_{u-1}\), which is stored in the inaccessible cache.  Given the two
decoded targets, the entire accessible window is therefore independent of
the two complementary symbols.
\end{Example}

Combining the two theorems, the normalized allocation of one formal block is
\begin{align}
\mu_{b,L}
\triangleq
\begin{cases}
\dfrac{K-b}{K-L},&1\leq b<L,\\[.8ex]
\dfrac{b}{L},&L\leq b\leq K-1.
\end{cases}
\end{align}
At \(b=L\), the direct encoder stores each position directly and attains
\(\mu_{L,L}=1\).

\subsection{Assembling a secure multi-access scheme}

Each physical cache stores its encoded proper-block segments and its share
of the common-symbol pool.  From an access window, each user recovers all
components of its parent cache.  Conditioned on those components, the
window reveals no further information about the parent variables.
The parent multicast message \(X_{\mathbf d}\) can therefore be reused
unchanged, preserving correctness, secrecy, and delivery rate.

\begin{Propo}[Local-to-global inheritance]
\label{prop:inherit}
Fix an admissible cyclic-incidence parent, a common scalarization from
Lemma~\ref{lem:common-scalarization} or
Remark~\ref{rem:fixed-field-scalarization}, and a valid local transformation for
every scalar group of every proper block.  Each randomized group uses its
own placement-time seed.  Collect all local seeds into one vector whose
components are mutually independent.  This vector is independent of the
library, the parent randomness, and the demand vector.
Encode the common-symbol pool by the deterministic
\((K,L)\) MDS layer.

Let \(\mathcal J\) index all scalarized proper-block groups.  Write
\(C_i^{(j)}\) for the segment contributed by group \(j\) to cache \(i\),
and \(C_i^{\rm com}\) for the concatenation of that cache's common-symbol
shares.  Store the segments in a fixed order with known boundaries:
\begin{align}
Z_i^{\rm lift}
\triangleq
\left((C_i^{(j)})_{j\in\mathcal J},C_i^{\rm com}\right).
\label{eq:lifted-cache-concatenation}
\end{align}
Then, for every user \(u\),
\begin{align}
H\!\left(\widetilde Z_u\mid
Z_{\mathcal A_u}^{\rm lift}\right)&=0.
\label{eq:effective-cache-reconstruction}
\end{align}
Let \(\Pi\) collect the files and all parent randomness, so that
\(X_{\mathbf d}\) is a deterministic function of \((\Pi,\mathbf d)\).
The variables form the Markov chain
\begin{align}
Z_{\mathcal A_u}^{\rm lift}
-\widetilde Z_u-\Pi.
\label{eq:lifting-global-markov}
\end{align}
Consequently, transmitting the parent's
\(X_{\mathbf d}\) unchanged gives, for every public demand vector
\(\mathbf d\),
\begin{align}
H\!\left(W_{d_u}\mid
X_{\mathbf d},Z_{\mathcal A_u}^{\rm lift},\mathbf d\right)&=0,
\\
I\!\left(W_{[N]\setminus\{d_u\}};
X_{\mathbf d},Z_{\mathcal A_u}^{\rm lift}
\middle|W_{d_u},\mathbf d\right)&=0.
\label{eq:lifting-global-security}
\end{align}
The normalized delivery rate equals the parent's rate.  These conclusions
hold for arbitrary joint laws of the parent symbols and block inputs,
including repeated references to the same parent variable.
\end{Propo}

\begin{IEEEproof}[Proof Sketch]
Fix \(u\).  Let \(\Theta_u\) be the ordered tuple of all decoded target
references together with the common-symbol pool.  Local reconstruction and
the MDS property recover \(\Theta_u\) from
\(Z_{\mathcal A_u}^{\rm lift}\).  Regrouping the decoded coordinates,
applying the inverse basis map, and deduplicating named references recover
\(\widetilde Z_u\).  Basis expansion and the supplied reference map recover
\(\Theta_u\) from the named components of \(\widetilde Z_u\).

For each proper-block group \(j\) with appearance span \(b_j\),
\eqref{eq:local-conditional-secrecy} holds for every input law.  Thus any two
input vectors with the same targets induce the same distribution on the
accessible cache segments:
\begin{align*}
P_{C_{\mathcal A_u}^{(j)}\mid V_{\mathbb Z_K}^{(j)}}
=P_{C_{\mathcal A_u}^{(j)}\mid
V_{\mathcal T_u^{(b_j)}}^{(j)}}.
\end{align*}
Condition on \(\Pi\), which fixes all formal inputs.  The local seeds remain
mutually independent because their collection is independent of \(\Pi\).
The conditional distribution of the proper-block segments therefore factors
across groups.  Each factor depends on the parent variables only through
that group's target inputs.  The accessible common-symbol shares are
deterministic functions of components of \(\widetilde Z_u\).
Hence the distribution of the entire physical window depends on \(\Pi\)
only through \(\Theta_u\), and therefore through \(\widetilde Z_u\).
This is the Markov property \eqref{eq:lifting-global-markov}, valid for
arbitrary joint laws of the parent source vectors.

Each user first reconstructs its parent cache by
\eqref{eq:effective-cache-reconstruction} and then applies the parent decoder.
This proves correctness.  Parent secrecy and the Markov chain give
\begin{align*}
&I\!\left(W_{[N]\setminus\{d_u\}};
X_{\mathbf d},\widetilde Z_u,Z_{\mathcal A_u}^{\rm lift}
\middle|W_{d_u},\mathbf d\right)\\
&\quad=
I\!\left(W_{[N]\setminus\{d_u\}};
X_{\mathbf d},\widetilde Z_u
\middle|W_{d_u},\mathbf d\right)
+I\!\left(W_{[N]\setminus\{d_u\}};
Z_{\mathcal A_u}^{\rm lift}
\middle|X_{\mathbf d},\widetilde Z_u,W_{d_u},\mathbf d\right)
=0.
\end{align*}
Because \(\widetilde Z_u\) is itself reconstructed from the lifted window,
removing it from the mutual-information expression leaves
\eqref{eq:lifting-global-security}.
Basis expansion preserves the original file and broadcast lengths.
The fixed-field realization scales both lengths by its repetition factor.
Both realizations retain the parent delivery rate.
\end{IEEEproof}

\begin{Remark}[Cache segment accounting]
In Proposition~\ref{prop:inherit}, the cache segments in
\eqref{eq:lifted-cache-concatenation} have fixed lengths and can be parsed
losslessly.  Their lengths add even when the encoded inputs are correlated
or refer repeatedly to the same parent variable.  These correlations can
reduce the entropy of the concatenated cache.
\end{Remark}

A scalarized block with \(d\) coordinates per formal position uses
\(d/L\) direct encoding calls when \(b\ge L\), and \(d/(K-L)\) dual calls
when \(b<L\).  Here \(d\) counts coordinates over the local coding field.
Each dual call uses \(K(L-b)\) fresh field symbols, giving
\(dK(L-b)/(K-L)\) random symbols per block.  These counts give the
implementation bounds in the next section.

Appendix~\ref{app:inheritance} gives the complete scalarization and inheritance
proofs.  Appendix~\ref{app:local-encoders} proves the two local encoder theorems.

\section{Memory-Rate Results and Bounds}
\label{sec:memory-rate-results}

The lifting in Section~\ref{sec:secure-lifting} gives memory-rate points for
secure multi-access coded caching (SMACC) by re-encoding secure single-access
caches and keeping each parent's
broadcast unchanged.  The subset-indexed parent gives a family of explicit
points, and a complementary parent gives a rate-one point.
Converse bounds determine how closely these points approach the optimum
over all SMACC schemes.  Together, the achievable points and converse bounds
give the exact memory-rate tradeoff at \(L=K-1\).  Known \(L=1\) results
provide additional achievable points.

\subsection{Lifted Achievable Family and Endpoints}

For an appearance span $b\in[1:K-1]$, define the normalized per-cache cost of
one proper $K$-position cyclic block by
\begin{align}
\mu_{b,L}
\triangleq
\begin{cases}
\dfrac{K-b}{K-L}, & 1\le b<L,\\[1.2ex]
\dfrac{b}{L}, & L\le b\le K-1.
\end{cases}
\label{eq:mu}
\end{align}
The randomized complement-transpose encoder attains the short-span branch,
and the deterministic direct encoder attains the long-span branch.  Both
cost formulas equal one at $b=L$.

For the subset-indexed parent of \cite{PrivateCodedCaching}, counting maximal
appearance runs gives the number of blocks and hence the required storage.
For $r\in[0:K-1]$, let $\Gamma(K,r,b)$ count the $r$-subsets with a maximal
span-$b$ run at a fixed cyclic start position.  This count is
\begin{align}
\Gamma(K,r,b)
\triangleq
\begin{cases}
\displaystyle\binom{K-b-2}{r-b},
   & 1\le b\le r\le K-2,\\[1.2ex]
1, & r=b=K-1,\\
0, & \text{otherwise}.
\end{cases}
\label{eq:gamma}
\end{align}

\begin{Lem}[Parent-profile cyclic decomposition]
\label{lem:cycliccount}
Consider a parent-symbol family
$\{Y_{\mathcal S}:\mathcal S\subseteq\mathbb Z_K,
|\mathcal S|=r\}$ in which $Y_{\mathcal S}$ belongs to parent cache
$\widetilde Z_u$ exactly when $u\in\mathcal S$.  For every fixed cyclic start
position, the number of $r$-subsets having a maximal appearance run of span
$b$ at that start is $\Gamma(K,r,b)$ in \eqref{eq:gamma}.  Placing each run at
its endpoint and collecting equal-ranked occurrences over the $K$ endpoints
produces exactly $\Gamma(K,r,b)$ complete $K$-position formal blocks.
Across all spans, user $u$'s target positions contain exactly the members
$Y_{\mathcal S}$ with $u\in\mathcal S$ from this family.  Repeated formal
occurrences refer to the same parent random variable.
\end{Lem}

\begin{IEEEproof}[Proof sketch]
For $1\le b\le r\le K-2$, a prescribed maximal run fixes its $b$ members and excludes
the two adjacent positions.  The remaining $r-b$ members can be chosen in
$\binom{K-b-2}{r-b}$ ways.  For $r=K-1$, only the complement of the position
immediately before the run gives the boundary case $b=K-1$.
For every nonempty proper incidence subset $\mathcal S$ and every
$u\in\mathcal S$, user $u$ belongs to exactly one maximal run of
$\mathcal S$.  Placing that run at its endpoint puts the corresponding symbol
in that user's span-$b$ target window.
Collecting these targets and removing repeated references recovers the
family members stored in each parent cache.
\end{IEEEproof}

By Lemma~\ref{lem:cycliccount}, the total per-cache allocation for an
$r$-indexed parent-symbol family is
\begin{align}
T_{\mathrm{blk}}(K,r,L)
&\triangleq
\sum_{b=1}^{K-1}\mu_{b,L}\Gamma(K,r,b)\nonumber\\
&=
\sum_{b=1}^{L-1}\frac{K-b}{K-L}\Gamma(K,r,b)
+\sum_{b=L}^{K-1}\frac{b}{L}\Gamma(K,r,b).
\label{eq:T}
\end{align}

\begin{Theo}[Lifted subset-indexed family]
\label{thm:memoryload}
Fix $N,K\ge2$ and $1\le L<K$.  For every
$t\in\{0,1,\ldots,K-2\}$, there is a securely achievable SMACC point
\begin{align}
R(t)&=\frac{K}{t+1},
\label{eq:R-t}\\
M(t)&=
\frac{N T_{\mathrm{blk}}(K,t,L)
      +T_{\mathrm{blk}}(K,t+1,L)}{\binom{K-1}{t}}.
\label{eq:M-t}
\end{align}
The construction re-encodes the caches of the subset-indexed secure
single-access parent at parameter $t$ and retains its broadcast.  A finite
scalarization over a field of characteristic greater than $K$ realizes the
point.
\end{Theo}

\begin{IEEEproof}[Proof sketch]
The parent contains $N$ share families indexed by $t$-subsets and one masking
key family indexed by $(t+1)$-subsets.  Each share and masking key has
normalized length $1/\binom{K-1}{t}$.  Decomposing each incidence subset into maximal cyclic runs
gives the multiplicities in \eqref{eq:gamma}.
For each user, collecting the target references from all $N$ share families
and the key family gives exactly its parent cache.
Each complete formal block of
span $b\ge L$ uses the direct encoder, and each block of span
$b<L$ uses the randomized complement-transpose encoder.  The local-to-global
inheritance result in Proposition~\ref{prop:inherit} ensures that each physical access window recovers
its parent cache and reveals no information beyond it.  Each user can
therefore recover its requested file using the unchanged parent broadcast.
Summing the local costs gives \eqref{eq:M-t}.  Keeping the parent broadcast
preserves the rate in \eqref{eq:R-t}.
\end{IEEEproof}

At $t=0$, the file-share contribution is absent and the masking-key family has
one singleton-run block.  Therefore Theorem~\ref{thm:memoryload} includes the
leftmost point of its indexed family
\begin{align}
(M,R)=\left(\frac{K-1}{K-L},K\right).
\label{eq:left-endpoint}
\end{align}

Increasing $t$ lowers the parent rate $K/(t+1)$ and changes the number and
spans of the appearance runs.  The two sums in \eqref{eq:T} give the storage
required for the short and long runs.  Memory sharing forms the lower convex
envelope of these points.  The supporting points determine which $t$-values
contribute to the achievable boundary.

\begin{Example}[A corner using both local encoders]
Take $(K,L,t)=(4,2,2)$.  For the file-share family,
\begin{align*}
\Gamma(4,2,1)=1,
\qquad
\Gamma(4,2,2)=1.
\end{align*}
The span-one run is short and costs
$\mu_{1,2}=(4-1)/(4-2)=3/2$.  The span-two run is long and costs
$\mu_{2,2}=1$.  Hence
\begin{align*}
T_{\mathrm{blk}}(4,2,2)=\frac52.
\end{align*}
The masking-key family has incidence size $t+1=3=K-1$.  Its only nonzero run
count is $\Gamma(4,3,3)=1$, with direct cost $\mu_{3,2}=3/2$.  Substituting in
\eqref{eq:M-t} and \eqref{eq:R-t} gives
\begin{align}
(M(2),R(2))
=\left(\frac{5N}{6}+\frac12,\frac43\right).
\end{align}
This point uses both local encoders.  Its masking-key family contributes the
$r=K-1$ case of the same run-count formula.
\end{Example}

A complementary secure single-access parent gives the lifted
family a rate-one endpoint.

\begin{Propo}[Achievable rate-one endpoint]
\label{prop:rate-one-achievability}
For every $N,K\ge2$ and $1\le L<K$, the point
\begin{align}
(M,R)=\left(\frac{(N-1)(K-1)}{L},1\right)
\label{eq:endpoint}
\end{align}
is securely achievable.
\end{Propo}

\begin{IEEEproof}[Proof sketch]
The unit-rate secure single-access construction of \cite{FangSmallRate} has
parent-cache size $(N-1)(K-1)$.  Applying the direct encoder to its cyclic
representation lets each user reconstruct its parent cache from $L$ physical
caches.  Each physical cache stores a $1/L$ share of the parent-cache
allocation, and the unchanged parent broadcast has rate one.
\end{IEEEproof}

The achievable envelope $R_{\mathrm{lift}}(M)$ is the memory-sharing closure
of the $t$-indexed points and the rate-one endpoint.  Let
$\lambda_0,\ldots,\lambda_{K-2},\lambda_\star$ be nonnegative coefficients
with sum one, where $\lambda_\star$ multiplies the endpoint
\eqref{eq:endpoint}.  Define
\begin{align}
R_{\mathrm{lift}}(M)
\triangleq
\inf\Bigg\{&
\sum_{t=0}^{K-2}\lambda_tR(t)+\lambda_\star:
\sum_{t=0}^{K-2}\lambda_tM(t)
+\lambda_\star\frac{(N-1)(K-1)}{L}\le M,
\nonumber\\[-.3ex]
&\lambda_0,\ldots,\lambda_{K-2},\lambda_\star\ge0,
\quad
\sum_{t=0}^{K-2}\lambda_t+\lambda_\star=1
\Bigg\},
\label{eq:ach-envelope-def}
\end{align}
with the infimum of the empty set defined as $+\infty$.  To achieve
\eqref{eq:ach-envelope-def}, split every file into independent parts of the
prescribed proportions, use independent placement randomness for each part,
and concatenate the component cache segments and broadcasts.
Rational coefficients are realized at a common finite blocklength, and real
coefficients are approached by rational sequences under the asymptotic
definition of achievability.  Memory sharing preserves exact
recovery and user-level file secrecy.
Each finite-rate construction remains achievable at every larger cache
budget.
Consequently,
\begin{align}
R^*(M)\le R_{\mathrm{lift}}(M)
\end{align}
whenever the right-hand side is finite.

The points in Theorem~\ref{thm:memoryload} and the endpoint
\eqref{eq:endpoint} determine $R_{\mathrm{lift}}$.  At $L=1$,
improved secure single-access points provide additional achievable rates.

\begin{Remark}[Complete-block multiplicity]
For a pure $t$-point and each span $b$, the number of complete formal block
instances is
\begin{align}
A_t(b)=N\Gamma(K,t,b)+\Gamma(K,t+1,b).
\label{eq:pure-point-block-multiplicity}
\end{align}
These block counts also determine the finite subpacketization,
fresh-randomness, and arithmetic bounds.
\end{Remark}

To quantify the work added by the lifting layer, we count encoder calls,
stored coordinates, fresh random symbols, online arithmetic, and matrix
precomputation.  The costs of parent
share generation and delivery remain those of the chosen secure
single-access construction.

\begin{Propo}[Implementation costs]
\label{prop:implementation-cost}
Fix a pure point $t\in\{0,\ldots,K-2\}$ and define
\begin{align}
C_t\triangleq\binom{K-1}{t},
\qquad
\chi\triangleq\operatorname{lcm}(L,K-L).
\end{align}
Choose a parent implementation over $\mathbb F_Q$ that meets the parent's
field requirements.  Here $Q=p^e$, where $p>K$ is prime and $e$ is a positive
integer divisible by $\chi$.
For any positive integer $m$ meeting the parent's blocklength requirements,
put
\begin{align}
F=mC_t,
\qquad
d\triangleq em.
\end{align}
Each file has $F$ parent-field symbols, and every parent share or masking key
has $m$ parent-field symbols.  Their fixed public $\mathbb F_p$-basis
representation gives $d$ prime-field coordinates per share or key, with
$dC_t=eF$.  All encoder and arithmetic counts use these
coordinates over $\mathbb F_p$.  With
$A_t(b)$ defined in \eqref{eq:pure-point-block-multiplicity}, the numbers of
direct and dual vector-encoder calls are, respectively,
\begin{align}
G_{\mathrm{dir}}(t;d)
&=\frac{d}{L}\sum_{b=L}^{K-1}A_t(b),
\\
G_{\mathrm{dual}}(t;d)
&=\frac{d}{K-L}\sum_{b=1}^{L-1}A_t(b).
\end{align}
The number of prime-field coordinates stored by the lifting layer in each physical
cache is
\begin{align}
P_{\mathrm{cache}}(t;d)
&=\frac{d}{K-L}\sum_{b=1}^{L-1}(K-b)A_t(b)
+\frac{d}{L}\sum_{b=L}^{K-1}bA_t(b)\nonumber\\
&=d\bigl[N T_{\mathrm{blk}}(K,t,L)
        +T_{\mathrm{blk}}(K,t+1,L)\bigr],
\label{eq:cache-scalar-coordinate-count}
\end{align}
The prescribed cache alphabet is
$\mathbb F_p^{P_{\mathrm{cache}}(t;d)}$, so
\begin{align}
\log_Q|\mathcal Z_i|=\frac{P_{\mathrm{cache}}(t;d)}{e},
\qquad
\frac{P_{\mathrm{cache}}(t;d)}{eF}=M(t).
\end{align}
Under the uniform affine-coset short-block realization, the number of fresh
independent $\mathbb F_p$ symbols sampled by the local encoders is
\begin{align}
Q_{\mathrm{loc}}(t;d)
=\frac{dK}{K-L}
\sum_{b=1}^{L-1}(L-b)A_t(b).
\label{eq:total-local-randomness}
\end{align}

Sparse evaluation of one direct group of span $b\ge L$
uses at most $KL(b-L+1)$ multiplications and
$K[L(b-L+1)-b]$ additions.  For short groups, set
$m_0=K(K-L)$ and $n_b=K(K-b)$.  A right inverse and kernel basis are precomputed
once for each active $(K,L,b)$ and reused across every block with those
parameters over the chosen prime field.  The setup uses
$O(m_0^2n_b)$ operations over $\mathbb F_p$.  The stored right inverse and kernel
basis together contain $n_b^2$ elements of $\mathbb F_p$.
After setup, dense matrix-vector multiplication for
one short group uses at most $n_b^2$ multiplications and $n_b(n_b-1)$ additions.
\end{Propo}

\begin{IEEEproof}[Proof sketch]
Since $L\mid e$ and $K-L\mid e$, the $d=em$ coordinates of each share or
masking key form $d/L$ direct-encoder groups or $d/(K-L)$ complement-dual
groups for every admissible parent length $m$.
Lemma~\ref{lem:cycliccount} gives $A_t(b)$ complete blocks.  Multiplying the
long-group count by $b$ and the short-group count by $K-b$ gives the
stored-coordinate count in \eqref{eq:cache-scalar-coordinate-count}.  Division
by $eF=dC_t$ recovers \eqref{eq:M-t}.  A span-$b$ short group solves a
full-row-rank affine system with $K(K-b)$ variables and $K(K-L)$ equations, whose
nullity is $K(L-b)$.  Multiplying this nullity by the short-group multiplicities
gives \eqref{eq:total-local-randomness}.  Counting the nonzero coefficients in
the direct matrices and applying standard dense matrix-vector multiplication
and Gauss-Jordan elimination over $\mathbb F_p$ gives the arithmetic bounds.
\end{IEEEproof}

\subsection{Exact Blockwise Allocation for Parent-Preserving Lifting}

We now determine the minimum fixed-length allocation when each supplied
parent block is encoded separately and the parent broadcast is retained.
Definition~\ref{def:block-separable-parent-preserving} specifies this class.
For the supplied subset-indexed parent, Theorem~\ref{thm:framework} shows that
the two encoders attain this minimum at every input length.  The lower
bound allows arbitrary segment alphabets.

\enlargethispage{\baselineskip}
\begin{Def}[Block-separable parent-preserving transformation]
\label{def:block-separable-parent-preserving}
Fix a secure single-access parent and a supplied decomposition of its proper
incidence symbols into formal cyclic blocks.  A transformed representation is
\emph{block separable and parent preserving} if it satisfies these conditions.
\begin{enumerate}[label=(\roman*),leftmargin=2.1em]
\item The parent delivery message is unchanged.
\item Every formal block is encoded into its own fixed-length segment in each
physical cache.  Each segment is a function of that formal block and a
block-local placement-time seed.  The seeds are independent of the parent
variables and mutually independent across formal blocks.  The per-block
encoders can be nonlinear, randomized, and asymmetric across caches.
\item For every joint distribution of the formal positions, every block and
every user satisfy the
local reconstruction and conditional complement-secrecy conditions
\eqref{eq:local-reconstruction} through \eqref{eq:local-conditional-secrecy}.
\end{enumerate}
Let $\beta$ index the separately stored factors of each physical cache,
including the common-symbol recovery layer as one additional factor when
present.  Let $\mathcal C_i^{(\beta)}$ be the segment alphabet allocated to
factor $\beta$ in physical cache $i$.  The cache value $Z_i$ is the tuple of
these segment values.  For file length $F$, define the per-factor allocation,
the largest total allocation across caches, and the largest cache entropy by
\begin{equation}
m_i^{(\beta)}
\triangleq \frac{1}{F}\log_q|\mathcal C_i^{(\beta)}|.
\end{equation}
\begin{equation}
\begin{aligned}
M_{\mathrm{alloc}}
&\triangleq\max_{i\in[K]}\sum_\beta m_i^{(\beta)},
\\
M_{\mathrm{ent}}
&\triangleq\max_{i\in[K]}\frac{H(Z_i)}{F}.
\end{aligned}
\end{equation}
\end{Def}

\begin{Remark}[Fixed-length allocation objective]
Theorem~\ref{thm:framework} optimizes the fixed-length segment resource
$M_{\mathrm{alloc}}$.  The induced cache entropy always satisfies
$M_{\mathrm{ent}}\le M_{\mathrm{alloc}}$, with strict inequality possible
when the parent contains repeated or correlated formal references.
\end{Remark}

Definition~\ref{def:block-separable-parent-preserving} requires local
reconstruction and conditional complement secrecy for every joint
distribution of the formal positions.  Applying these requirements to
independent, equal-entropy inputs gives the blockwise lower bound in
Theorem~\ref{thm:framework}.

\begin{Theo}[Exact minimum universal blockwise allocation]
\label{thm:framework}
For every $K\ge2$ and $1\le L<K$, there exists a finite field
$\mathbb F_Q$ such that the following holds simultaneously for all proper
appearance spans $b\in[1:K-1]$ and all positive integers $h$.
Consider a representation in
Definition~\ref{def:block-separable-parent-preserving}.  Evaluate it on
mutually independent formal inputs $V_0,\ldots,V_{K-1}$, each uniform over
$\mathbb F_Q^h$.  For this block, let $\mathcal C_i$ be the prescribed segment
alphabet in cache $i$.  Define the normalized allocation in cache $i$ and its
average across caches by
\begin{align*}
a_i\triangleq\frac1h\log_Q|\mathcal C_i|,
\qquad
\bar a_b\triangleq\frac1K\sum_{i=0}^{K-1}a_i.
\end{align*}
Then
\begin{align}
\bar a_b\ge\mu_{b,L}
=
\begin{cases}
\dfrac{K-b}{K-L},&b<L,\\[1.2ex]
\dfrac bL,&b\ge L.
\end{cases}
\end{align}
The direct and complement-dual encoders attain
$a_i=\mu_{b,L}$ in every cache for every joint input law.
Their fixed-length segments satisfy
\begin{align}
|\mathcal C_i|=Q^{h\mu_{b,L}},
\qquad
\frac1h\log_Q|\mathcal C_i|=\mu_{b,L}.
\label{eq:exact-basis-allocation}
\end{align}
Thus $\mu_{b,L}$ is the exact minimum worst-cache universal allocation.
The lower bound also holds over every other finite field.

For the subset-indexed parent realized over this construction field,
summing the blockwise minima gives the exact minimum normalized
worst-cache allocation $M(t)$ in \eqref{eq:M-t} for its supplied block list.
The lifting preserves the parent delivery message and rate.
\end{Theo}

\begin{IEEEproof}[Proof sketch]
For $b\ge L$, the $L$ segments accessible to each user reconstruct $b$
independent inputs and have joint entropy at least $bh$.  Summing over the
$K$ cyclic windows counts each cache $L$ times and gives
$\bar a_b\ge b/L$.

For $b<L$, independence of the formal inputs and conditional local secrecy
make the accessible window independent of the $K-b$ complementary inputs.
All caches together recover every input.  The $K-L$ inaccessible caches
therefore convey $(K-b)h$ units of complementary information conditioned on
the accessible window.  Summing over users counts each cache $K-L$ times
and gives $\bar a_b\ge(K-b)/(K-L)$.

Choose a prime $p>K$, put $\chi=\operatorname{lcm}(L,K-L)$, and take
$Q=p^\chi$.  Fix an $\mathbb F_p$-basis of $\mathbb F_Q$ and expand each input vector
into $\chi h$ base-field coordinates.  The direct encoder processes these
coordinates in groups of $L$, and the complement-dual encoder processes
them in groups of $K-L$.  They store $\chi h\mu_{b,L}$ base-field
coordinates per cache.  The inverse basis map recovers the original target
symbols, and the local channel factorization preserves conditional secrecy.
The base-$Q$ log-cardinality of the segment alphabet is exactly
$h\mu_{b,L}$, since $\log_Q p=1/\chi$.  The construction therefore attains
the lower bound at the original input length $h$.
The same basis serves every block, so separate-segment additivity gives
$M(t)$.
\end{IEEEproof}

\begin{Remark}[Subfield-coordinate storage]
A segment in $\mathbb F_p^{\,\chi h\mu_{b,L}}$ has
$p^{\chi h\mu_{b,L}}$ possible values and hence fixed-length allocation
$h\mu_{b,L}$ in base-$Q$ units.  The basis map is a bijective representation
of the original $h$ symbols of $\mathbb F_Q$.
\end{Remark}

\begin{Example}[Why the two allocation branches differ]
Let $(K,L)=(5,2)$.  For a span-three block, every two-cache access window must
recover three independent input symbols.  Summing this requirement over the
five cyclic users counts each cache twice and forces average allocation at
least $3/2=\mu_{3,2}$.  For a span-one block, the accessible window must be
independent of four complementary symbols.
Adding the three inaccessible caches makes those four symbols recoverable.
Averaging the inaccessible-cache requirement gives $4/3=\mu_{1,2}$.
The inaccessible caches determine the short-span bound, and the accessible
caches determine the long-span bound in $\mu_{b,L}$.  The two formulas
coincide at $b=L$.
\end{Example}

Appendix~\ref{app:allocation} gives the complete proofs of the run counts,
achievable family, implementation bounds, and allocation theorem.

\subsection{Unrestricted Converses and Approximation Guarantees}

The optimum $R^*(M)$ ranges over every placement, delivery rule, and finite
alphabet allowed by the system model.  Treating each access window as a
single effective cache transfers secure single-access converses to SMACC
and gives the exact rate-one threshold.

\begin{Propo}[Effective-cache reduction to secure single access]
\label{prop:effective-cache-reduction}
Let $R^*_{\mathrm{SA},K,N}(m)$ denote the optimum rate for secure
single-access coded caching with $K$ users, $N$ files, and per-user cache size
$m$.  For every $N,K\ge2$, $1\le L<K$, and $M\ge0$,
\begin{align}
R^*_{K,L,N}(M)
\ge R^*_{\mathrm{SA},K,N}(LM).
\label{eq:effective-cache-general-reduction}
\end{align}
\end{Propo}

\begin{IEEEproof}
Given any length-$F$ SMACC code with normalized cache memory $M_F$, assign
single-access user $k$ the effective cache alphabet
$\widetilde{\mathcal Z}_k=\prod_{i\in\mathcal A_k}\mathcal Z_i$ and the
effective cache value $\widetilde Z_k=Z_{\mathcal A_k}$.  In this proof,
$\widetilde Z_k$ denotes the full SMACC access window.  The fixed-length
constraint gives
\begin{align*}
\log_q|\widetilde{\mathcal Z}_k|
\le\sum_{i\in\mathcal A_k}\log_q|\mathcal Z_i|
\le LM_FF.
\end{align*}
The original broadcast, decoder, correctness condition, and conditional file
secrecy condition are unchanged.  Thus every SMACC code induces a secure
single-access code with memory at most $LM_F$ and the same rate.
Passing to code sequences gives \eqref{eq:effective-cache-general-reduction}.
\end{IEEEproof}

\begin{Cor}[Exact rate-one cache threshold]
\label{cor:exact-unit-rate}
For every $N,K\ge2$ and $1\le L<K$, the minimum cache size at which rate at
most one is securely achievable is
\begin{align}
M_{R=1}^{\star}=\frac{(N-1)(K-1)}{L}.
\end{align}
\end{Cor}

\begin{IEEEproof}
Proposition~\ref{prop:rate-one-achievability} gives achievability.  Conversely,
Proposition~\ref{prop:effective-cache-reduction} maps any sequence achieving
rate one to secure single access with normalized cache budget at most $LM$.
The rate-one entropy converse of \cite{FangSmallRate}, under the
code-sequence normalization, gives $LM\ge(N-1)(K-1)$.
\end{IEEEproof}

The following cyclic-window converse was established by Namboodiri and Rajan
\cite[Theorem~6]{ImprovedLowerBoundsForMultiAccessCodedCaching}.  We state the
bound in our notation and give a proof for the present system model.

\begin{Theo}[Secrecy-aware converse for unrestricted SMACC schemes
{\cite[Theorem~6]{ImprovedLowerBoundsForMultiAccessCodedCaching}}]
\label{thm:macc-converse}
For every $M\ge0$, define
\begin{align}
R_{\mathrm{conv}}(M)
\triangleq
\max_{s\in\{1,\ldots,\min\{\lfloor N/2\rfloor,K\}\}}
\frac{s\lfloor N/s\rfloor-1
-\min\{s-1,K-L\}M}{\lfloor N/s\rfloor-1}.
\label{eq:R-conv}
\end{align}
At every feasible cache size $M$, the optimum over all codes allowed by the
system model satisfies
\begin{align}
R^*_{K,L,N}(M)\ge R_{\mathrm{conv}}(M).
\label{eq:macc-th6}
\end{align}
The bound applies to both $N\ge K$ and $K>N$, including repeated demands.
\end{Theo}

\begin{IEEEproof}[Proof sketch]
Fix $s$ consecutive users.  Choose $\lfloor N/s\rfloor$ public demand
vectors so that the files requested by these users form disjoint groups
of size $s$.  Relative to the first
user's $L$-cache window, the union of the $s$ access windows introduces at most
$\min\{s-1,K-L\}$ additional physical caches.
Conditioned on the first user's requested file, that user's observation
under the first demand vector contains no information about the other files
requested by these users across the chosen demand vectors.  The additional caches and
remaining broadcasts make those files recoverable.
Bounding their fixed-length alphabets gives the $s$-indexed term in
\eqref{eq:R-conv}.  Maximizing over $s$ gives the stated bound.  The same
bound holds when $K>N$: the remaining users' requests are arbitrary, so the
chosen demand vectors can contain repeated entries.
\end{IEEEproof}

The following lemma extends the necessary cache-memory bound in
\cite[Proposition~2]{ImprovedLowerBoundsForMultiAccessCodedCaching} to arbitrary
$N,K\ge2$.

\begin{Lem}[Necessary feasibility threshold]
\label{lem:feasible-memory}
Every feasible cache size satisfies
\begin{align}
M\ge M_{\mathrm{feas}}
\triangleq\frac{\min\{N,K\}-1}{K-L}.
\label{eq:feasible-memory}
\end{align}
\end{Lem}

\begin{IEEEproof}[Proof sketch]
Let $r=\min\{N,K\}$ and choose a demand in which $r$ users request distinct
files.  Conditioned on its requested file, the first user's access window and
the broadcast contain no information about the other $r-1$ files.
The full cache collection and the same broadcast recover all $r$ requested
files.  Hence the $K-L$ caches outside the first user's window must supply the
required $(r-1)F$ units of additional information.  The base-$q$
log-cardinality of their joint cache alphabet is at most $(K-L)MF$.
\end{IEEEproof}

Applying the secure single-access cut-set bound of
\cite[Theorem~2]{PrivateCodedCaching} through
Proposition~\ref{prop:effective-cache-reduction} gives the $s$th lower-bound
term
\begin{align*}
\frac{s\lfloor N/s\rfloor-1-L(s-1)M}{\lfloor N/s\rfloor-1}.
\end{align*}
Feasibility of the effective single-access code gives $LM\ge1$.
The cited cut-set is stated up to $LM=N(K-1)$.  For larger effective cache
sizes, each displayed term is at most one and is covered by
$R_{\mathrm{conv}}(M)\ge1$.
The cyclic coefficient satisfies
\begin{align*}
\min\{s-1,K-L\}\le L(s-1),
\end{align*}
so every cyclic $s$-term in \eqref{eq:R-conv} is at least as large as the
corresponding single-access term.  We use $R_{\mathrm{conv}}(M)$ to bound the
achievable rate relative to the optimum.  The effective-cache reduction also
gives the exact rate-one threshold in Corollary~\ref{cor:exact-unit-rate}.

Where the lifted envelope is finite, the unrestricted optimum satisfies
\begin{align}
R_{\mathrm{conv}}(M)
\le R^*(M)
\le R_{\mathrm{lift}}(M).
\label{eq:converse-optimum-lift-sandwich}
\end{align}
Consequently, every upper bound on
$R_{\mathrm{lift}}(M)/R_{\mathrm{conv}}(M)$ also bounds
$R_{\mathrm{lift}}(M)/R^*(M)$.  Define the access-balance parameter
\begin{align}
\alpha_{K,L}
\triangleq
\min\left\{\frac{L}{K},1-\frac{L}{K}\right\}.
\end{align}

\subsubsection{Libraries at least as large as the user population}

When $N\ge K$, the feasibility threshold in
Lemma~\ref{lem:feasible-memory} equals the $t=0$ lifted memory
$(K-1)/(K-L)$.  Hence the lifted envelope begins at $M_{\mathrm{feas}}$ and is
finite throughout $M\ge M_{\mathrm{feas}}$.

\begin{Theo}[Approximation guarantee for $N\ge K$]
\label{thm:order-opt}
Assume $N\ge K$.  For every feasible cache size $M$,
\begin{align}
\frac{R_{\mathrm{lift}}(M)}{R_{\mathrm{conv}}(M)}
\le\frac{12}{\alpha_{K,L}}.
\label{eq:order-opt-main}
\end{align}
The same right-hand side bounds
$R_{\mathrm{lift}}(M)/R^*(M)$.
\end{Theo}

\begin{IEEEproof}[Proof sketch]
The piecewise-linear envelope bounds from
\eqref{eq:M-t} through \eqref{eq:endpoint} are compared with selected $s$-terms of
\eqref{eq:R-conv}.  The local-cost factors $K/L$ and $K/(K-L)$ bound the
storage required by the two encoders.  Combining these bounds with the
feasibility threshold and the exact rate-one point gives
$12/\alpha_{K,L}$ in \eqref{eq:order-opt-main}.
\end{IEEEproof}

In particular, if $\alpha_{K,L}\ge\lambda$ for a fixed
$\lambda\in(0,1/2]$, then both ratios in
Theorem~\ref{thm:order-opt} are at most $12/\lambda$.  This gives a uniform
constant for access fractions $L/K$ bounded away from zero and one.

\begin{Remark}[Illustration of the memory-rate bounds]
Figure~\ref{fig:memory-rate-comparisons} compares the achievable envelope
with the converse for \(K=12\) and two access sizes.  For \(N=24\) and
\(L=6\), the lifting envelope starts at the feasibility threshold
\(M=11/6\).  The two largest-access panels illustrate
Theorem~\ref{thm:lkm1-boundary}: for \(N=24\), the exact tradeoff is the
segment from \((11,12)\) to \((23,1)\), and for \(N=8\), the optimum is
\(R^*_{12,11,8}(M)=1\) for every \(M\ge7\).
\end{Remark}

\begin{figure}[!htbp]
  \centering
  \includegraphics[width=0.98\textwidth]{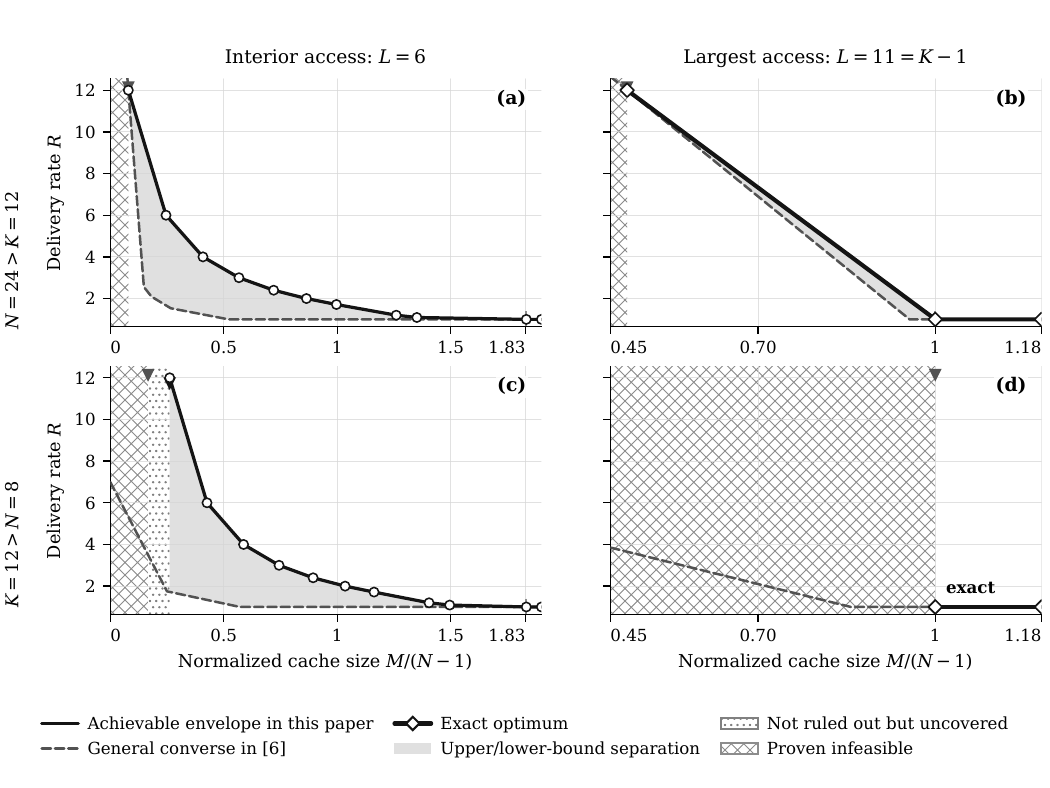}
  \caption{Memory-rate comparisons for \(K=12\).  The columns use
  \(L=6\) and \(L=11=K-1\); the rows use \(N=24>K\) and \(K>N=8\).
  The horizontal axis is \(M/(N-1)\).
  The solid curve and open circles give this paper's achievable envelope
  \(R_{\mathrm{lift}}(M)\) and its selected corners.
  The heavier solid curve and diamonds give the exact largest-access tradeoff
  in \cite[Theorem~1]{FangLargestAccess}, restated as
  Theorem~\ref{thm:lkm1-boundary}.
  The dashed curve gives the general secrecy-aware converse
  \(R_{\mathrm{conv}}(M)\) from
  \cite[Theorem~6]{ImprovedLowerBoundsForMultiAccessCodedCaching}, stated here
  as Theorem~\ref{thm:macc-converse}.
  Gray shading fills the gap between the achievable envelope and the converse.
  Cross hatching marks memories excluded by Lemma~\ref{lem:feasible-memory}.
  In panel~(c), dot hatching marks \(7/6\le M<11/6\), between that necessary
  memory bound and the first lifted point.
  Panel~(a) attains the necessary memory bound at its first achievable point;
  panels~(b) and~(d) give the exact largest-access tradeoffs.}
  \label{fig:memory-rate-comparisons}
\end{figure}

\FloatBarrier

For $N\ge K\ge4$ and $L=2$, the lifted envelope has the
$K$-independent approximation factor $24$.

\begin{Propo}[Refined approximation guarantee for $L=2$]
\label{prop:L2-order}
Assume $N\ge K\ge4$ and $L=2$.  For every feasible cache size $M$,
\begin{align}
\frac{R_{\mathrm{lift}}(M)}{R_{\mathrm{conv}}(M)}\le24,
\qquad
\frac{R_{\mathrm{lift}}(M)}{R^*(M)}\le24.
\label{eq:L2-order-main}
\end{align}
\end{Propo}

\begin{IEEEproof}[Proof sketch]
At $L=2$, the bounds $R_{\mathrm{lift}}(M)\le3N/M$ for $M\le N-1$ and
$R_{\mathrm{lift}}(M)\le2$ for $M\ge N-1$, together with the corresponding
terms of \eqref{eq:R-conv}, give the factor $24$.
Equation~\eqref{eq:converse-optimum-lift-sandwich} gives the same factor
relative to $R^*(M)$.
\end{IEEEproof}

\subsubsection{More users than files}

For $K>N$, we compare the lifted rate with the converse from the first
lifted memory onward.  Proposition~\ref{prop:order-opt-small-library} defines
this threshold and gives a tighter approximation factor at larger cache
budgets.  It also identifies when rate one is optimal throughout the
comparison domain.

\begin{Propo}[%
Approximation guarantee for $K>N$
]
\label{prop:order-opt-small-library}
Assume $2\le N<K$, and define
\begin{subequations}
\label{eq:small-library-thresholds}
\begin{align}
M_{\mathrm{feas}}
&\triangleq\frac{N-1}{K-L},
\\
M_{\mathrm{start}}
&\triangleq
\min\left\{
\frac{K-1}{K-L},
\frac{(N-1)(K-1)}{L}
\right\},
\\
M_{\mathrm{imp}}
&\triangleq
\frac{\max\{K-1,2(K-N)\}}{\min\{L,K-L\}}.
\end{align}
\end{subequations}
For every $M\ge M_{\mathrm{start}}$,
\begin{align}
\frac{R_{\mathrm{lift}}(M)}{R_{\mathrm{conv}}(M)}
&\le\frac{12(K/N)}{\alpha_{K,L}},
&
\frac{R_{\mathrm{lift}}(M)}{R^*(M)}
&\le\frac{12(K/N)}{\alpha_{K,L}}.
\label{eq:small-library-order-main}
\end{align}
For $M\ge M_{\mathrm{imp}}$, both right-hand sides can be replaced by
\begin{align}
\frac{12\min\{K/N,9/8\}}{\alpha_{K,L}}
\le\frac{27}{2\alpha_{K,L}}.
\label{eq:small-library-refined-order-main}
\end{align}
If
\begin{align}
\frac{(N-1)(K-1)}{L}\le\frac{K-1}{K-L},
\label{eq:small-library-unit-first-condition}
\end{align}
then $R_{\mathrm{lift}}(M)=R_{\mathrm{conv}}(M)=R^*(M)=1$ for every
$M\ge M_{\mathrm{start}}$.
\end{Propo}

\begin{IEEEproof}[Proof sketch]
The bounds on the memory required at the lifted corner points give the
approximation in \eqref{eq:small-library-order-main}.  The factor $K/N$
comes from the parent rate for $K$ users and the converse based on $N$
independent files.  Retaining the dependence on $N$ in these memory bounds
and applying memory sharing between adjacent corners gives the refined
factor in \eqref{eq:small-library-refined-order-main} for
$M\ge M_{\mathrm{imp}}$.  Under
\eqref{eq:small-library-unit-first-condition}, the rate-one point is the
earliest lifted point, and the $s=1$ converse term makes rate one exact
thereafter.
\end{IEEEproof}

\begin{Remark}[Comparison domain for $K>N$]
For $K>N$, define the interval between the necessary feasibility threshold
and the first lifted memory by
\begin{align}
\mathcal M_{\mathrm{uncov}}
\triangleq [M_{\mathrm{feas}},M_{\mathrm{start}}).
\end{align}
The lifted envelope and its ratio guarantee begin at
$M_{\mathrm{start}}$.  At $L=K-1$, the interval is empty,
and the exact boundary result applies throughout $M\ge M_{\mathrm{feas}}$.
\end{Remark}

Complete proofs of the converse and approximation guarantees appear in
Appendix~\ref{app:order-opt}.

\subsection{Boundary Characterizations}

\subsubsection{\texorpdfstring{The secure single-access boundary $L=1$}{The secure single-access boundary L=1}}

When $L=1$, each user reads one cache, and the model is secure single access.

\begin{Propo}[Known secure single-access benchmark]
Assume $L=1$.  For $K\ge3$ and $t\in[1:K-2]$, the securely achievable
single-access points of \cite{FangSmallRate} include
\begin{align}
M_{\mathrm{SA}}^{\mathrm{imp}}(t)
&=\frac{Nt}{K-t}+1-\frac{1}{\binom{K-1}{t}},
&
R_{\mathrm{SA}}^{\mathrm{imp}}(t)
&=\frac{K}{t+1}.
\label{eq:L1-improved-points}
\end{align}
The exact minimum feasible cache size is $M=1$.  At that cache size the optimum
rate is $K-1$ for $N=2$ and $K$ for $N\ge3$.  The exact rate-one cache
threshold is $(N-1)(K-1)$, consistently with
Corollary~\ref{cor:exact-unit-rate} at $L=1$.
\end{Propo}

\subsubsection{\texorpdfstring{The largest-access boundary $L=K-1$}{The largest-access boundary L=K-1}}

For $L=K-1$, every user observes all but one physical cache.  The lifted family
contains the two points $(K-1,K)$ and $(N-1,1)$.  When $N>K$, the line through
these points is supported by the affine converse
\begin{align}
(K-1)M+(N-K)R\ge KN-2K+1.
\label{eq:lkm1-line}
\end{align}
The exact $L=K-1$ tradeoff was first established in the conference
version~\cite{FangLargestAccess}.  Here it follows from two lifted points and
matching converse bounds.

\begin{Lem}[Elementary largest-access converse]
\label{lem:lkm1-basic}
For $L=K-1$, every securely achievable pair satisfies
\begin{align}
R&\ge1,
\\
M&\ge\min\{N,K\}-1.
\end{align}
\end{Lem}

\begin{IEEEproof}[Proof sketch]
Security under an alternative demand makes a user's accessible-cache tuple
independent of a file it did not request.  Correctness under a demand for that
file then forces the broadcast to contain at least one file unit.  For the
memory bound, let $D=\min\{N,K\}$ users request $D$ distinct files.
The full cache collection and broadcast decode all $D$ files.  One user's
accessible caches together with the broadcast are independent of the other
$D-1$ files.
Consequently, the omitted physical cache must carry at least $D-1$ file units.
\end{IEEEproof}

\begin{Propo}[Affine largest-access converse]
\label{prop:lkm1-line}
If $L=K-1$ and $N\ge K+1$, every securely achievable pair satisfies
\eqref{eq:lkm1-line}.
\end{Propo}

\begin{IEEEproof}[Proof sketch]
Symmetrize over file indices and fix a demand for $K$ distinct files.
For $a\in[1:N]$, define the cumulative entropy
$\mathcal E(a)=H(W_{1:a},Z_{[K]})$.  Submodularity and file symmetry make this
sequence discretely concave.  Its average increment over the requested
files after the first is therefore at least its average increment over the
unrequested files.
Collective decoding and conditioning bound the broadcast entropy and the
residual uncertainty in the unrequested files by these increments.  Adding
the secrecy bound for the cache omitted by one user gives \eqref{eq:lkm1-line}.
\end{IEEEproof}

\begin{Theo}[Exact largest-access tradeoff
{\cite[Theorem~1]{FangLargestAccess}}]
\label{thm:lkm1-boundary}
Fix $L=K-1$.
\begin{enumerate}[label=(\roman*),leftmargin=2.1em]
\item If $N\ge K+1$, no securely achievable pair exists for $M<K-1$, and
\begin{align}
R^*_{K,K-1,N}(M)
=
\begin{cases}
\displaystyle
\frac{KN-2K+1-(K-1)M}{N-K},
& K-1\le M\le N-1,\\[2ex]
1, & M\ge N-1.
\end{cases}
\end{align}
\item If $K\ge N$, no securely achievable pair exists for $M<N-1$, and
\begin{align}
R^*_{K,K-1,N}(M)=1,
\qquad M\ge N-1.
\end{align}
\end{enumerate}
\end{Theo}

\begin{IEEEproof}[Proof sketch]
Theorem~\ref{thm:memoryload} at $t=0$ gives the achievable point $(K-1,K)$,
and \eqref{eq:endpoint} gives $(N-1,1)$.
Lemma~\ref{lem:lkm1-basic} gives the
elementary lower bounds.  If $N\ge K+1$,
Proposition~\ref{prop:lkm1-line} rearranges to exactly the memory-sharing line
between the two achievable points.  If $K\ge N$, the elementary memory bound
and the unit-rate point already match at $M=N-1$.
\end{IEEEproof}

Because $K=2$ forces $L=1=K-1$, the theorem also gives the exact tradeoff for
two-user secure single access.

\begin{Cor}[Exact two-user secure single-access tradeoff]
\label{cor:k2-single-access}
For $K=2$ and $L=1$, no securely achievable pair exists for $M<1$.  If
$N\ge3$, then
\begin{align}
R^*_{2,1,N}(M)
=
\begin{cases}
\dfrac{2N-3-M}{N-2}, & 1\le M\le N-1,\\[1.5ex]
1, & M\ge N-1.
\end{cases}
\end{align}
If $N=2$, then
\begin{align}
R^*_{2,1,2}(M)=1,
\qquad M\ge1.
\end{align}
\end{Cor}

Appendix~\ref{app:largest-access} gives the complete largest-access converse
and tradeoff proofs, including the two-user case.

\section{Conclusion}

The proposed lifting re-encodes the cache blocks of a secure single-access
parent with a cyclic-block representation.  Each \(L\)-cache window recovers
its target formal positions and reveals no information about the complementary
positions conditioned on those targets.  Each user can therefore run the
parent decoder with the original multicast message, preserving its delivery
rate and file secrecy.  Deterministic Pascal-window
encoding realizes blocks with appearance spans \(b\ge L\).  A complement-dual
affine-coset encoder with placement-time randomness realizes those with \(b<L\).

For the subset-indexed parent and its supplied block list, the two encoders
minimize the largest
per-cache fixed-length allocation among parent-preserving separate-segment
transformations valid for every input distribution.  Their costs give the
achievable memory-rate envelope.  The unrestricted converse bounds give
approximation guarantees for every feasible memory when \(N\ge K\), and for
\(M\ge M_{\mathrm{start}}\) when \(K>N\).  At \(L=K-1\), the achievable
envelope and the sharp converse give the exact tradeoff, including the
two-user secure single-access case.

\bibliographystyle{IEEEtran}
\bibliography{references}

\clearpage
\appendices

\section{Parent Representation, Scalarization, and Secure Inheritance}
\label{app:inheritance}

Each formal occurrence refers to a named parent random variable.  The lift
encodes proper blocks separately, stores common symbols in an MDS pool, and
concatenates the resulting segments in a fixed order within each physical cache.
Throughout, the public demand vector is fixed in the conditioning.  The entire
vector of local seeds is sampled during placement, independently of the file
library, parent randomness, and demand vector.  Each scalar group of each
randomized block instance has its own seed vector.  These seed vectors are
mutually independent.

\subsection{Common scalarization and the common-symbol pool}

\begin{IEEEproof}[Proof of Lemma~\ref{lem:common-scalarization} and Remark~\ref{rem:fixed-field-scalarization}]
Put \(\chi=\operatorname{lcm}(L,K-L)\).  Suppose the parent is supplied over
\(\mathbb F_Q\), where \(Q=p^e\), \(p>K\), and \(\chi\mid e\).
Fix a basis \((\beta_1,\ldots,\beta_e)\) of \(\mathbb F_Q\) over
\(\mathbb F_p\), and write
\begin{align}
y=\sum_{a=1}^{e}y_a\beta_a,
\qquad
\varphi(y)=(y_1,\ldots,y_e)\in\mathbb F_p^e.
\end{align}
The map \(\varphi\) is a deterministic bijection.  Applied coordinatewise to
a named parent symbol \(Y\in\mathbb F_Q^{\ell_Y}\), it gives
\(\overline Y\in\mathbb F_p^{d_Y}\), where \(d_Y=e\ell_Y\).
Both \(L\) and \(K-L\) divide \(e\), so they divide \(d_Y\) for every
positive integer \(\ell_Y\).
Every formal reference to \(Y\) uses this same vector \(\overline Y\).

Partition its coordinates into groups of length \(L\) for the direct encoder,
or length \(K-L\) for the dual encoder.  The respective numbers of groups
are exactly \(d_Y/L\) and \(d_Y/(K-L)\).
The local decoders recover the target coordinates.  Applying
\(\varphi^{-1}\) coordinatewise then recovers the parent symbols.
Condition on all parent symbols.  The local seeds remain independent, so
the outputs of different scalar groups are conditionally independent.
Within each group, the distribution of the accessible output depends only
on the target coordinates.  Basis expansion is a bijection and preserves
conditional mutual information.  Conditional secrecy therefore holds for
every joint law of the parent variables.

For a proper block with \(\ell\) parent-field coordinates at each formal
position, the direct encoder stores \(b\) prime-field coordinates per group
in each cache.  The dual encoder stores \(K-b\) coordinates per group.
Concatenating the groups gives the following exact base-\(Q\) cache lengths:
\begin{align}
\log_Q|\mathcal C_i|
&=
\begin{cases}
\displaystyle
\log_Q p^{e\ell b/L}=\ell b/L,&b\ge L,\\[1ex]
\displaystyle
\log_Q p^{e\ell(K-b)/(K-L)}
=\ell(K-b)/(K-L),&b<L.
\end{cases}
\end{align}
The exponents are integers because \(L\mid e\) and \(K-L\mid e\).
These equalities hold for every original length \(\ell\).
The parent files, named symbols, and delivery message retain their original
\(\mathbb F_Q\) values and lengths.

Every user must recover each common symbol from its \(L\) accessible caches.
For a common symbol \(Y\), split its expansion into rows
\(\overline Y_1,\ldots,\overline Y_L\in\mathbb F_p^{d_Y/L}\).
Choose \(K\) distinct elements of \(\mathbb F_p\) and form the corresponding
\(K\times L\) Vandermonde generator \(G\).  Every \(L\) rows are
nonsingular because their determinant is a product of nonzero differences.
Cache \(i\) stores
\begin{align}
C_i^{\rm com}(Y)=\sum_{j=1}^{L}G_{i,j}\overline Y_j.
\end{align}
Any \(L\) shares recover \(\overline Y\), and hence \(Y\).
Each share has prescribed alphabet \(\mathbb F_p^{d_Y/L}\), whose
base-\(Q\) length is \(\ell_Y/L\).
Concatenating these shares in a fixed order gives the common layer, a
deterministic function of the common parent symbols.

For rational memory sharing over a common parent field, write the finite
weights as \(\lambda_s=a_s/A\), with nonnegative integers \(a_s\) summing
to \(A\), and let \(g_s\) be a valid file length of component \(s\).
Independent copies of the component code realize every positive integer
multiple of its file length.  Choose the total file length
\begin{align}
F=A\operatorname{lcm}_s g_s
\end{align}
and assign \(F_s=\lambda_sF=a_s\operatorname{lcm}_s g_s\) coordinates to
component \(s\), a multiple of \(g_s\).
Each component decodes its own file part.  Independence of the file parts and
component seeds makes the joint conditional information leakage the sum of
the component leakages, each zero.  The concatenated alphabets satisfy
\begin{align}
M_F\le\sum_s\lambda_sM_s,
\qquad
R_F\le\sum_s\lambda_sR_s.
\end{align}
Basis expansion preserves these lengths in parent-field units.

For a parent supplied over a fixed field \(\mathbb F_q\) of characteristic
greater than \(K\), use independent copies to make the length of every named
symbol divisible by \(\chi\).  Let \(\mathcal Y\) be its finite set of named
symbols of positive length, and let \(\ell_Y\) now denote the number of
\(\mathbb F_q\)-coordinates in \(Y\).
Set
\begin{align}
r_Y=\frac{\chi}{\gcd(\chi,\ell_Y)},
\qquad
r=\operatorname{lcm}_{Y\in\mathcal Y}r_Y,
\end{align}
with the least common multiple of an empty family equal to one.
Then \(\widehat d_Y=r\ell_Y\) is divisible by \(\chi\).
Use \(r\) independent parent layers and assign
\((Y^{[1]},\ldots,Y^{[r]})\) to every occurrence of the same named symbol.
The \(\widehat d_Y/L\) direct groups and
\(\widehat d_Y/(K-L)\) dual groups operate over \(\mathbb F_q\), and the
common-symbol MDS layer uses a \(K\times L\) generator over the same field.

Let \(O_u^{[a]}=(\widetilde Z_u^{[a]},X_{\mathbf d}^{[a]})\) denote user
\(u\)'s parent observation in layer \(a\).  Independence of the parent layers
and their seeds gives
\begin{align}
&I\!\left(
  \bigl(W_{[N]\setminus\{d_u\}}^{[a]}\bigr)_{a=1}^{r};
  \bigl(O_u^{[a]}\bigr)_{a=1}^{r}
  \middle|
  \bigl(W_{d_u}^{[a]}\bigr)_{a=1}^{r},\mathbf d
\right)\notag\\
&\quad=
\sum_{a=1}^{r}
 I\!\left(W_{[N]\setminus\{d_u\}}^{[a]};O_u^{[a]}
 \middle|W_{d_u}^{[a]},\mathbf d\right)=0.
\end{align}
Each layer uses its own parent decoder.  Product alphabets multiply every
file, cache, and broadcast length by \(r\), preserving normalized memory
and delivery rate for this fixed-field realization.
\end{IEEEproof}

\subsection{Local-to-global inheritance}

\begin{IEEEproof}[Proof of Proposition~\ref{prop:inherit}]
Let \(S_{\rm par}\) collect all parent placement and delivery seeds, including
those used by stochastic parent encoders, and put
\(\Pi=(W_{[N]},S_{\rm par})\).  For each demand vector, all parent caches and
\(X_{\mathbf d}\) are deterministic functions of \(\Pi\).
For the fixed-field realization, coordinate expansion and its inverse are
the identity maps on each repeated parent layer.

Fix user \(u\).  Let \(\mathcal J\) index all proper-block scalar-group
instances.  For \(j\in\mathcal J\), let \(b_j\) be the appearance span of its
block and \(V^{(j)}\) its full formal source vector.  Define
\begin{align*}
T_{u,j}\triangleq V^{(j)}_{\mathcal T_u^{(b_j)}},
\qquad
B_{u,j}\triangleq C_{\mathcal A_u}^{(j)}.
\end{align*}
Let \(G_{\rm com}\) be the common-symbol pool and
\(D_u=C_{\mathcal A_u}^{\rm com}\) the corresponding accessible MDS shares.
Collect all target groups and the common symbols in the fixed-order tuple
\begin{align}
\Theta_u\triangleq\bigl((T_{u,j})_{j\in\mathcal J},G_{\rm com}\bigr).
\end{align}
The admissible representation supplies deterministic maps \(\delta_u\) and
\(\epsilon_u\) satisfying
\begin{align}
\widetilde Z_u=\delta_u(\Theta_u),
\qquad
\Theta_u=\epsilon_u(\widetilde Z_u).
\label{eq:supp-inherit-reference-equivalence}
\end{align}
The map \(\delta_u\) regroups the decoded coordinates, applies the inverse basis
map, and removes duplicate named references.  The map \(\epsilon_u\)
expands each parent component in the fixed basis, recopies it at every
formal occurrence, and partitions its coordinates into the prescribed groups.

The physical cache stores the ordered tuple
\begin{align}
Z_i^{\rm lift}=\left(\bigl(C_i^{(j)}\bigr)_{j\in\mathcal J},C_i^{\rm com}\right).
\label{eq:supp-inherit-aggregate-cache}
\end{align}
Local reconstruction gives \(H(T_{u,j}\mid B_{u,j})=0\), and the MDS layer
gives \(H(G_{\rm com}\mid D_u)=0\).  The user therefore recovers
\(\Theta_u\) from \(Z_{\mathcal A_u}^{\rm lift}\) and applies \(\delta_u\), so
\begin{align}
H(\widetilde Z_u\mid Z_{\mathcal A_u}^{\rm lift})=0.
\label{eq:supp-inherit-effective-recovery}
\end{align}

Local validity for every joint input law means that the accessible-output law
depends only on the target.  Thus there is a stochastic kernel \(Q_{u,j}\) such that
\begin{align}
\Pr\{B_{u,j}=b\mid V^{(j)}=v\}
=Q_{u,j}\bigl(b\mid v_{\mathcal T_u^{(b_j)}}\bigr).
\end{align}
To obtain this kernel, take any two formal input vectors with the same target
and assign positive probability to each.  Zero conditional mutual information
forces their accessible-output laws to coincide.  Since the two vectors were
arbitrary, the kernel depends only on the target, including when formal
positions are correlated or refer to the same parent variable.

Each \(V^{(j)}\) is a deterministic function of \(\Pi\).  Given
\(\Pi=\pi\), all formal sources are fixed.  The local seeds retain their joint
product law because their entire vector is independent of parent randomness,
the file library, and the demand vector.  For every accessible-output
tuple \((c_j)_{j\in\mathcal J}\),
\begin{align}
&\Pr\{(B_{u,j})_{j\in\mathcal J}=(c_j)_{j\in\mathcal J}\mid\Pi=\pi\}
\notag\\
&\hspace{4em}=
\prod_{j\in\mathcal J}
Q_{u,j}\bigl(c_j\mid T_{u,j}(\pi)\bigr).
\end{align}
The common shares satisfy \(D_u=m_u(G_{\rm com})\).  Together with the preceding
product formula, this shows that the conditional distribution of the full view
\(Z_{\mathcal A_u}^{\rm lift}\equiv((B_{u,j})_{j\in\mathcal J},D_u)\) given \(\Pi\)
is determined by \(\Theta_u\).
Equation~\eqref{eq:supp-inherit-reference-equivalence} expresses \(\Theta_u\)
as a deterministic function of \(\widetilde Z_u\), so this distribution is
determined by \(\widetilde Z_u\).
Equivalently,
\begin{align}
I(\Pi;Z_{\mathcal A_u}^{\rm lift}\mid\widetilde Z_u)=0.
\end{align}
Since \(X_{\mathbf d}\) is determined by \(\Pi\), conditioning also on the
broadcast and requested file preserves this independence:
\begin{align}
I\!\left(W_{[N]\setminus\{d_u\}};Z_{\mathcal A_u}^{\rm lift}
\middle|\widetilde Z_u,X_{\mathbf d},W_{d_u},\mathbf d\right)=0.
\label{eq:supp-inherit-view-zero}
\end{align}

The user first recovers its parent cache by
\eqref{eq:supp-inherit-effective-recovery}, then applies the parent decoder:
\begin{align}
H(W_{d_u}\mid Z_{\mathcal A_u}^{\rm lift},X_{\mathbf d},\mathbf d)=0.
\end{align}
For secrecy, insert the recoverable cache \(\widetilde Z_u\) into the
observation and apply the chain rule:
\begin{align}
&I\!\left(W_{[N]\setminus\{d_u\}};
Z_{\mathcal A_u}^{\rm lift},X_{\mathbf d}\middle|W_{d_u},\mathbf d\right)
\notag\\
&=I\!\left(W_{[N]\setminus\{d_u\}};
\widetilde Z_u,Z_{\mathcal A_u}^{\rm lift},X_{\mathbf d}
\middle|W_{d_u},\mathbf d\right)
\notag\\
&=I\!\left(W_{[N]\setminus\{d_u\}};
\widetilde Z_u,X_{\mathbf d}\middle|W_{d_u},\mathbf d\right)
\notag\\
&\quad+I\!\left(W_{[N]\setminus\{d_u\}};Z_{\mathcal A_u}^{\rm lift}
\middle|\widetilde Z_u,X_{\mathbf d},W_{d_u},\mathbf d\right)=0.
\end{align}
The two terms vanish by parent security and
\eqref{eq:supp-inherit-view-zero}, respectively.
The transmitted random variable is exactly the parent's \(X_{\mathbf d}\).
Basis expansion preserves the original file and broadcast lengths.  The
fixed-field realization multiplies both lengths by its finite repetition factor.
The delivery rate is unchanged in both realizations.
\end{IEEEproof}

\section{Complete Proofs for the Direct and Dual Local Encoders}
\label{app:local-encoders}

User, cache, and formal-position indices in this section belong to
\(\mathbb Z_K\) and are reduced modulo \(K\).  Scalar-coordinate indices,
polynomial exponents, and antidiagonal indices use ordinary integers.
Write \(q=p^\nu\), with \(\nu\) a positive integer and \(p>K\) the prime
characteristic of the local coding field.
For a block matrix, \(M(\mathcal R,\mathcal S)\) denotes the indicated block
rows and columns in their displayed cyclic orders.

\subsection{Sparse roots and Pascal windows}

The direct decoder reduces to inverting consecutive windows of a Pascal
matrix.  We establish their invertibility using a bound on root multiplicity
for sparse polynomials.

\begin{Lem}[Multiplicity of a nonzero root of a sparse polynomial]
\label{lem:supp-sparse-root}
Let \(f\) be a polynomial with \(J\) nonzero terms:
\begin{align}
f(z)=\sum_{j=0}^{J-1}c_jz^{e_j},
\qquad
c_j\ne0,
\qquad
0\le e_0<\cdots<e_{J-1}.
\end{align}
If the field characteristic is larger than \(e_{J-1}\), every nonzero root
of \(f\) has multiplicity at most \(J-1\).
\end{Lem}

\begin{IEEEproof}
For \(J=1\), a nonzero monomial has no nonzero root.  Suppose \(J\ge2\) and a nonzero \(\alpha\)
has multiplicity at least \(J\).  Write
\(e^{\underline r}=e(e-1)\cdots(e-r+1)\) for the falling factorial, with
\(e^{\underline 0}=1\).  Then \(f^{(r)}(\alpha)=0\) for
\(r=0,\ldots,J-1\).  Multiplying the \(r\)th equation by \(\alpha^r\) gives
\begin{align}
\sum_{j=0}^{J-1}c_j\alpha^{e_j}e_j^{\underline r}=0,
\qquad r=0,\ldots,J-1.
\label{eq:supp-falling-system}
\end{align}
The coefficient matrix
\([e_j^{\underline r}]_{r,j=0}^{J-1}\) is obtained from the ordinary Vandermonde
matrix by a unit lower-triangular row transformation: each
\(x^{\underline r}\) is monic of degree \(r\).  Its determinant is therefore
\begin{align}
\prod_{0\le a<c\le J-1}(e_c-e_a),
\end{align}
which is nonzero under the characteristic assumption.  Thus
\eqref{eq:supp-falling-system} forces every
\(c_j\alpha^{e_j}\) to vanish, a contradiction.
\end{IEEEproof}

\begin{Lem}[Pascal-window invertibility]
\label{lem:supp-pascal-window}
Let \(w\ge0\), \(r\ge1\), and \(0\le\delta\le w\), with
\(\binom wt=0\) outside \(0\le t\le w\).  If the field
characteristic is larger than \(w+r-1\), then the matrix
\begin{align}
B_{w,r,\delta}
\triangleq
\left[\binom{w}{\delta+a-c}\right]_{a,c=0}^{r-1}
\end{align}
is nonsingular.
\end{Lem}

\begin{IEEEproof}
For \(w=0\), the matrix is \(I_r\).  For \(w\ge1\), take
\(v=(v_0,\ldots,v_{r-1})^{\mathsf T}\) in its kernel and set
\begin{align}
V(z)=\sum_{c=0}^{r-1}v_cz^c,
\qquad Q(z)=(1+z)^wV(z).
\end{align}
The coefficients of \(Q\) in the consecutive degree window
\(\delta,\ldots,\delta+r-1\) vanish.  Since \(\deg Q\le w+r-1\), write
\begin{align}
Q(z)=Q_-(z)+z^{\delta+r}Q_+(z),
\quad
\deg Q_-<\delta,
\quad
\deg Q_+<w-\delta.
\end{align}
The two remaining degree ranges allow at most \(w\) nonzero monomials in total.
Divisibility by \((1+z)^w\) makes \(-1\) a root of multiplicity at least
\(w\).  If \(Q\ne0\), every exponent is below the characteristic, so
Lemma~\ref{lem:supp-sparse-root} bounds this multiplicity by \(w-1\),
a contradiction.  Therefore \(Q=0\), so \(V=0\) and \(v=0\).
\end{IEEEproof}

\subsection{\texorpdfstring{Direct encoder for \(L\le b\le K-1\)}{Direct encoder for L<= b<= K-1}}

\begin{IEEEproof}[Proof of Theorem~\ref{thm:direct-local-transform}]
Put \(\omega=b-L\).  In the encoder
\eqref{eq:direct-local-encoder}, scalar coordinate \(m\in[0:b-1]\) of each cache is
\begin{align}
C_{i,m}=
\sum_{k=\max(0,m-\omega)}^{\min(m,L-1)}
\binom{\omega}{m-k}v_{i+m-k,k}.
\label{eq:supp-direct-scalar-entry}
\end{align}
Consequently, the global matrix, with row index \((i,m)\) and column index
\((n,k)\), has entry
\begin{align}
H_{(i,m),(n,k)}=
\begin{cases}
\binom{\omega}{m-k},
 &n\equiv i+m-k\pmod K,\ 0\le m-k\le\omega,\\
0,&\text{otherwise}.
\end{cases}
\label{eq:supp-direct-H-entry}
\end{align}

By cyclic translation, it suffices to consider \(u=0\).  If
\(i\in\mathcal A_0=[0:L-1]\) and a source block \(n\) contributes to cache
\(i\), then \(n=i+j\) for some \(j\in[0:\omega]\).  Its ordinary integer
representative obeys
\(0\le i+j\le L-1+\omega=b-1<K\), so
\(n\in[0:b-1]=\mathcal T_0^{(b)}\).  Hence the complement submatrix in
\eqref{eq:direct-local-zero} is zero for \(u=0\), and cyclic translation gives
the same property for every user.

There are \(Lb=bL\) local scalar equations and unknowns.  Whenever
\(v_{n,k}\) occurs in \(C_{i,m}\), \eqref{eq:supp-direct-H-entry} gives the
antidiagonal invariant
\begin{align}
n+k=i+m\triangleq s.
\end{align}
Thus the local system splits over \(s=0,\ldots,b+L-2\).  For fixed \(s\), set
\begin{align}
\alpha_s&=\max(0,s-b+1),&
\beta_s&=\min(L-1,s),
\notag\\
r_s&=\beta_s-\alpha_s+1,&
\delta_s&=s-\alpha_s-\beta_s.
\end{align}
Order the unknowns and equations as
\begin{align}
v_q^{(s)}&=v_{s-(\alpha_s+q),\alpha_s+q},
&&q=0,\ldots,r_s-1,
\notag\\
(i,m)&=(\beta_s-a,s-\beta_s+a),
&&a=0,\ldots,r_s-1.
\end{align}
Substitution in \eqref{eq:supp-direct-scalar-entry} gives
\begin{align}
C_{\beta_s-a,s-\beta_s+a}
=\sum_{q=0}^{r_s-1}
\binom{\omega}{\delta_s+a-q}v_q^{(s)}.
\end{align}
The coefficient block is \(B_{\omega,r_s,\delta_s}\).  For
\(0\le s\le L-1\), we have \(\delta_s=0\).
For \(L\le s\le b-1\), we have \(\delta_s=s-L+1\in[1:\omega]\).
For \(b\le s\le b+L-2\), we have \(\delta_s=\omega\).  Hence
\(0\le\delta_s\le\omega\).  Moreover, \(r_s\le L\), so
\begin{align}
\omega+r_s-1\le b-1<K<p.
\end{align}
Each antidiagonal block is nonsingular by Lemma~\ref{lem:supp-pascal-window}.
Ordering the equations and unknowns by \(s\) puts these blocks on the diagonal,
so the local matrix satisfies
\eqref{eq:direct-local-invertible} and can be inverted as in
\eqref{eq:direct-local-decoder}.

The accessible output is a deterministic function only of the target source
tuple.  Hence, for every joint input law,
\begin{align}
I\!\left(\mathbf V_{(\mathcal T_u^{(b)})^{\rm c}};
\mathbf C_{\mathcal A_u}\middle|
\mathbf V_{\mathcal T_u^{(b)}}\right)=0.
\end{align}
If \(\mathbf Hv=0\), applying each user's local inverse gives
  \(v_{\mathcal T_u^{(b)}}=0\).  The target windows cover \(\mathbb Z_K\), so
  \(v=0\) and \(\mathbf H\) has full column rank.  Each formal position has
  \(L\) source coordinates and each cache segment has \(b\) stored coordinates,
  so the normalized cost is \(b/L\).
\end{IEEEproof}

At \(b=L\), we have \(\omega=0\): the Pascal blocks reduce
to identity blocks, the local seed length is zero, and the normalized cost is
one.

\subsection{Complement alignment and systematic affine sampling}

For a block vector \(x=(x_0^{\mathsf T},\ldots,x_{K-1}^{\mathsf T})^{\mathsf T}\),
the positive shift is defined by
\((P_\Delta x)_{i+\Delta}=x_i\), equivalently
\((P_\Delta x)_r=x_{r-\Delta}\).  For a cyclic set, write
\(\mathcal S\oplus\Delta=\{i+\Delta:i\in\mathcal S\}\).

\begin{Lem}[Complement alignment]
\label{lem:supp-complement-alignment}
Let \(1\le b<L<K\), and set
\begin{align}
L'=K-L,
\qquad b'=K-b,
\qquad \Delta=L-b.
\end{align}
For the auxiliary direct system define
\(\mathcal A'_a=\{a,\ldots,a+L'-1\}\) and
\(\mathcal T'_a=\{a,\ldots,a+b'-1\}\).  With \(a=u+L\),
\begin{align}
\mathcal A'_a&=\mathcal A_u^{\rm c},&
\mathcal T'_a&=(\mathcal T_u^{(b)})^{\rm c}\oplus\Delta,&
(\mathcal T'_a)^{\rm c}&=\mathcal T_u^{(b)}\oplus\Delta.
\end{align}
\end{Lem}

\begin{IEEEproof}
The length-\(L'\) window beginning at \(u+L\) is the complement of
the length-\(L\) window beginning at \(u\).  The shifted target complement is
\begin{align*}
(\mathcal T_u^{(b)})^{\rm c}\oplus\Delta
&=\{u+b,\ldots,u-1\}\oplus(L-b)\\
&=\{u+L,\ldots,u+L+b'-1\}=\mathcal T'_{u+L}.
\end{align*}
Complementation gives
\((\mathcal T'_{u+L})^{\rm c}=\mathcal T_u^{(b)}\oplus\Delta\).
\end{IEEEproof}

\begin{Lem}[Systematic right inverse and kernel basis]
\label{lem:supp-systematic-affine}
Let \(A\in\mathbb F_q^{m\times n}\) have full row rank and put \(d=n-m\).
To parameterize the solutions of \(Ac=y\) for a fixed
\(y\in\mathbb F_q^m\), choose a permutation matrix \(\Pi\) such that
\begin{align}
A\Pi=[A_{\rm P}\;A_{\rm F}],
\qquad A_{\rm P}\in\mathbb F_q^{m\times m}\ \text{nonsingular}.
\end{align}
Define
\begin{align}
J&=\Pi\begin{pmatrix}A_{\rm P}^{-1}\\0_{d\times m}\end{pmatrix},&
B&=\Pi\begin{pmatrix}-A_{\rm P}^{-1}A_{\rm F}\\I_d\end{pmatrix}.
\end{align}
Then \(AJ=I_m\), the columns of \(B\) form a basis of \(\ker A\), and
\begin{align}
\{c:Ac=y\}=\{Jy+Bs:s\in\mathbb F_q^d\}.
\label{eq:supp-affine-parametrization}
\end{align}
The parametrization is one-to-one.  Thus a seed
\(S\sim\operatorname{Unif}(\mathbb F_q^d)\) makes \(Jy+BS\) uniform on the
full solution set \(\{c:Ac=y\}\).
\end{Lem}

\begin{IEEEproof}
Direct multiplication gives \(AJ=I_m\) and \(AB=0\).  In the permuted
coordinates, the free-coordinate block of \(B\) is \(I_d\), so its columns
are independent.  Rank-nullity gives \(\dim\ker A=d\), and hence those
columns form a basis.  Every solution differs from \(Jy\) by a unique kernel
vector.  Thus \eqref{eq:supp-affine-parametrization} is a bijection between
seed values and solutions, and a uniform seed gives a uniform solution.
\end{IEEEproof}

\subsection{\texorpdfstring{Short-block encoder for \(1\le b<L\)}{Short-block encoder for 1<= b<L}}

\begin{IEEEproof}[Proof of Theorem~\ref{thm:short-local-transform}]
Set \((L',b',\Delta)=(K-L,K-b,L-b)\).  Since
\(1\le L'<b'\le K-1\), the direct theorem supplies an auxiliary encoder.  Write its
global relation as \(\widetilde{\mathbf D}=\mathbf H\widetilde{\mathbf X}\),
where
\begin{align}
\mathbf H\in\mathbb F_q^{Kb'\times KL'}
\quad\text{has full column rank}.
\end{align}
Set \(A=\mathbf H^{\mathsf T}\), which has full row rank, and let
\begin{align}
d=Kb'-KL'=K(L-b).
\label{eq:supp-dual-nullity}
\end{align}
Construct the right inverse \(J\) and kernel basis \(B\) as in
Lemma~\ref{lem:supp-systematic-affine}.  Conditional on a source value
\(\mathbf V=v\), the placed vector
\(\mathbf C=JP_\Delta v+B\mathbf U\), with
\(\mathbf U\sim\operatorname{Unif}(\mathbb F_q^d)\), is uniform over
\begin{align}
\mathcal F(v)\triangleq
\{c\in\mathbb F_q^{Kb'}:Ac=P_\Delta v\},
\qquad |\mathcal F(v)|=q^d.
\end{align}
Each scalar group of each randomized block instance uses its own copy of
the seed vector \(\mathbf U\).  These vectors are mutually independent and
are sampled during placement, jointly independently of the file library,
parent randomness, and demand vector.

Fix \(u\) and order the rows as
\begin{align}
\mathcal R_u&=\mathcal T_u^{(b)}\oplus\Delta,&
\mathcal S_u&=(\mathcal T_u^{(b)})^{\rm c}\oplus\Delta,
\end{align}
and the columns as \((\mathcal A_u,\mathcal A_u^{\rm c})\).  The shift
convention gives
\begin{align}
(P_\Delta\mathbf V)_{\mathcal R_u}
 &=\mathbf V_{\mathcal T_u^{(b)}},&
(P_\Delta\mathbf V)_{\mathcal S_u}
 &=\mathbf V_{(\mathcal T_u^{(b)})^{\rm c}}.
\end{align}
By Lemma~\ref{lem:supp-complement-alignment}, the inaccessible columns are
indexed by the auxiliary access window \(\mathcal A'_{u+L}\).
The row set \(\mathcal R_u\) is the auxiliary target complement
\((\mathcal T'_{u+L})^{\rm c}\), and \(\mathcal S_u\) is the auxiliary target
window \(\mathcal T'_{u+L}\).  Transposing the auxiliary direct matrix's zero
and invertible submatrices gives
\begin{align}
A(\mathcal R_u,\mathcal A_u^{\rm c})&=0,&
A(\mathcal S_u,\mathcal A_u^{\rm c})&\ \text{is nonsingular}.
\end{align}
Therefore \(A\mathbf C=P_\Delta\mathbf V\) has the block form
\begin{align}
\begin{pmatrix}M_{11}&0\\M_{21}&M_{22}\end{pmatrix}
\begin{pmatrix}\mathbf C_{\mathcal A_u}\\
\mathbf C_{\mathcal A_u^{\rm c}}\end{pmatrix}
=
\begin{pmatrix}\mathbf V_{\mathcal T_u^{(b)}}\\
\mathbf V_{(\mathcal T_u^{(b)})^{\rm c}}\end{pmatrix},
\qquad M_{22}\ \text{nonsingular}.
\label{eq:supp-dual-block-triangular}
\end{align}
The block dimensions are
\begin{align*}
M_{11}&:bL'\times Lb',&
0&:bL'\times L'b',\\
M_{21}&:b'L'\times Lb',&
M_{22}&:b'L'\times L'b'=b'L'\times b'L'.
\end{align*}
The first block row recovers the target:
\begin{align}
\mathbf V_{\mathcal T_u^{(b)}}
=\mathbf H^{\mathsf T}
(\mathcal T_u^{(b)}\oplus\Delta,\mathcal A_u)
\mathbf C_{\mathcal A_u}.
\end{align}

To determine what the accessible caches reveal, fix a target value
\(v_{\mathcal T_u^{(b)}}\) and define the set of compatible accessible outputs:
\begin{align}
\mathcal G_u(v_{\mathcal T_u^{(b)}})
=\{c_{\mathcal A_u}:M_{11}c_{\mathcal A_u}
=v_{\mathcal T_u^{(b)}}\}.
\end{align}
Projecting \(\mathcal F(v)\) onto the accessible coordinates gives a
bijection onto this set.  The first row of
\eqref{eq:supp-dual-block-triangular} places every projected solution in the
set, and the second row uniquely determines the inaccessible coordinates:
\begin{align}
c_{\mathcal A_u^{\rm c}}
=M_{22}^{-1}
\bigl(v_{(\mathcal T_u^{(b)})^{\rm c}}-M_{21}c_{\mathcal A_u}\bigr).
\end{align}
For the following probability expression, abbreviate the fixed target and
complement values as \(v_T=v_{\mathcal T_u^{(b)}}\) and
\(v_{T^{\rm c}}=v_{(\mathcal T_u^{(b)})^{\rm c}}\).
Uniformity on the global solution set \(\mathcal F(v)\) therefore implies
\begin{align}
&\Pr\{\mathbf C_{\mathcal A_u}=c\mid
\mathbf V_{\mathcal T_u^{(b)}}=v_T,
\mathbf V_{(\mathcal T_u^{(b)})^{\rm c}}=v_{T^{\rm c}}\}
\notag\\
&\hspace{3em}=
\frac{\mathbf 1\{c\in\mathcal G_u(v_T)\}}
{|\mathcal G_u(v_T)|},
\end{align}
for every source realization in the support of an arbitrary joint source law.
The right-hand side depends only on \(v_T\), and therefore
\begin{align}
I\!\left(\mathbf V_{(\mathcal T_u^{(b)})^{\rm c}};
\mathbf C_{\mathcal A_u}\middle|
\mathbf V_{\mathcal T_u^{(b)}}\right)=0.
\end{align}
Each source position has \(L'=K-L\) coordinates
and each cache stores \(b'=K-b\), giving cost \((K-b)/(K-L)\).
Equation~\eqref{eq:supp-dual-nullity} gives the exact local seed dimension.
\end{IEEEproof}

\subsection{Expanded short-block check}

The case \((K,L,b)=(4,3,2)\) shows how the affine construction recovers target
symbols while hiding their complement.
The auxiliary parameters are
\((L',b',\Delta)=(1,2,1)\).  Stack the four cache pairs
\(\mathbf C_0,\ldots,\mathbf C_3\) in cyclic order, with the two coordinates
inside each pair ordered as \((c_{i,0},c_{i,1})\).  The transposed auxiliary
system is
\begin{align}
\underbrace{\begin{pmatrix}
1&0&0&0&0&0&0&1\\
0&1&1&0&0&0&0&0\\
0&0&0&1&1&0&0&0\\
0&0&0&0&0&1&1&0
\end{pmatrix}}_{A=\mathbf H^{\mathsf T}}
\mathbf C
=
\begin{pmatrix}V_3\\V_0\\V_1\\V_2\end{pmatrix}
=P_1\mathbf V.
\end{align}
The four rows have disjoint supports, so the matrix has rank four and nullity
four.  Choose the free coordinates
\(S_i=c_{i,1}\), with the \(S_i\) independent and uniform.  Solving each row gives
\begin{align}
\mathbf C_i=(V_{i-1}-S_{i-1},S_i)^{\mathsf T},
\qquad i\in\mathbb Z_4.
\end{align}
User \(u\) reads \(\mathbf C_u,\mathbf C_{u+1},\mathbf C_{u+2}\) and obtains
\begin{align*}
V_u=S_u+(V_u-S_u),
\qquad
V_{u+1}=S_{u+1}+(V_{u+1}-S_{u+1}).
\end{align*}
Conditioned on these two targets, \(V_{u+2}\) is absent from the accessible
window.  The symbol \(V_{u-1}\) appears as \(V_{u-1}-S_{u-1}\), and its
matching pad lies in the inaccessible cache.  Since that pad is independent
and uniform, the accessible view conditioned on the two targets is independent
of the complementary pair.

\section{Exact Run Counting, Achievable Family, and Parent-Preserving Allocation Optimality}
\label{app:allocation}

\subsection{Maximal runs and formal-block packing}

\begin{IEEEproof}[Proof of Lemma~\ref{lem:cycliccount}]
Fix a start \(a\in\mathbb Z_K\), an incidence size \(r\), and an appearance
span \(b\).  Let \(\mathcal R_{a,r,b}\) be the family of \(r\)-element
subsets of \(\mathbb Z_K\) having a maximal run of span \(b\) starting at
\(a\).  For \(1\le b\le r\le K-2\), a subset having
\(\{a,a+1,\ldots,a+b-1\}\) as a maximal cyclic run must contain those \(b\)
positions and exclude both neighbors \(a-1\) and \(a+b\).  Its remaining
\(r-b\) elements can be chosen freely from the other \(K-b-2\) positions.
Hence
\begin{align}
|\mathcal R_{a,r,b}|=\binom{K-b-2}{r-b}.
\end{align}
For \(r=K-1\), the subset is a single run of span \(K-1\).
For a fixed start \(a\), the subset is \(\mathbb Z_K\setminus\{a-1\}\), so
\(\Gamma(K,K-1,K-1)=1\).  Every remaining parameter choice produces zero
maximal runs, giving the stated values of \(\Gamma\).

Fix user \(u\) and a parent
symbol \(Y_{\mathcal S}\) with \(u\in\mathcal S\).  Since
\(|\mathcal S|\le K-1\), exactly one maximal run of \(\mathcal S\) contains
\(u\).  Let that run have start \(a\), span \(b\), and endpoint
\(p=a+b-1\).
The endpoint relation is
\begin{align}
u\in\{a,\ldots,p\}
\quad\Longleftrightarrow\quad
p\in\{u,\ldots,u+b-1\}=\mathcal T_u^{(b)}.
\end{align}
Let \(\rho\in[1:\Gamma(K,|\mathcal S|,b)]\) be the run's rank among this
family's occurrences with endpoint \(p\) and span \(b\).
This rank identifies its formal block, where
\(Y_{\mathcal S}\) occupies the endpoint position in user \(u\)'s target
collection.  Conversely, every position in that collection comes from a run
containing \(u\), so its indexing subset contains
\(u\).
Thus the target collection consists exactly of this family's named components
in \(\widetilde Z_u\).
A subset with several maximal runs contributes one formal position per run.
A fixed user \(u\in\mathcal S\) selects the position for the unique run of
\(\mathcal S\) containing that user.
Decoding the selected positions recovers each named component of the family in
\(\widetilde Z_u\) exactly once.
\end{IEEEproof}

\begin{Lem}[Exact number of complete block families]
For a subset-indexed incidence family of size \(r\), the number of complete
\(K\)-position formal blocks is
\begin{align}
B(K,r)\triangleq\sum_{b=1}^{K-1}\Gamma(K,r,b)
=
\begin{cases}
0,&r=0,\\
\binom{K-2}{r-1},&1\le r\le K-1.
\end{cases}
\end{align}
\end{Lem}

\begin{IEEEproof}
For fixed \((r,b)\), there are \(\Gamma(K,r,b)\) run occurrences at each
start and hence \(K\Gamma(K,r,b)\) occurrences around the cycle.  A complete
formal block has \(K\) endpoint positions, so endpoint packing produces
exactly \(\Gamma(K,r,b)\) blocks.  For \(1\le r\le K-2\), the hockey-stick
identity gives
\begin{align}
\sum_{b=1}^{r}\binom{K-b-2}{r-b}
=\sum_{j=0}^{r-1}\binom{K-r+j-2}{j}
=\binom{K-2}{r-1}.
\end{align}
The case \(r=K-1\) consists of the single boundary block, and \(r=0\) has
no proper run.
\end{IEEEproof}

\Needspace{6\baselineskip}
\subsection{The unit-rate parent realization}

To obtain a unit-rate endpoint, we use a parent that transmits one file's worth
of symbols for every demand vector.  Work over the chosen field, with all file
and key operations performed coordinatewise.
Index users by \([K]\).  In this construction, let \(D_1=0\) and
\(D_n=W_n-W_1\), \(n\in\{2,\ldots,N\}\), denote the file differences.
Draw mutually independent uniform keys \(S_j^{(i)}\), \(i\in[K-1]\),
\(j\in[N-1]\), from the same vector alphabet as a file.
Their joint collection is independent of the library.
Set \(S_0^{(i)}=0\) for each \(i\in[K-1]\).
For every \(j\in[N-1]\), form the \(K\)-position vector
\begin{align}
V_i^{(j)}&=S_j^{(i)},&&i\in[K-1],
\notag\\
V_K^{(j)}&=D_{j+1}+\sum_{i=1}^{K-1}S_j^{(i)}.
\end{align}
Parent user 1 holds positions \(1,\ldots,K-1\) of every vector.
Parent user \(i+1\), \(i\in[K-1]\), holds every position except \(i\).
Each position is a file-sized vector, so each parent cache stores
\((N-1)(K-1)\) file-sized vectors.  The \(N-1\) formal vectors give exactly
\(N-1\) appearance-span-\(K-1\) blocks.

For demand \(\mathbf d\), transmit
\begin{align}
X_{\mathbf d}^{\rm ep}
=W_{d_1}+\sum_{i=1}^{K-1}
\bigl(S_{d_1-1}^{(i)}-S_{d_{i+1}-1}^{(i)}\bigr).
\end{align}
User 1 recovers its file by subtracting all key differences.
For user \(i+1\), define the masked file differences \(P_n^{(i)}\) and the
broadcast after subtracting the available key differences, \(Y_i\), by
\begin{align}
P_1^{(i)}&=0,&
P_n^{(i)}&=D_n+S_{n-1}^{(i)},\quad n\ge2,
\notag\\
Y_i&=X_{\mathbf d}^{\rm ep}
-\sum_{\substack{h=1\\h\ne i}}^{K-1}
\bigl(S_{d_1-1}^{(h)}-S_{d_{h+1}-1}^{(h)}\bigr).
\end{align}
User \(i+1\)'s cache supplies every \(P_n^{(i)}\) and each key used in the
subtraction.  Substituting the broadcast into the decoder gives
\begin{align}
Y_i+P_{d_{i+1}}^{(i)}-P_{d_1}^{(i)}=W_{d_{i+1}},
\end{align}
so each user recovers its requested file.

For user 1, the cache and broadcast are functions of its requested file and
independent keys.  For user \(i+1\), the parent cache is in one-to-one correspondence
with the keys \(S_j^{(h)}\) having \(h\ne i\) and the masked differences
\((P_2^{(i)},\ldots,P_N^{(i)})\).  For any fixed library, the map from each
\(S_{n-1}^{(i)}\) to \(P_n^{(i)}\) is a translation by \(D_n\).
The entire tuple therefore has the same uniform distribution for every
library value and is independent of the library.
The broadcast after key subtraction is
\begin{align}
Y_i=W_{d_{i+1}}-P_{d_{i+1}}^{(i)}+P_{d_1}^{(i)}.
\end{align}
Conditioned on the requested file, both the cache and broadcast are therefore
functions of variables independent of the remaining files.  Thus
\begin{align}
I\!\left(W_{[N]\setminus\{d_k\}};
\widetilde Z_k^{\rm ep},X_{\mathbf d}^{\rm ep}
\middle|W_{d_k},\mathbf d\right)=0,
\qquad k\in[K].
\end{align}
The parent memory-rate pair is \(((N-1)(K-1),1)\).  Directly encoding its
\(N-1\) blocks gives the stated physical-cache endpoint.

\subsection{Achievable memory-rate family}

\begin{IEEEproof}[Proof of Theorem~\ref{thm:memoryload}]
By Lemma~\ref{lem:cycliccount}, the \(N\) file-share incidence families of
size \(t\) and the masking-key family of size \(t+1\) each split into the
counted formal blocks.
For each user \(u\), collecting the target references from all \(N\) share
families and the key family gives exactly \(\widetilde Z_u(t)\).
By Proposition~\ref{prop:inherit}, the lifted scheme uses the same parent message
at rate \(R(t)=K/(t+1)\).

Every parent share and masking key has normalized size
\(1/\binom{K-1}{t}\).  For each file, the incidence-size-\(t\) family produces
\(\Gamma(K,t,b)\) complete appearance-span-\(b\) blocks.  Since a block stores
\(\mu_{b,L}\) parent-symbol units per cache, the total
file-share contribution is
\begin{align}
\frac{N}{\binom{K-1}{t}}
\sum_{b=1}^{K-1}\mu_{b,L}\Gamma(K,t,b)
=\frac{N T_{\rm blk}(K,t,L)}{\binom{K-1}{t}}.
\end{align}
The masking-key family contributes
\(T_{\rm blk}(K,t+1,L)/\binom{K-1}{t}\).  Adding the two contributions gives
the stated formula for \(M(t)\), with the boundary count
\(\Gamma(K,K-1,K-1)=1\) when \(t=K-2\).

The unit-rate parent contains exactly \(N-1\) appearance-span-\(K-1\) blocks.
Because \(K-1\ge L\), each direct lift costs \((K-1)/L\) per cache and leaves
the parent transmission unchanged, realizing
\(((N-1)(K-1)/L,1)\).  Rational memory sharing is valid by
Lemma~\ref{lem:common-scalarization}, and every resulting code also fits in
larger caches.
\end{IEEEproof}

\subsection{Subpacketization, local randomness, and arithmetic costs}

\begin{IEEEproof}[Proof of Proposition~\ref{prop:implementation-cost}]
Choose a parent implementation over \(\mathbb F_Q\), where \(Q=p^e\),
\(p>K\), and \(\chi=\operatorname{lcm}(L,K-L)\) divides \(e\).
The field must also satisfy the parent's requirements.  Fix a public \(\mathbb F_p\)-basis of
\(\mathbb F_Q\) and a positive integer \(m\) satisfying the parent's
blocklength requirements.  Let \(F\) be the file length over \(\mathbb F_Q\)
and \(d\) the length of an expanded share or key over \(\mathbb F_p\), with
\begin{align}
C_t&=\binom{K-1}{t},&
\chi&=\operatorname{lcm}(L,K-L),&
F&=mC_t,&
d&=em.
\end{align}
Each file has \(F\) symbols over \(\mathbb F_Q\), and each parent file share
and masking key has \(m\) symbols over \(\mathbb F_Q\).
Expanding each share or key in this basis gives \(d\) coordinates over \(\mathbb F_p\),
so \(dC_t=eF\).  The file-share volume is
\(d\binom Kt\) prime-field coordinates per file, and the total masking-key
volume is \(d\binom K{t+1}\) prime-field coordinates.

For appearance span \(b\), the exact number of complete blocks is
\begin{align}
A_t(b)=N\Gamma(K,t,b)+\Gamma(K,t+1,b).
\end{align}
In the packing of Lemma~\ref{lem:cycliccount}, the \(K\) run starts supply
the \(K\) formal positions of a complete block.  There are
\(\Gamma(K,t,b)\) blocks per share family and \(\Gamma(K,t+1,b)\) blocks in
the key family.  Summing over the \(N\) share families and the key family gives
\(A_t(b)\).
The total number of complete blocks is
\begin{align}
\sum_{b=1}^{K-1}A_t(b)=NB(K,t)+B(K,t+1).
\end{align}

Since \(L\mid e\) and \(K-L\mid e\), the basis-coordinate length \(d=em\)
gives \(d/L\) groups per long block and \(d/(K-L)\) groups per short block
for every admissible parent length \(m\).  One encoder call per group gives
\begin{align}
G_{\rm dir}(t;d)&=\frac dL\sum_{b=L}^{K-1}A_t(b),
\notag\\
G_{\rm dual}(t;d)&=\frac d{K-L}\sum_{b=1}^{L-1}A_t(b).
\end{align}
Each long group stores \(b\) prime-field coordinates per cache.  Each short
group stores \(K-b\) prime-field coordinates per cache.
Summing the stored coordinates over the groups gives the fixed segment length
per physical cache:
\begin{align}
P_{\rm cache}(t;d)
&=\frac d{K-L}\sum_{b=1}^{L-1}(K-b)A_t(b)
+\frac dL\sum_{b=L}^{K-1}bA_t(b)
\notag\\*
&=d\bigl[NT_{\rm blk}(K,t,L)+T_{\rm blk}(K,t+1,L)\bigr].
\end{align}
With \(dC_t=eF\), the normalized allocation is
\begin{align}
\frac{P_{\rm cache}(t;d)}{eF}
&=\frac{NT_{\rm blk}(K,t,L)+T_{\rm blk}(K,t+1,L)}{C_t}
=M(t).
\end{align}
The prescribed cache alphabet is
\(\mathcal Z_i=\mathbb F_p^{P_{\rm cache}(t;d)}\).
Its exact length in parent-field units is
\begin{align}
\log_Q|\mathcal Z_i|
=P_{\rm cache}(t;d)\log_Qp
=\frac{P_{\rm cache}(t;d)}e
=FM(t),
\end{align}
which satisfies the physical-cache alphabet-size constraint.

For a short block put \(L'=K-L\) and \(b'=K-b\).  The affine encoder solves a
full-row-rank system over \(\mathbb F_p\), with \(Kb'\) cache coordinates and
\(KL'\) equations, so
its nullity is
\begin{align}
Kb'-KL'=K(L-b).
\end{align}
The affine encoder draws \(K(L-b)\) independent uniform
\(\mathbb F_p\) symbols per short scalar group.  The \(d/(K-L)\) groups in each
of the \(A_t(b)\) blocks therefore require a total of
\begin{align}
Q_{\rm loc}(t;d)
=\frac{dK}{K-L}
\sum_{b=1}^{L-1}(L-b)
\bigl[N\Gamma(K,t,b)+\Gamma(K,t+1,b)\bigr].
\label{eq:supp-total-local-randomness}
\end{align}
The direct encoder is deterministic.  The fresh prime-field seed length in
\eqref{eq:supp-total-local-randomness} scales linearly with \(m\).
The aggregate seed alphabet has base-\(Q\) length \(Q_{\rm loc}(t;d)/e\).

All arithmetic below is over \(\mathbb F_p\).
For one direct scalar group, \(\omega=b-L\), and each cache has
\(\omega+1\) sparse coefficient matrices with \(L\) nonzero entries each.
Evaluating the nonzero terms takes at most \(L(\omega+1)\) multiplications,
and summing them takes exactly \(L(\omega+1)-b\) additions across the
\(b\) outputs.  Across all caches, groups, and blocks,
\begin{align}
\mathsf{Mult}_{\rm dir}
&\le\frac dL\sum_{b=L}^{K-1}A_t(b)\,KL(b-L+1),
\notag\\
\mathsf{Add}_{\rm dir}
&\le\frac dL\sum_{b=L}^{K-1}A_t(b)\,
K\bigl[L(b-L+1)-b\bigr].
\end{align}

For a short span \(b\), write \(n_b=K(K-b)\).  Absorb the cyclic shift into
the precomputed right inverse and evaluate
\begin{align}
\mathbf C=
\begin{bmatrix}JP_\Delta&B\end{bmatrix}
\begin{bmatrix}\mathbf V\\\mathbf U\end{bmatrix}.
\end{align}
The vectors \(\mathbf V\) and \(\mathbf U\) together contain
\(K(K-L)+K(L-b)=n_b\) coordinates, and \(\mathbf C\) also has \(n_b\) coordinates.
Evaluating the dense matrix-vector product therefore uses at most \(n_b^2\) multiplications and
\(n_b(n_b-1)\) additions per group.
Across all short blocks, these per-group counts give
\begin{align}
\mathsf{Mult}_{\rm dual}
&\le\frac d{K-L}\sum_{b=1}^{L-1}A_t(b)[K(K-b)]^2,
\notag\\
\mathsf{Add}_{\rm dual}
&\le\frac d{K-L}\sum_{b=1}^{L-1}A_t(b)
K(K-b)[K(K-b)-1].
\end{align}

For the short encoder, put \(m_0=K(K-L)\).  The full-row-rank auxiliary
transpose has size \(m_0\times n_b\).  Precompute its right inverse and
kernel basis for use in every group of span \(b\).
Classical Gauss-Jordan elimination with pivot bookkeeping constructs a right
inverse and kernel basis in \(O(m_0^2n_b)\) operations over \(\mathbb F_p\).
The two matrices contain
\(n_bm_0+n_b(n_b-m_0)=n_b^2\) elements of \(\mathbb F_p\).
They depend only on \((K,L,b)\) over the chosen prime field, so the same
matrices can be reused across files, block families, and scalar groups.
The Pascal coefficients of a long span can likewise be generated
once and reused.
The operation counts above cover the local placement encoder applied to the parent
shares and masking keys expressed in the fixed public basis.
\end{IEEEproof}

At \(b=L\), the direct encoder achieves unit local cost with zero fresh local
randomness.
When \(L=1\), every proper block belongs to
the direct branch.

\subsection{Exact allocation optimality for parent-preserving block-separable lifting}

\begin{Lem}[Exact allocation by field-basis scalarization]
\label{lem:exact-basis-scalarization}
Fix $K\ge2$ and $1\le L<K$, let
$\chi=\operatorname{lcm}(L,K-L)$, and choose a prime $p>K$.
Over $\mathbb F_Q$, where $Q=p^\chi$, every proper span $b$ and every
input length $h\ge1$ admit a valid local encoder with
\begin{align}
V_i\in\mathbb F_Q^h,
\qquad
C_i\in\mathbb F_p^{\,\chi h\mu_{b,L}},
\qquad
\frac1h\log_Q|\mathcal C_i|=\mu_{b,L}.
\end{align}
For every joint input law, the encoder gives exact recovery and conditional
secrecy.  Every named reference uses the same basis representation.
\end{Lem}

\begin{IEEEproof}
Fix an $\mathbb F_p$-basis $(e_1,\ldots,e_\chi)$ of $\mathbb F_Q$.
The coordinate map
\begin{align}
\phi:\mathbb F_Q^h&\longrightarrow\mathbb F_p^{\chi h},\\
(v_1,\ldots,v_h)&\longmapsto
(v_{1,1},\ldots,v_{1,\chi},\ldots,v_{h,1},\ldots,v_{h,\chi}),\\
v_j&=\sum_{r=1}^{\chi}v_{j,r}e_r
\end{align}
is a fixed bijection.  Both $L$ and $K-L$ divide $\chi h$.
For $b\ge L$, partition $\phi(V_i)$ into $\chi h/L$ groups of $L$
coordinates and apply Theorem~\ref{thm:direct-local-transform} over
$\mathbb F_p$ to each group.  Each cache receives
$(\chi h/L)b=\chi h\mu_{b,L}$ base-field coordinates.
For $b<L$, partition $\phi(V_i)$ into $\chi h/(K-L)$ groups of $K-L$ coordinates
and apply Theorem~\ref{thm:short-local-transform} over $\mathbb F_p$,
using mutually independent placement-time seeds whose joint vector is
independent of the source.  Each cache receives
$(\chi h/(K-L))(K-b)=\chi h\mu_{b,L}$ base-field coordinates.

Decode the target coordinates in each group, then apply $\phi^{-1}$ to
recover the original target symbols.  Conditioned on a complete source value $v$,
each group's accessible-output law depends only on its target coordinates.
Conditioning on $v$ leaves the placement-time seeds mutually independent and
does not change their distributions.
The joint accessible-output law is therefore the product of the group-output
laws and depends on $v$ only through
$v_{\mathcal T_u^{(b)}}$.  Hence conditional secrecy holds under every joint
source law.  Every occurrence of a named
symbol uses the same fixed coordinate map.

The segment alphabet has cardinality
$p^{\chi h\mu_{b,L}}$.  Since $Q=p^\chi$, its normalized allocation is
\begin{align}
\frac{\log_Q|\mathcal C_i|}{h}
=\frac{\chi h\mu_{b,L}\log_Q p}{h}
=\mu_{b,L}.
\end{align}
Thus the original $h$ symbols of $\mathbb F_Q$ are encoded exactly through
$\phi$.
\end{IEEEproof}

\begin{Lem}[Full-cache recovery for one formal block]
\label{lem:supp-all-caches}
If every user reconstructs its nonempty target window from its access window,
then the full physical-cache collection reconstructs all \(K\) formal source
positions.
\end{Lem}

\begin{IEEEproof}
For every \(\ell\in\mathbb Z_K\), the target window
\(\mathcal T_\ell^{(b)}\) contains \(\ell\), and its decoder recovers
\(V_\ell\) from \(C_{\mathcal A_\ell}\), a subcollection of all caches.
The full cache collection therefore recovers every formal source position.
\end{IEEEproof}

\begin{IEEEproof}[Proof of Theorem~\ref{thm:framework}]
Fix one appearance-span-\(b\) block, and let \(Z_i^{(b)}\) be its segment in
cache \(i\).  Evaluate the bounds under the independent-input distribution
specified in the theorem.  Each formal position has entropy \(h\), measured
with base-\(Q\) logarithms.  In the rest of this proof, divide all entropies
by \(h\), and write \(a_i^{(b)}=h^{-1}\log_Q|\mathcal C_i^{(b)}|\) for the
corresponding segment allocation.  Thus, in these normalized units,
\begin{align}
H(Z_i^{(b)})\le a_i^{(b)},
\qquad
\bar a_b\triangleq\frac1K\sum_{i=0}^{K-1}a_i^{(b)}.
\end{align}

For \(b\ge L\), the formal inputs are mutually independent and have unit
entropy in these units.  Exact target recovery gives, for every user,
\begin{align}
H(Z_{\mathcal A_u}^{(b)})
\ge I(V_{\mathcal T_u^{(b)}};Z_{\mathcal A_u}^{(b)})
=H(V_{\mathcal T_u^{(b)}})=b.
\end{align}
Each cache belongs to exactly \(L\) access windows.  Summing over users and
applying subadditivity gives
\begin{align}
Kb
&\le\sum_u H(Z_{\mathcal A_u}^{(b)})
\le L\sum_iH(Z_i^{(b)})
\le L\sum_i a_i^{(b)}
=LK\bar a_b.
\end{align}
Thus \(\bar a_b\ge b/L\).

For \(b<L\), the chain rule, input independence, and local conditional secrecy give
\begin{align}
I(V_{(\mathcal T_u^{(b)})^{\rm c}};Z_{\mathcal A_u}^{(b)})
&\le I(V_{(\mathcal T_u^{(b)})^{\rm c}};
V_{\mathcal T_u^{(b)}},Z_{\mathcal A_u}^{(b)})\notag\\
&=0.
\end{align}
Consequently,
\begin{align}
H(V_{(\mathcal T_u^{(b)})^{\rm c}}\mid Z_{\mathcal A_u}^{(b)})=K-b.
\end{align}
Lemma~\ref{lem:supp-all-caches} gives zero conditional entropy for the same
complement given the full cache collection.  Consequently,
\begin{align}
H(Z_{\mathcal A_u^{\rm c}}^{(b)})
&\ge H(Z_{\mathcal A_u^{\rm c}}^{(b)}\mid Z_{\mathcal A_u}^{(b)})
\notag\\
&\ge I(V_{(\mathcal T_u^{(b)})^{\rm c}};
Z_{\mathcal A_u^{\rm c}}^{(b)}\mid Z_{\mathcal A_u}^{(b)})
=K-b.
\end{align}
Every cache is outside exactly \(K-L\) access windows.  Summing the bound over
users and applying subadditivity gives
\begin{align}
K(K-b)
\le\sum_uH(Z_{\mathcal A_u^{\rm c}}^{(b)})
\le(K-L)\sum_iH(Z_i^{(b)})
\le K(K-L)\bar a_b.
\end{align}
Hence \(\bar a_b\ge(K-b)/(K-L)\).

Lemma~\ref{lem:exact-basis-scalarization} attains $a_i=\mu_{b,L}$
in every cache, for every positive integer $h$, over the same field
$\mathbb F_Q$.  Since the lower bound applies to every segment alphabet,
this is the exact minimum worst-cache allocation.

Let $h_\beta$ be the common length of the formal positions in block $\beta$.
The total allocation normalized by file length $F$ obeys
\begin{align}
\max_i\sum_\beta m_i^{(\beta)}
&\ge\frac1K\sum_i\sum_\beta m_i^{(\beta)}
\ge\sum_\beta\frac{h_\beta}{F}\mu_{b_\beta,L}.
\end{align}
Realize the subset-indexed parent over $\mathbb F_Q$
with $h_\beta=m$ and $F=mC_t$.  The same basis expansion for all blocks
stores $\chi m\mu_{b_\beta,L}$ base-field coordinates in each cache.
Their base-$Q$ allocation is exactly $m\mu_{b_\beta,L}$.
Every parent-cache symbol is recovered in its original field representation,
so the parent delivery message and rate are unchanged.
Summing the $A_t(b)$ blocks of each span gives
\begin{align}
\frac1{C_t}\sum_{b=1}^{K-1}A_t(b)\mu_{b,L}=M(t).
\end{align}
Thus $M(t)$ is the exact minimum normalized worst-cache allocation for the
fixed parent and block list in the parent-preserving block-separable class.
\end{IEEEproof}

\section{Proofs of the General Converse and Rate Comparisons}
\label{app:order-opt}
Theorem~\ref{thm:macc-converse} bounds the delivery rate using the overlap
between cyclic access windows.  Lemma~\ref{lem:feasible-memory} bounds the
cache memory needed for recovery under file secrecy.  For $N\ge K$, these
bounds and the cyclic-run memory estimate prove
Theorem~\ref{thm:order-opt} for every feasible memory.  For $K>N$, the
memory estimates prove Proposition~\ref{prop:order-opt-small-library} on
its stated intervals.  For $L=2$, the corner formulas and a midpoint
estimate prove Proposition~\ref{prop:L2-order}.

\subsection{Rate and Memory Converses}
The proof of Theorem~\ref{thm:macc-converse} follows the cut-set argument in
\cite[Section~V-D]{ImprovedLowerBoundsForMultiAccessCodedCaching}, using the
conditional secrecy and fixed-length resource definitions of our model.

\begin{IEEEproof}[Proof of Theorem~\ref{thm:macc-converse}]
Fix a finite code.  In the finite-code calculation, $M$ and $R$ denote its
actual normalized resources $M_F$ and $R_F$.
Let $s$ be the number of consecutive users considered and $q_s$ the number of
disjoint groups of requested files.  Choose
\begin{align}
s\in\{1,2,\ldots,\min\{\lfloor N/2\rfloor,K\}\},
\qquad
q_s\triangleq\left\lfloor\frac Ns\right\rfloor.
\label{eq:converse-s-q-choice}
\end{align}
The range of $s$ in \eqref{eq:converse-s-q-choice} ensures $q_s\ge2$.  The
$s$ consecutive users in $[s]$ have access-set union and reference access set
\begin{align}
\mathcal P_s\triangleq\bigcup_{k=1}^{s}\mathcal A_k,
\qquad
\mathcal A\triangleq\mathcal A_1.
\label{eq:converse-window-union}
\end{align}
The cyclic access pattern gives
\begin{align}
|\mathcal P_s\backslash\mathcal A|
=\min\{s-1,K-L\}.
\label{eq:converse-extra-cache-count}
\end{align}
In \eqref{eq:converse-window-union}, each user after user $1$ adds one cache
to $\mathcal A$ until the union contains all $K$ caches.  The number added
is therefore capped by the $K-L$ caches outside $\mathcal A$.

For each $\ell\in[q_s]$, choose a demand vector satisfying
\begin{align}
d_k^{(\ell)}=(\ell-1)s+k,
\qquad k\in[s].
\label{eq:converse-demand-family}
\end{align}
The remaining $K-s$ users have arbitrary valid demands.  Since $sq_s\le N$,
every requested index in \eqref{eq:converse-demand-family} belongs to $[N]$.
Write the broadcasts as
\begin{align}
X_\ell\triangleq X_{\mathbf d^{(\ell)}},
\qquad \ell\in[q_s].
\label{eq:converse-broadcast-abbreviation}
\end{align}
For a fixed $\ell$, every user $k\in[s]$ recovers $W_{(\ell-1)s+k}$ from
$(Z_{\mathcal A_k},X_\ell)$.  Since $\mathcal A_k\subseteq\mathcal P_s$,
correctness gives
\begin{align}
H(W_{(\ell-1)s+1:\ell s}\mid Z_{\mathcal P_s},X_\ell)=0,
\qquad \ell\in[q_s].
\label{eq:converse-demand-block-decode}
\end{align}
The $q_s$ decoding relations in \eqref{eq:converse-demand-block-decode}
jointly give
\begin{align}
H(W_{1:sq_s}\mid Z_{\mathcal P_s},X_{1:q_s})=0.
\label{eq:converse-all-demand-decode}
\end{align}

Under $\mathbf d^{(1)}$, user $1$ requests $W_1$. The conditional security condition for that user gives
\begin{align}
I(W_{2:N};Z_{\mathcal A},X_1\mid W_1,\mathbf d^{(1)})=0.
\label{eq:converse-reference-security}
\end{align}
The demand vector is fixed.  Removing it from the conditioning and using
file independence gives
\begin{align*}
I(W_{2:N};Z_{\mathcal A},X_1)
&\le I(W_{2:N};W_1,Z_{\mathcal A},X_1)\\
&=I(W_{2:N};W_1)
  +I(W_{2:N};Z_{\mathcal A},X_1\mid W_1)\\
&=0.
\end{align*}
In particular, $I(W_{2:sq_s};Z_{\mathcal A},X_1)=0$.  All $q_s$ broadcasts
and the caches in $\mathcal P_s\setminus\mathcal A$, together with the
reference window, recover $sq_s$ files.  The reference observation
$(Z_{\mathcal A},X_1)$ contains no information about $W_{2:sq_s}$.
The added caches and broadcasts must therefore supply the remaining
$(sq_s-1)F$ units of file entropy.  By \eqref{eq:converse-all-demand-decode},
\begin{subequations}
\begin{align}
(sq_s-1)F
&=H(W_{2:sq_s})
\label{eq:converse-entropy-file-count}\\
&=H(W_{2:sq_s}\mid Z_{\mathcal A},X_1)
\label{eq:converse-entropy-security}\\
&=I(W_{2:sq_s};Z_{\mathcal P_s\backslash\mathcal A},X_{2:q_s}
\mid Z_{\mathcal A},X_1)
\label{eq:converse-entropy-decode}\\
&\le H(Z_{\mathcal P_s\backslash\mathcal A},X_{2:q_s})
\label{eq:converse-entropy-observation}\\
&\le \min\{s-1,K-L\}MF+(q_s-1)RF.
\label{eq:converse-entropy-memory-rate}
\end{align}
\end{subequations}
The last inequality follows from subadditivity,
\eqref{eq:converse-extra-cache-count}, and the fixed-length storage and
delivery constraints.

Dividing \eqref{eq:converse-entropy-memory-rate} by $(q_s-1)F$ gives
\begin{align}
R\ge
\frac{s\lfloor N/s\rfloor-1-\min\{s-1,K-L\}M}
{\lfloor N/s\rfloor-1}.
\label{eq:converse-single-s-bound}
\end{align}

For an achievable target pair, return to the sequence notation.
For each $\varepsilon>0$, the limsup constraints give
$M_{F_j}\le M+\varepsilon$ and $R_{F_j}\le R+\varepsilon$ for all
sufficiently large $j$.  Applying the finite-code inequality to this
sequence and letting $\varepsilon\downarrow0$ gives the bound for the
target pair.
Bound \eqref{eq:converse-single-s-bound} holds for each admissible $s$ in
\eqref{eq:converse-s-q-choice}, with $q_s-1>0$, and for every achievable
rate $R$.  Maximizing over $s$ and taking the infimum over achievable $R$
gives \eqref{eq:macc-th6}.
\end{IEEEproof}

\begin{IEEEproof}[Proof of Lemma~\ref{lem:feasible-memory}]
Fix a finite code, and write $M$ for its actual normalized cache memory
$M_F$ in the calculation below.
To bound the memory, choose one demand in which as many users as possible
request distinct files.  Let the number of distinct requested files be
\begin{align}
r\triangleq\min\{N,K\},
\label{eq:feasible-distinct-file-count}
\end{align}
and choose the demand vector
\begin{align}
\mathbf d=(1,2,\ldots,r,r,\ldots,r),
\label{eq:feasible-demand-choice}
\end{align}
where the repeated entries are absent when $r=K$.  The first $r$ users in
\eqref{eq:feasible-demand-choice} request the $r$ distinct files $W_{1:r}$.
Relative to user $1$, define the union of their access windows and the
caches outside the reference window by
\begin{align}
\mathcal P_r\triangleq\bigcup_{k=1}^{r}\mathcal A_k,
\qquad
\mathcal C_r\triangleq\mathcal P_r\backslash\mathcal A_1.
\label{eq:feasible-access-partition}
\end{align}
The cyclic topology gives $\mathcal A_1\subseteq\mathcal P_r$ and
$|\mathcal C_r|\le\min\{r-1,K-L\}\le K-L$.
For each user $k\in[r]$, correctness under $\mathbf d$ makes $W_k$ a
deterministic function of $(Z_{\mathcal A_k},X_{\mathbf d})$.
Since $\mathcal A_k\subseteq\mathcal P_r$, the combined cache view and
broadcast recover all $r$ requested files:
\begin{align}
H(W_{1:r}\mid Z_{\mathcal P_r},X_{\mathbf d})=0.
\label{eq:feasible-full-decode}
\end{align}
Conditional security for user $1$, whose requested file is $W_1$, gives
\begin{align}
I(W_{2:N};Z_{\mathcal A_1},X_{\mathbf d}\mid W_1,\mathbf d)=0.
\label{eq:feasible-user-security}
\end{align}
Fixing $\mathbf d$ and using file independence converts
\eqref{eq:feasible-user-security} into
$I(W_{2:N};Z_{\mathcal A_1},X_{\mathbf d})=0$.  Marginalizing the files
outside $W_{2:r}$ gives
\begin{align}
I(W_{2:r};Z_{\mathcal A_1},X_{\mathbf d})=0.
\label{eq:feasible-requested-security}
\end{align}
The observation $(Z_{\mathcal P_r},X_{\mathbf d})$ recovers all $r$ requested
files.  Its subcollection $(Z_{\mathcal A_1},X_{\mathbf d})$ satisfies
\eqref{eq:feasible-requested-security}.  Splitting the mutual information
along $\mathcal P_r=\mathcal A_1\mathbin{\dot\cup}\mathcal C_r$ gives
\begin{subequations}
\begin{align}
rF
&=I(W_{1:r};Z_{\mathcal P_r},X_{\mathbf d})
\label{eq:feasible-chain-decode}\\
&=I(W_{1:r};Z_{\mathcal A_1},X_{\mathbf d})
+I(W_{1:r};Z_{\mathcal C_r}\mid Z_{\mathcal A_1},X_{\mathbf d})
\label{eq:feasible-chain-partition}\\
&=I(W_1;Z_{\mathcal A_1},X_{\mathbf d}\mid W_{2:r})
+I(W_{1:r};Z_{\mathcal C_r}\mid Z_{\mathcal A_1},X_{\mathbf d})
\label{eq:feasible-chain-security}\\
&\le H(W_1)+H(Z_{\mathcal C_r})
\label{eq:feasible-chain-entropy}\\
&\le F+(K-L)MF.
\label{eq:feasible-chain-memory}
\end{align}
\end{subequations}
The entropy of $W_{1:r}$ is bounded by the entropy of one file and the caches in
$\mathcal C_r$.  Since $|\mathcal C_r|\le K-L$ and $H(Z_i)\le MF$, their
total entropy is at most $F+(K-L)MF$.  Dividing by the file length gives
$r\le1+(K-L)M$.  Substitution of \eqref{eq:feasible-distinct-file-count}
yields \eqref{eq:feasible-memory}.
For an achievable target pair and each $\varepsilon>0$, the sequence
satisfies $M_{F_j}\le M+\varepsilon$ for all sufficiently large $j$.
Applying the finite-code bound and letting $\varepsilon\downarrow0$
gives the same threshold for the target memory.
\end{IEEEproof}

\subsection{\texorpdfstring{Approximation Guarantee When $N\ge K$}{Approximation Guarantee When N>= K}}
For $N\ge K$, define the reference memory interval by
\begin{align}
\mathcal M_{\mathrm{ref}}
\triangleq
\left[\frac{K-1}{K-L},\frac{(N-1)(K-1)}{L}\right].
\label{eq:order-reference-domain}
\end{align}
For $N=K$ and $L=K-1$, this interval consists of the point $M=K-1$.

By \eqref{eq:left-endpoint}, the achievable envelope obeys
\begin{align}
R_{\mathrm{lift}}(M)\le K,\qquad \forall M\in\mathcal M_{\mathrm{ref}}.
\label{eq:order-trivial}
\end{align}

To compare the achievable rate and converse at a given memory budget, we
first bound the memory of each lifted corner using the run count.  Two
floor-selection lemmas then identify a corner within the budget and a
converse term at the same memory.
These three lemmas hold for all integers $N,K\ge2$ and $1\le L<K$, subject to
their stated hypotheses.

\begin{Lem}[Cyclic-run cost bound]\label{lem:order-run-cost}
Let $m_0\triangleq\min\{L,K-L\}$. Then
\begin{align}
T_{\mathrm{blk}}(K,0,L)&=0,\nonumber\\
\sum_{b=1}^{K-1}\Gamma(K,r,b)&=\binom{K-2}{r-1},
\qquad 1\le r\le K-1,
\label{eq:order-gamma-t}\\
T_{\mathrm{blk}}(K,r,L)&\le
\frac{K-1}{m_0}\binom{K-2}{r-1},
\qquad 1\le r\le K-1.
\label{eq:order-run-cost}
\end{align}
\end{Lem}
\begin{IEEEproof}
Both branches of \eqref{eq:mu} obey
\begin{align}
\mu_{b,L}\le\frac{K-1}{m_0},
\qquad 1\le b\le K-1.
\label{eq:order-mu-uniform-bound}
\end{align}
For $b<L$, the bound follows from $K-b\le K-1$ and $K-L\ge m_0$.  For
$b\ge L$, it follows from $b\le K-1$ and $L\ge m_0$.  Hence
\begin{align}
T_{\mathrm{blk}}(K,r,L)
\le\frac{K-1}{m_0}\sum_{b=1}^{K-1}\Gamma(K,r,b),
\qquad \forall r\in\{0,1,\ldots,K-1\}.
\label{eq:order-T-upper}
\end{align}
Lemma~\ref{lem:cycliccount} gives the zero case and, for $1\le r\le K-1$,
the binomial sum in \eqref{eq:order-gamma-t}, which counts the complete
$K$-position formal block families over all appearance spans.
The uniform local-cost bound in \eqref{eq:order-T-upper} yields
the aggregate cost \eqref{eq:order-run-cost}.
\end{IEEEproof}

\begin{Lem}[Floor-selected achievable envelope]\label{lem:floor-selected-envelope}
Let $a>0$, and let $T\ge0$ be an integer.  Suppose the points indexed by
$t\in\{0,\ldots,T\}$ belong to the achievable envelope $R_{\mathrm{lift}}$,
have rate $K/(t+1)$, and satisfy $M(t)\le a(t+1)$.  Suppose a rate-$K$
anchor belonging to $R_{\mathrm{lift}}$ is available at budget $M$ by memory
monotonicity.  If
$0<M<a(T+2)$, then
\begin{align}
R_{\mathrm{lift}}(M)\le\frac{2Ka}{M}.
\label{eq:floor-selected-envelope}
\end{align}
\end{Lem}
\begin{IEEEproof}
If $M/a<1$, the rate-$K$ anchor gives
$R_{\mathrm{lift}}(M)\le K<2Ka/M$.  For $M/a\ge1$, the floor-selected indices
are $u=\lfloor M/a\rfloor$ and $t=u-1$.  The assumed upper range gives
$1\le u\le T+1$, and
\begin{align}
M(t)\le au\le M,
\qquad
R_{\mathrm{lift}}(M)\le\frac Ku,
\qquad
u\ge\frac{M}{2a}.
\label{eq:floor-selected-core}
\end{align}
The corner fits within budget $M$ and achieves rate $K/u$.
The estimate $u\ge M/(2a)$ gives
\eqref{eq:floor-selected-envelope}.
\end{IEEEproof}

\begin{Lem}[Floor-selected converse]\label{lem:floor-selected-converse}
Assume $1\le M\le N/2$, and set
\begin{align}
s=\left\lfloor\frac{N}{2M}\right\rfloor,
\qquad
g=\left\lfloor\frac Ns\right\rfloor.
\label{eq:floor-converse-indices}
\end{align}
If $s$ is admissible in \eqref{eq:R-conv} and $s\le K-L+1$, then
\begin{align}
R_{\mathrm{conv}}(M)&\ge\frac{s}{2}\ge\frac{N}{8M},
\label{eq:floor-converse-general}\\
R_{\mathrm{conv}}(M)&\ge\frac{N}{6M},
\qquad M\le N/4.
\label{eq:floor-converse-sharp}
\end{align}
\end{Lem}
\begin{IEEEproof}
The condition $s\le K-L+1$ gives $\min\{s-1,K-L\}=s-1$.
Also, $s\le N/(2M)$ implies $g\ge2M-1$, so $M\le(g+1)/2$.
The selected term of \eqref{eq:R-conv} therefore satisfies
\begin{align}
R_{\mathrm{conv}}(M)
\ge\frac{sg-1-(s-1)M}{g-1}
=\frac{s(g-M)+M-1}{g-1}
\ge\frac{s}{2}.
\label{eq:floor-converse-core}
\end{align}
Apply the elementary floor bounds with $x=N/(2M)$.  For $N/(2M)\ge1$,
$\lfloor x\rfloor\ge x/2$ gives $s\ge N/(4M)$.  For $N/(2M)\ge2$,
the sharper $\lfloor x\rfloor\ge2x/3$ gives $s\ge N/(3M)$.
The first estimate yields \eqref{eq:floor-converse-general}, and the
second yields \eqref{eq:floor-converse-sharp}.
\end{IEEEproof}

\begin{Lem}\label{lem:order-memory}
Assume $N\ge K$. Then, for every $t\in\{0,1,\ldots,K-2\}$,
\begin{align}
M(t)\le \frac{N(t+1)}{\min\{L,K-L\}}.
\label{eq:order-memory}
\end{align}
\end{Lem}
\begin{IEEEproof}
With $m_0=\min\{L,K-L\}$, Lemma~\ref{lem:order-run-cost} bounds the cost
of each block family.  At $t=0$, the file-share term vanishes:
$T_{\mathrm{blk}}(K,0,L)=0$ and
\begin{align}
T_{\mathrm{blk}}(K,1,L)\le \frac{K-1}{m_0}.
\label{eq:order-t0-key-block-bound}
\end{align}
Consequently,
\begin{align}
M(0)
=\frac{N\,T_{\mathrm{blk}}(K,0,L)+T_{\mathrm{blk}}(K,1,L)}{\binom{K-1}{0}}
\le \frac{K-1}{m_0}
\le \frac{N}{m_0},
\label{eq:order-t0-memory-chain}
\end{align}
where the last step uses $K-1\le N$.

For $1\le t\le K-2$, including the boundary $t=K-2$, the run-count
identity at $r=t+1$ reads
\begin{align}
\sum_{b=1}^{K-1}\Gamma(K,t+1,b)=\binom{K-2}{t},
\label{eq:order-gamma-t1}
\end{align}
and applying Lemma~\ref{lem:order-run-cost} at $r=t$ and $r=t+1$ gives
\begin{subequations}
\begin{align}
T_{\mathrm{blk}}(K,t,L)&\le \frac{K-1}{m_0}\binom{K-2}{t-1},
\label{eq:order-T-t-upper}\\
T_{\mathrm{blk}}(K,t+1,L)&\le \frac{K-1}{m_0}\binom{K-2}{t}.
\label{eq:order-T-tplusone-upper}
\end{align}
\end{subequations}
The binomial ratios simplify as
\begin{subequations}
\begin{align}
\frac{\binom{K-2}{t-1}}{\binom{K-1}{t}}
&=\frac{t}{K-1},
\label{eq:order-memory-ratio-file}\\
\frac{\binom{K-2}{t}}{\binom{K-1}{t}}
&=\frac{K-1-t}{K-1}.
\label{eq:order-memory-ratio-key}
\end{align}
\end{subequations}
Substitution in \eqref{eq:M-t} gives
\begin{subequations}
\begin{align}
M(t)
&\le \frac{K-1}{m_0}\cdot
\frac{N\binom{K-2}{t-1}+\binom{K-2}{t}}{\binom{K-1}{t}}
\label{eq:order-memory-substitute}\\
&= \frac{Nt+K-1-t}{m_0}
\label{eq:order-memory-simplify}\\
&\le \frac{N(t+1)}{m_0},
\label{eq:order-memory-final}
\end{align}
\end{subequations}
Since $K-1-t\le N$, the bound $(Nt+K-1-t)/m_0$ is at most
$N(t+1)/m_0$.  Together with the $t=0$ case, this proves
\eqref{eq:order-memory}.
\end{IEEEproof}

By \eqref{eq:order-memory}, every corner $t$ satisfying
$N(t+1)/m_0\le M$ fits within the budget.  For $m_0M/N\ge1$, the floor
choice selects the largest such index and achieves its rate at budget $M$.

\begin{Lem}\label{lem:order-envelope}
Assume $N\ge K$. Then, for every $M\in\mathcal M_{\mathrm{ref}}$,
\begin{align}
R_{\mathrm{lift}}(M)\le \frac{2KN}{\min\{L,K-L\}\,M}.
\label{eq:order-envelope}
\end{align}
\end{Lem}
\begin{IEEEproof}
With $m_0=\min\{L,K-L\}$, apply Lemma~\ref{lem:floor-selected-envelope} at
$a=N/m_0$ and $T=K-2$.  Lemma~\ref{lem:order-memory} gives the required
memory bound, and \eqref{eq:order-trivial} gives the rate-$K$ anchor.
For $M\in\mathcal M_{\mathrm{ref}}$, the inequality $m_0\le L$ gives
\begin{align}
\frac{m_0M}{N}
\le \frac{m_0}{L}\frac{(N-1)(K-1)}{N}
< K-1.
\label{eq:order-envelope-u-range-chain}
\end{align}
In particular, $0<M<a(T+2)$.  The lemma therefore gives
\eqref{eq:order-envelope}.
\end{IEEEproof}

\begin{Propo}\label{prop:order-gap-interval}
Assume $N\ge K$. Then for every $M\in\mathcal M_{\mathrm{ref}}$,
\begin{align}
\frac{R_{\mathrm{lift}}(M)}{R_{\mathrm{conv}}(M)}
\le
\frac{12}{\min\left\{\frac{L}{K},\,1-\frac{L}{K}\right\}}.
\label{eq:order-gap-interval-proof}
\end{align}
\end{Propo}
\Needspace{7\baselineskip}
\begin{IEEEproof}

\emph{1) Small memory:} Assume
\begin{align}
\frac{K-1}{K-L}\le M\le \frac{N}{2(K-L+1)}.
\label{eq:order-regime-C}
\end{align}
The small-memory converse indices are
\begin{align}
s\triangleq \left\lfloor \frac{K-L+1}{2}\right\rfloor
=\left\lceil \frac{K-L}{2}\right\rceil,
\qquad
g\triangleq \left\lfloor \frac{N}{s}\right\rfloor.
\label{eq:order-C-sq-definition}
\end{align}
The assumed small-memory range, $N\ge K$, and
$s\le\lfloor K/2\rfloor$ make $s$ admissible in \eqref{eq:R-conv}.
The definition of $s$ also gives
\begin{align}
\min\{s-1,K-L\}=s-1.
\label{eq:order-C-min-evaluation}
\end{align}
The upper memory bound in \eqref{eq:order-regime-C} and the definition of
$s$ give
\begin{align}
g\ge \frac{N}{s}-1
\ge \frac{2N}{K-L+1}-1
\ge 4M-1,
\label{eq:order-C-q-chain}
\end{align}
or equivalently
\begin{align}
M\le \frac{g+1}{4}.
\label{eq:order-q-quarter}
\end{align}
The selected converse term is therefore
\begin{align}
R_{\mathrm{conv}}(M)\ge \frac{sg-1-(s-1)M}{g-1}.
\label{eq:order-C-conv-substitute}
\end{align}
Admissibility gives $g\ge2$, so the denominator is positive.
Equation~\eqref{eq:order-C-q-chain} gives the further bound
$g\ge4M-1\ge3$.  Subtracting $s(g-1)/2$ from the numerator gives
\begin{subequations}
\begin{align}
sg-1-(s-1)M-\frac{s(g-1)}{2}
&= \frac{sg+s-2}{2}-(s-1)M
\label{eq:order-regime-C-chain-algebra1}\\
&\ge \frac{sg+s-2}{2}-\frac{(s-1)(g+1)}{4}
\label{eq:order-regime-C-chain-quarter}\\
&= \frac{(s+1)(g+1)-4}{4}
\label{eq:order-regime-C-chain-algebra2}\\
&\ge 0.
\label{eq:order-regime-C-chain-nonnegative}
\end{align}
\end{subequations}
Hence
\begin{align}
R_{\mathrm{conv}}(M)\ge \frac{s}{2}\ge \frac{K-L}{4},
\label{eq:order-lb-C}
\end{align}
and the uniform achievable bound \eqref{eq:order-trivial} gives
\begin{align}
\frac{R_{\mathrm{lift}}(M)}{R_{\mathrm{conv}}(M)}
\le \frac{4K}{K-L}
\le \frac{4K}{\min\{L,K-L\}}
\le \frac{12K}{\min\{L,K-L\}}.
\label{eq:order-C-ratio}
\end{align}

\emph{2) Medium memory:} Assume
\begin{align}
\frac{N}{2(K-L+1)}\le M\le \frac{N}{4}.
\label{eq:order-regime-A1}
\end{align}
The medium-memory converse indices are
\begin{align}
s\triangleq \left\lfloor \frac{N}{2M}\right\rfloor,
\qquad
g\triangleq \left\lfloor \frac{N}{s}\right\rfloor.
\label{eq:order-A1-sq-definition}
\end{align}
The condition $M\le N/4$ gives $s\ge2$.
The reference interval also gives $M\ge1$, hence
$s\le\lfloor N/2\rfloor$.  The lower bound in
\eqref{eq:order-regime-A1} gives $s\le K-L+1$, and therefore
\begin{align}
\min\{s-1,K-L\}=s-1.
\label{eq:order-A1-min-evaluation}
\end{align}
Thus $s$ satisfies both the admissibility and window-size conditions in
Lemma~\ref{lem:floor-selected-converse}, which gives
\begin{subequations}
\begin{align}
g&\ge 2M-1,
\label{eq:order-A1-q-lower}\\
M&\le \frac{g+1}{2},
\label{eq:order-q-half}\\
R_{\mathrm{conv}}(M)
&\ge \frac{sg-1-(s-1)M}{g-1}
\label{eq:order-A1-conv-substitute}\\
&= \frac{s(g-M)+M-1}{g-1}
\label{eq:order-A1-conv-rewrite}\\
&\ge \frac{s}{2},
\label{eq:order-A1-conv-s-over-two}\\
s&\ge \frac{N}{3M},
\label{eq:order-s-lower-A1}\\
R_{\mathrm{conv}}(M)&\ge \frac{N}{6M}.
\label{eq:order-lb-A1}
\end{align}
\end{subequations}
The achievable envelope in Lemma~\ref{lem:order-envelope} and the converse
bound \eqref{eq:order-lb-A1} yield
\begin{align}
\frac{R_{\mathrm{lift}}(M)}{R_{\mathrm{conv}}(M)}
\le \frac{2KN/(\min\{L,K-L\}M)}{N/(6M)}
= \frac{12K}{\min\{L,K-L\}}.
\label{eq:order-A1-ratio}
\end{align}

\emph{3) High memory:} First consider the lower subinterval
\begin{align}
\frac{N}{4}\le M\le \frac{N}{2}.
\label{eq:order-regime-A2}
\end{align}
The $s=1$ term of \eqref{eq:R-conv} is
\begin{align}
R_{\mathrm{conv}}(M)&\ge 1.
\label{eq:order-A2-conv-one}
\end{align}
By Lemma~\ref{lem:order-envelope},
\begin{align}
R_{\mathrm{lift}}(M)\le \frac{2KN}{\min\{L,K-L\}\,M}\le \frac{8K}{\min\{L,K-L\}}.
\label{eq:order-A2-ach-bound}
\end{align}
The last step uses $M\ge N/4$.  Dividing the achievable bound by
the unit lower bound \eqref{eq:order-A2-conv-one} gives
\begin{align}
\frac{R_{\mathrm{lift}}(M)}{R_{\mathrm{conv}}(M)}
\le \frac{8K}{\min\{L,K-L\}}
\le \frac{12K}{\min\{L,K-L\}}.
\label{eq:order-A2-ratio}
\end{align}

The upper subinterval is
\begin{align}
M\ge \frac{N}{2}.
\label{eq:order-regime-B}
\end{align}
The same $s=1$ term gives
\begin{align}
R_{\mathrm{conv}}(M)&\ge 1.
\label{eq:order-B-conv-one}
\end{align}
On this interval, Lemma~\ref{lem:order-envelope} gives
\begin{align}
R_{\mathrm{lift}}(M)\le \frac{2KN}{\min\{L,K-L\}\,M}\le \frac{4K}{\min\{L,K-L\}},
\label{eq:order-B-ach-bound}
\end{align}
and the same unit converse gives
\begin{align}
\frac{R_{\mathrm{lift}}(M)}{R_{\mathrm{conv}}(M)}
\le \frac{4K}{\min\{L,K-L\}}
\le \frac{12K}{\min\{L,K-L\}}.
\label{eq:order-B-ratio}
\end{align}

Let $B=N/[2(K-L+1)]$.  Within $\mathcal M_{\mathrm{ref}}$,
\eqref{eq:order-regime-C} covers memory up to $B$, and
\eqref{eq:order-regime-A1} covers memory from $B$ to $N/4$.
The high-memory case covers $M\ge N/4$.  The two high-memory intervals
\eqref{eq:order-regime-A2} and \eqref{eq:order-regime-B} meet at $N/2$.
Thus \eqref{eq:order-C-ratio}, \eqref{eq:order-A1-ratio},
\eqref{eq:order-A2-ratio}, and \eqref{eq:order-B-ratio} bound the ratio by
$12K/\min\{L,K-L\}$ throughout the reference interval.  The identity
$K/\min\{L,K-L\}=1/\min\{L/K,1-L/K\}$
then gives
\eqref{eq:order-gap-interval-proof}.
\end{IEEEproof}

\begin{IEEEproof}[Proof of Theorem~\ref{thm:order-opt}]
Proposition~\ref{prop:order-gap-interval} proves
\eqref{eq:order-gap-interval-proof} for $M\in\mathcal M_{\mathrm{ref}}$.
For $M>(N-1)(K-1)/L$, the endpoint \eqref{eq:endpoint} and memory
monotonicity give $R_{\mathrm{lift}}(M)=1$.  The term of
\eqref{eq:R-conv} with $s=1$ gives $R_{\mathrm{conv}}(M)\ge 1$,
so the ratio bound also holds above the endpoint.
Lemma~\ref{lem:feasible-memory} gives $M\ge(K-1)/(K-L)$ for every
feasible cache size when $N\ge K$.  The reference interval and the upper
region cover every feasible $M$, proving \eqref{eq:order-opt-main}.
\end{IEEEproof}

\subsection{\texorpdfstring{Comparison When Users Outnumber Files ($K>N$)}{Comparison When Users Outnumber Files (K>N)}}
Throughout this subsection, $2\le N<K$.  We first locate the memory range
covered by the lifted envelope, then compare its rate with the converse.
Let $m_0$ denote the smaller of $L$ and $K-L$.  Define
$M_{\mathrm{feas}}$ as the necessary memory threshold from
Lemma~\ref{lem:feasible-memory}, $M_0$ as the memory of the $t=0$
construction, and $M_\star$ as the memory of the unit-rate endpoint:
\begin{align}
m_0\triangleq\min\{L,K-L\},
\qquad
M_{\mathrm{feas}}\triangleq\frac{N-1}{K-L},
\qquad
M_0\triangleq\frac{K-1}{K-L},
\qquad
M_\star\triangleq\frac{(N-1)(K-1)}{L}.
\label{eq:small-library-basic-parameters}
\end{align}
The lifted envelope is finite from
$M_{\mathrm{start}}=\min\{M_0,M_\star\}$ by
Lemma~\ref{lem:small-library-start}.  The global ratio bound below applies for
$M\ge M_{\mathrm{start}}$.

\begin{Lem}[Starting memory of the lifted envelope]
\label{lem:small-library-start}
For every $t\in\{1,\ldots,K-2\}$, $M(t)>M_0$.
Consequently, for $M\ge0$, the envelope in \eqref{eq:ach-envelope-def} is
finite exactly when $M\ge M_{\mathrm{start}}$.
\end{Lem}
\begin{IEEEproof}
The maximal appearance runs partition each $r$-subset into consecutive
sets of users.  Summing their lengths over all $r$-subsets therefore gives
$r\binom Kr$.  By cyclic symmetry, one $K$th of this total comes from runs
ending at each user, giving
\begin{align}
\sum_{b=1}^{K-1}b\Gamma(K,r,b)
&=\frac rK\binom Kr=\binom{K-1}{r-1},
\qquad 1\le r\le K-1.
\label{eq:small-library-run-moment}
\end{align}
Apply the local-cost inequality $\mu_{b,L}\ge(K-b)/(K-L)$ to
\eqref{eq:M-t}.  Evaluating the resulting sums with
\eqref{eq:order-gamma-t} and \eqref{eq:small-library-run-moment} gives
\begin{align}
M(t)&\ge M_0+\frac{t}{K-L}
\left(\frac{K(N-1)}{K-1}-\frac{N}{K-t}\right),
\qquad 1\le t\le K-2.
\label{eq:small-library-start-lower}
\end{align}
Since $K-t\ge2$ and $N\ge2$, the parenthesized factor satisfies
\begin{align*}
\frac{K(N-1)}{K-1}-\frac{N}{K-t}
&\ge\frac{N(K+1)-2K}{2(K-1)}
\ge\frac{1}{K-1}>0.
\end{align*}
Thus the minimum memory among the $t$-indexed corners is $M(0)=M_0$.
After including the rate-one point, the smallest memory among the points
that define the envelope is $\min\{M_0,M_\star\}$.
Every convex combination in \eqref{eq:ach-envelope-def} has at least this
memory.  A point attaining the minimum is available at every larger cache
budget.
\end{IEEEproof}

\begin{Lem}\label{lem:small-library-memory}
For every $t\in\{0,1,\ldots,K-2\}$,
\begin{align}
M(t)\le
\frac{(N-1)(t+1)+K-N}{m_0}.
\label{eq:small-library-memory-bound}
\end{align}
\end{Lem}
\begin{IEEEproof}
For $t=0$, the identity in \eqref{eq:left-endpoint} and the inequality $m_0\le K-L$ give
\begin{align}
M(0)=\frac{K-1}{K-L}
\le\frac{K-1}{m_0}
=\frac{(N-1)+K-N}{m_0}.
\label{eq:small-library-memory-t0}
\end{align}
For every $t\in\{1,\ldots,K-2\}$, the cyclic-run bound in
Lemma~\ref{lem:order-run-cost} at $r=t$ and $r=t+1$, together with the binomial
cancellations \eqref{eq:order-memory-ratio-file} through 
\eqref{eq:order-memory-ratio-key}, gives
\begin{align}
M(t)
&\le\frac{Nt+K-1-t}{m_0}
\label{eq:small-library-memory-intermediate}
\\
&=\frac{(N-1)(t+1)+K-N}{m_0}.
\label{eq:small-library-memory-rewrite}
\end{align}
Combining this bound with the $t=0$ case in
\eqref{eq:small-library-memory-t0} proves
\eqref{eq:small-library-memory-bound}.
\end{IEEEproof}

The bound \eqref{eq:small-library-memory-bound} depends on the library size
through $N-1$ and $K-N$.  A uniform upper bound on this expression is
\begin{align}
M(t)\le\frac{K(t+1)}{m_0},
\qquad t\in\{0,1,\ldots,K-2\},
\label{eq:small-library-coarse-memory-bound}
\end{align}
because
\begin{align*}
K(t+1)-[(N-1)(t+1)+K-N]
=(K-N+1)t+1>0.
\end{align*}

Suppose $M_0<M_\star$ and $M\in[M_0,M_\star]$.  To bound the rate at
budget $M$, compare two memory-sharing constructions.  The first gives
$V_{<}(M)$: regard each corner as achievable at the upper memory budget in
\eqref{eq:small-library-memory-bound}, as allowed by memory monotonicity.
Interpolate between adjacent corners whose budgets bracket $M$, and use
the $t=0$ point at smaller budgets.  The second construction shares between
$(M_0,K)$ and $(M_\star,1)$.  In the definitions below, $x(M)$ gives the
interpolation position, and $u(M)-1$ is the index of the lower corner when
interpolation is used.  Define $U_{<}(M)$ as the smaller of the two rates:
\begin{align}
x(M)&\triangleq\frac{m_0M-(K-N)}{N-1},
\label{eq:small-library-x-definition}\\
u(M)&\triangleq\max\{1,\lfloor x(M)\rfloor\},
\label{eq:small-library-u-definition}\\
V_{<}(M)&\triangleq
\begin{cases}
K,&x(M)<1,\\[0.5ex]
\dfrac{K\bigl(2u(M)+1-x(M)\bigr)}{u(M)(u(M)+1)},&x(M)\ge1,
\end{cases}
\label{eq:small-library-adjacent-sharing-bound}\\
U_{<}(M)&\triangleq
\min\left\{
V_{<}(M),
1+(K-1)\frac{M_\star-M}{M_\star-M_0}
\right\}.
\label{eq:small-library-pointwise-envelope}
\end{align}

\begin{Lem}\label{lem:small-library-pointwise-envelope}
Assume $M_0<M_\star$ and $M\in[M_0,M_\star]$. Then the quantities in \eqref{eq:small-library-x-definition} through \eqref{eq:small-library-pointwise-envelope} satisfy
\begin{align}
1&\le u(M)\le K-2,
\label{eq:small-library-pointwise-index-range}\\
R_{\mathrm{lift}}(M)&\le U_{<}(M).
\label{eq:small-library-pointwise-achievable-bound}
\end{align}
\end{Lem}
\begin{IEEEproof}
The definition in \eqref{eq:small-library-u-definition} gives $u(M)\ge1$. The upper bound $M\le M_\star$ and the inequality $m_0\le L$ give
\begin{subequations}
\begin{align}
x(M)
&\le
\frac{(N-1)(K-1)m_0/L-(K-N)}{N-1}
\label{eq:small-library-u-upper-substitute}\\
&\le K-1-\frac{K-N}{N-1}
<K-1.
\label{eq:small-library-u-upper-final}
\end{align}
\end{subequations}
Since $u(M)=\max\{1,\lfloor x(M)\rfloor\}$, the bound $x(M)<K-1$
gives \eqref{eq:small-library-pointwise-index-range}.

If $x(M)<1$, the $t=0$ point in \eqref{eq:left-endpoint} is available because $M\ge M_0$, and hence
\begin{align}
R_{\mathrm{lift}}(M)\le K=V_{<}(M).
\label{eq:small-library-x-below-one-rate}
\end{align}
For $x(M)\ge1$, $u(M)=\lfloor x(M)\rfloor$ and
\begin{align}
u(M)\le x(M)<u(M)+1.
\label{eq:small-library-adjacent-index-interval}
\end{align}
For $j\in\{u(M),u(M)+1\}$, define the upper memory budget
\begin{align}
\widehat M(j)\triangleq\frac{(N-1)j+K-N}{m_0}.
\label{eq:small-library-padded-memory}
\end{align}
By \eqref{eq:small-library-pointwise-index-range}, the index
$j-1\in\{0,1,\ldots,K-2\}$ is admissible.
Lemma~\ref{lem:small-library-memory} bounds the memory of corner $t=j-1$
by $\widehat M(j)$, at rate $K/j$.  Memory monotonicity makes
$(\widehat M(j),K/j)$ achievable.  The budgets $\widehat M(u(M))$ and $\widehat M(u(M)+1)$
bracket $M$.

The interpolation coefficient is
\begin{align}
\theta\triangleq u(M)+1-x(M).
\label{eq:small-library-adjacent-sharing-coefficient}
\end{align}
Equations \eqref{eq:small-library-adjacent-sharing-coefficient} and
\eqref{eq:small-library-adjacent-index-interval} place $\theta$ in $(0,1]$.
Memory sharing between the two adjacent corners at these budgets gives
\begin{subequations}
\begin{align}
\theta\widehat M(u(M))+(1-\theta)\widehat M(u(M)+1)
&=\frac{(N-1)x(M)+K-N}{m_0}
\label{eq:small-library-adjacent-sharing-memory-substitute}\\
&=M.
\label{eq:small-library-adjacent-sharing-memory}
\end{align}
\end{subequations}
The identity $\theta u(M)+(1-\theta)(u(M)+1)=x(M)$ and
\eqref{eq:small-library-x-definition} make the mixed memory equal to $M$.
The memory-sharing rate is
\begin{subequations}
\begin{align}
\theta\frac{K}{u(M)}+(1-\theta)\frac{K}{u(M)+1}
&=\frac{K\bigl(2u(M)+1-x(M)\bigr)}{u(M)(u(M)+1)}
\label{eq:small-library-adjacent-sharing-rate-evaluate}\\
&=V_{<}(M).
\label{eq:small-library-adjacent-sharing-rate}
\end{align}
\end{subequations}
Together with the $x(M)<1$ case in
\eqref{eq:small-library-x-below-one-rate}, this proves
$R_{\mathrm{lift}}(M)\le V_{<}(M)$.

For the second rate in $U_{<}(M)$, share between $(M_0,K)$ and
$(M_\star,1)$ with coefficient
\begin{align}
\lambda=\frac{M_\star-M}{M_\star-M_0}.
\label{eq:small-library-endpoint-sharing-coefficient}
\end{align}
The corresponding convex combination has memory $M$ and rate
\begin{align}
\lambda K+(1-\lambda)
=1+(K-1)\frac{M_\star-M}{M_\star-M_0}.
\label{eq:small-library-endpoint-sharing-rate}
\end{align}
Choosing the smaller of the adjacent-corner rate and
\eqref{eq:small-library-endpoint-sharing-rate} gives
\eqref{eq:small-library-pointwise-achievable-bound}.
\end{IEEEproof}

\begin{Lem}\label{lem:small-library-envelope}
If $M_0\le M\le M_\star$, then
\begin{align}
R_{\mathrm{lift}}(M)
\le \frac{2K^2}{m_0M}.
\label{eq:small-library-envelope}
\end{align}
\end{Lem}
\begin{IEEEproof}
Apply Lemma~\ref{lem:floor-selected-envelope} with $a=K/m_0$ and $T=K-2$.
Bound \eqref{eq:small-library-coarse-memory-bound} gives
$M(t)\le a(t+1)$, and $M\ge M_0$ makes the $t=0$ rate-$K$ point available.  For
$m_0M/K\ge1$, the floor-selected index is
\begin{align}
u\triangleq\left\lfloor\frac{m_0M}{K}\right\rfloor,
\qquad t\triangleq u-1.
\label{eq:small-library-coarse-envelope-index}
\end{align}
The endpoint bounds the floor argument by
\begin{align}
\frac{m_0M}{K}
\le\frac{m_0}{L}\frac{(N-1)(K-1)}{K}
<N-1<K-1.
\label{eq:small-library-coarse-index-range}
\end{align}
Thus $M<a(T+2)$, and the lemma gives
$2Ka/M=2K^2/(m_0M)$ as the upper bound in
\eqref{eq:small-library-envelope}.
\end{IEEEproof}

\begin{Lem}\label{lem:small-library-converse}
If $M_0\le M\le N/4$, then
\begin{align}
R_{\mathrm{conv}}(M)\ge\frac{N}{6M}.
\label{eq:small-library-converse}
\end{align}
\end{Lem}
\begin{IEEEproof}
For $M\in[M_0,N/4]$, the converse indices are
\begin{align}
s\triangleq\left\lfloor\frac{N}{2M}\right\rfloor,
\qquad
g\triangleq\left\lfloor\frac Ns\right\rfloor.
\label{eq:small-library-converse-index}
\end{align}
The bounds $M\ge M_0\ge1$ and $M\le N/4$ imply $2\le s\le\lfloor N/2\rfloor<K$. Since $N\le K-1$,
\begin{align}
N(K-L)
\le (K-1)(K-L)
\le 2(K-1)(K-L+1),
\label{eq:small-library-M0-comparison}
\end{align}
and division by the positive denominators gives
\begin{align}
M_0=\frac{K-1}{K-L}
\ge\frac{N}{2(K-L+1)}.
\label{eq:small-library-M0-lower}
\end{align}
Consequently $s\le N/(2M)\le K-L+1$, so
\begin{align}
\min\{s-1,K-L\}=s-1.
\label{eq:small-library-min-evaluation}
\end{align}
Lemma~\ref{lem:floor-selected-converse} therefore applies and gives
\begin{subequations}
\begin{align}
R_{\mathrm{conv}}(M)
&\ge\frac{sg-1-(s-1)M}{g-1}
\label{eq:small-library-converse-substitute}\\
&=\frac{s(g-M)+M-1}{g-1}
\label{eq:small-library-converse-rewrite}\\
&\ge\frac{s}{2}.
\label{eq:small-library-converse-s-half}
\end{align}
\end{subequations}
The lemma also bounds the selected index by
\begin{align}
s\ge\frac{N}{3M}.
\label{eq:small-library-converse-s-lower}
\end{align}

Combining \eqref{eq:small-library-converse-s-lower} with
\eqref{eq:small-library-converse-s-half} gives
\eqref{eq:small-library-converse}.
\end{IEEEproof}

\begin{IEEEproof}[Proof of Proposition~\ref{prop:order-opt-small-library}]
The $s=1$ term in \eqref{eq:R-conv} gives
\begin{align}
R_{\mathrm{conv}}(M)\ge1
\label{eq:small-library-unit-converse}
\end{align}
for every $M$.  First suppose $M_\star\le M_0$, so
$M_{\mathrm{start}}=M_\star$.  The endpoint in \eqref{eq:endpoint} gives
$R_{\mathrm{lift}}(M)\le1$ for every $M\ge M_{\mathrm{start}}$.
Every component rate in \eqref{eq:ach-envelope-def} is at least one,
so $R_{\mathrm{lift}}(M)=1$.
The achievable endpoint and $R_{\mathrm{conv}}(M)\le R^*(M)$ give
$R_{\mathrm{conv}}(M)\le1$.  Together with
\eqref{eq:small-library-unit-converse}, this gives
$R_{\mathrm{conv}}(M)=1$ throughout this regime.

Now suppose $M_0<M_\star$.  For $M\ge M_\star$, the endpoint and unit
converse give
$R_{\mathrm{lift}}(M)=R_{\mathrm{conv}}(M)=1$.
For $M_0\le M\le M_\star$, Lemma~\ref{lem:small-library-pointwise-envelope}
gives \eqref{eq:small-library-pointwise-achievable-bound}. If $M\le N/4$,
Lemma~\ref{lem:small-library-converse} gives
\begin{align}
\frac{R_{\mathrm{lift}}(M)}{R_{\mathrm{conv}}(M)}
\le\frac{U_{<}(M)}{N/(6M)}
=\frac{6M}{N}U_{<}(M).
\label{eq:small-library-pointwise-ratio-low}
\end{align}
For $M>N/4$, \eqref{eq:small-library-pointwise-achievable-bound} and
\eqref{eq:small-library-unit-converse} imply
\begin{align}
\frac{R_{\mathrm{lift}}(M)}{R_{\mathrm{conv}}(M)}
\le U_{<}(M).
\label{eq:small-library-pointwise-ratio-high}
\end{align}
Equations~\eqref{eq:small-library-pointwise-ratio-low} and
\eqref{eq:small-library-pointwise-ratio-high} bound
$R_{\mathrm{lift}}(M)/R_{\mathrm{conv}}(M)$ throughout $[M_0,M_\star]$.

To obtain the uniform ratio bound on $[M_0,M_\star]$, use the simpler
envelope in Lemma~\ref{lem:small-library-envelope}.
If $M\le N/4$, Lemmas~\ref{lem:small-library-envelope} and
\ref{lem:small-library-converse} give
\begin{align}
\frac{R_{\mathrm{lift}}(M)}{R_{\mathrm{conv}}(M)}
\le
\frac{2K^2/(m_0M)}{N/(6M)}
=\frac{12K^2}{Nm_0}.
\label{eq:small-library-ratio-low-memory}
\end{align}
If $M>N/4$, Lemma~\ref{lem:small-library-envelope} and
\eqref{eq:small-library-unit-converse} give
\begin{align}
\frac{R_{\mathrm{lift}}(M)}{R_{\mathrm{conv}}(M)}
\le\frac{2K^2}{m_0M}
<\frac{8K^2}{Nm_0}
\le\frac{12K^2}{Nm_0}.
\label{eq:small-library-ratio-high-memory}
\end{align}
Using $m_0/K=\min\{L/K,1-L/K\}$, the bound $12K^2/(Nm_0)$ in
\eqref{eq:small-library-ratio-low-memory} and
\eqref{eq:small-library-ratio-high-memory} becomes
\eqref{eq:small-library-order-main} on $M_0\le M\le M_\star$.
The unit-rate equality applies for $M\ge M_\star$ and throughout the
$M_\star\le M_0$ regime.  Thus \eqref{eq:small-library-order-main}
holds for every $M\ge M_{\mathrm{start}}$.

To remove the factor $K/N$ from the ratio bound, use adjacent-corner
interpolation for $M\ge M_{\mathrm{imp}}$.  The definition of $M_{\mathrm{imp}}$ and
$m_0\le K-L$ give
\begin{align}
M_{\mathrm{imp}}\ge\frac{K-1}{m_0}\ge M_0.
\label{eq:small-library-Mimp-ge-M0}
\end{align}
The ratio is one for $M\ge M_\star$.  On
$M_{\mathrm{imp}}\le M<M_\star$, bound \eqref{eq:small-library-Mimp-ge-M0}
gives $M\ge M_0$ and $M_0<M_\star$, so
Lemma~\ref{lem:small-library-pointwise-envelope} applies.  The definition of
$M_{\mathrm{imp}}$ gives the two memory inequalities
\begin{align}
m_0M\ge K-1,
\qquad
m_0M\ge2(K-N).
\label{eq:small-library-refined-memory-conditions}
\end{align}
For $x(M)$ in \eqref{eq:small-library-x-definition}, the condition
$m_0M\ge K-1$ in \eqref{eq:small-library-refined-memory-conditions} gives
$x(M)\ge1$, so the adjacent-sharing branch in
\eqref{eq:small-library-adjacent-sharing-bound} applies.
For $u(M)\le x(M)<u(M)+1$, completing the square bounds the
interpolation rate as follows:
\begin{subequations}
\begin{align}
\frac{x(M)V_{<}(M)}{K}
&=\frac{x(M)(2u(M)+1-x(M))}{u(M)(u(M)+1)}
\label{eq:small-library-adjacent-factor-exact}\\
&\le\frac{(u(M)+1/2)^2}{u(M)(u(M)+1)}
\label{eq:small-library-adjacent-factor-square}\\
&=1+\frac{1}{4u(M)(u(M)+1)}
\label{eq:small-library-adjacent-factor-identity}\\
&\le\frac98.
\label{eq:small-library-adjacent-factor-final}
\end{align}
\end{subequations}
The bound on the quadratic numerator follows by completing the square:
\begin{align*}
x(M)(2u(M)+1-x(M))
=(u(M)+1/2)^2-(x(M)-u(M)-1/2)^2
\le (u(M)+1/2)^2.
\end{align*}
The last inequality in \eqref{eq:small-library-adjacent-factor-final} uses
$u(M)\ge1$.
The pointwise envelope bound gives
\begin{subequations}
\begin{align}
R_{\mathrm{lift}}(M)
&\le V_{<}(M)
\label{eq:small-library-refined-ach-adjacent}\\
&\le\frac{9K(N-1)}{8[m_0M-(K-N)]}
\label{eq:small-library-refined-ach-denominator}\\
&\le\frac{9K(N-1)}{4m_0M}.
\label{eq:small-library-refined-ach-final}
\end{align}
\end{subequations}
Substituting $x(M)=[m_0M-(K-N)]/(N-1)$ into the
$9/8$ bound and applying $m_0M-(K-N)\ge m_0M/2$ gives
\eqref{eq:small-library-refined-ach-adjacent} through \eqref{eq:small-library-refined-ach-final}.

If $M\le N/4$, Lemma~\ref{lem:small-library-converse} and \eqref{eq:small-library-refined-ach-final} yield
\begin{subequations}
\begin{align}
\frac{R_{\mathrm{lift}}(M)}{R_{\mathrm{conv}}(M)}
&\le
\frac{9K(N-1)/(4m_0M)}{N/(6M)}
\label{eq:small-library-refined-ratio-low-substitute}\\
&=\frac{27K(N-1)}{2Nm_0}
\label{eq:small-library-refined-ratio-low-evaluate}\\
&\le\frac{27K}{2m_0}.
\label{eq:small-library-refined-ratio-low-final}
\end{align}
\end{subequations}
If $M>N/4$, inequalities \eqref{eq:small-library-unit-converse} and
\eqref{eq:small-library-refined-ach-final} give
\begin{subequations}
\begin{align}
\frac{R_{\mathrm{lift}}(M)}{R_{\mathrm{conv}}(M)}
&\le\frac{9K(N-1)}{4m_0M}
\label{eq:small-library-refined-ratio-high-substitute}\\
&<\frac{9K(N-1)}{Nm_0}
\label{eq:small-library-refined-ratio-high-memory}\\
&\le\frac{9K}{m_0}
\label{eq:small-library-refined-ratio-high-library}\\
&\le\frac{27K}{2m_0}.
\label{eq:small-library-refined-ratio-high-final}
\end{align}
\end{subequations}
Both \eqref{eq:small-library-refined-ratio-low-final} and
\eqref{eq:small-library-refined-ratio-high-final} bound the ratio by
$27K/(2m_0)$ for $M\ge M_{\mathrm{imp}}$.  Taking the minimum with
\eqref{eq:small-library-order-main} gives
\eqref{eq:small-library-refined-order-main} after substituting
$K/m_0=1/\min\{L/K,1-L/K\}$.

\end{IEEEproof}

\Needspace{6\baselineskip}
\subsection{\texorpdfstring{Constant-Factor Bound for $L=2$}{Constant-Factor Bound for L=2}}
Throughout this subsection,
\begin{align*}
L=2,\qquad N\ge K\ge 4.
\end{align*}
Define the midpoint corner index by
\begin{align}
h\triangleq \left\lfloor \frac{K-1}{2}\right\rfloor.
\label{eq:L2-h}
\end{align}

\begin{Lem}\label{lem:L2-explicit}
For each $t\in\{0,1,\ldots,h\}$, define the coefficients $A_t$ and $B_t$ by
\begin{align}
A_t&=
\frac{t(K-t-1)}{(K-2)^2}
+
\frac{t(t-1)(2K-t-2)}{2(K-2)(K-1)(K-t)},
\label{eq:L2-A}\\
B_t&=
\frac{2K^3-2K^2t-8K^2+Kt^2+3Kt+10K-4}
{2(K-2)^2(K-1)}.
\label{eq:L2-B}
\end{align}
Then the memory at corner $t$ is
\begin{align}
M(t)=N A_t+B_t.
\label{eq:L2-memory-decomp}
\end{align}
\end{Lem}
\begin{IEEEproof}
Specializing \eqref{eq:T} to $L=2$ gives
\begin{align}
T_{\mathrm{blk}}(K,s,2)
=
\frac{K-1}{K-2}\Gamma(K,s,1)
+
\sum_{b=2}^{K-1}\frac{b}{2}\Gamma(K,s,b).
\label{eq:L2-T-pre}
\end{align}
For $s=0$, Lemma~\ref{lem:cycliccount} gives
$T_{\mathrm{blk}}(K,0,2)=0$.  For every
$s\in\{1,\ldots,h+1\}$, the condition $K\ge4$ gives $s\le K-2$.
Therefore Lemma~\ref{lem:cycliccount} yields
\begin{align}
\Gamma(K,s,1)=\binom{K-3}{s-1},
\qquad
\Gamma(K,s,b)=\binom{K-b-2}{s-b},\quad 2\le b\le s.
\label{eq:L2-gamma-specialization}
\end{align}
Let $S_s$ denote the weighted binomial sum
\begin{align}
S_s\triangleq \sum_{b=2}^{s}b\binom{K-b-2}{s-b}.
\label{eq:L2-b-sum-def}
\end{align}
Substituting \eqref{eq:L2-gamma-specialization} into \eqref{eq:L2-T-pre}
rewrites the $b\ge2$ term as $S_s/2$.
Here and below, binomial coefficients with negative lower indices are zero.
With $j=s-b$, the ordinary and weighted hockey-stick identities give
\begin{align}
S_s
=s\binom{K-3}{s-2}-(K-s-1)\binom{K-3}{s-3}.
\label{eq:L2-b-sum}
\end{align}
Identities \eqref{eq:L2-b-sum-reindex} through \eqref{eq:L2-b-sum-hockey2}
give \eqref{eq:L2-b-sum}.  Substituting this sum into
\eqref{eq:L2-T-pre} gives, for $1\le s\le h+1$,
\begin{align}
T_{\mathrm{blk}}(K,s,2)
&=\frac{K-1}{K-2}\binom{K-3}{s-1}
+\frac{s}{2}\binom{K-3}{s-2}
-\frac{K-s-1}{2}\binom{K-3}{s-3}.
\label{eq:L2-T-closed}
\end{align}

The coefficient of $N$ in $M(t)$ is
\begin{align}
A_t=\frac{T_{\mathrm{blk}}(K,t,2)}{\binom{K-1}{t}}.
\label{eq:L2-A-coefficient-ratio}
\end{align}
The zero case has $A_0=0$. For $1\le t\le h$, the specialization $s=t$ in
\eqref{eq:L2-T-closed} and the three binomial ratios give
\begin{align}
A_t
&=\frac{t(K-t-1)}{(K-2)^2}
+\frac{t(t-1)(2K-t-2)}{2(K-2)(K-1)(K-t)}.
\label{eq:L2-A-closed}
\end{align}
The ratios \eqref{eq:L2-A-ratio-one} through \eqref{eq:L2-A-ratio-three}
give \eqref{eq:L2-A-closed} for $1\le t\le h$.
The equality $A_0=0$ completes \eqref{eq:L2-A} at the zero index.

The constant term in $M(t)$ is
\begin{align}
B_t=\frac{T_{\mathrm{blk}}(K,t+1,2)}{\binom{K-1}{t}}.
\label{eq:L2-B-coefficient-ratio}
\end{align}
Substituting $s=t+1$ in \eqref{eq:L2-T-closed} and simplifying the
binomial ratios gives
\begin{align}
B_t
&=\frac{2K^3-2K^2t-8K^2+Kt^2+3Kt+10K-4}
{2(K-2)^2(K-1)}.
\label{eq:L2-B-expanded}
\end{align}
Equations~\eqref{eq:L2-B-ratio-one} through
\eqref{eq:L2-B-numerator-identity} give the expression for $B_t$ in
\eqref{eq:L2-B-expanded}.  Combining it with
\eqref{eq:L2-A-closed} proves \eqref{eq:L2-memory-decomp}.
\end{IEEEproof}

\begin{Lem}[Coefficient bounds]\label{lem:L2-coefficient-bounds}
For $K\ge6$ and $0\le t\le h$, the coefficients in \eqref{eq:L2-A} through \eqref{eq:L2-B} satisfy
\begin{align*}
A_t\le\frac{3t}{2K},
\qquad
B_t\le\frac{3N}{2K}.
\end{align*}
\end{Lem}
\begin{IEEEproof}
For $K\ge10$, formula \eqref{eq:L2-A-diff} and the positive lower bound
\eqref{eq:L2-P-positive} give $A_t\le3t/(2K)$.  The values in
\eqref{eq:L2-A-small-K-values} are nonnegative for $6\le K\le9$,
with equality at $t=0$.  The numerator of $B_t$ has the negative derivative
\eqref{eq:L2-B-numerator-derivative}.  The denominator of $B_t$ is positive
and independent of $t$, so $B_t$ decreases on $0\le t\le h$.
Equation~\eqref{eq:L2-B-bound} gives $B_t\le3N/(2K)$.
\end{IEEEproof}

\begin{Lem}\label{lem:L2-memory}
For every $t\in\{0,1,\ldots,h\}$,
\begin{align}
M(t)\le \frac{3N(t+1)}{2K}.
\label{eq:L2-memory-bound}
\end{align}
\end{Lem}
\begin{IEEEproof}
For $K\ge6$, Lemmas~\ref{lem:L2-explicit} and \ref{lem:L2-coefficient-bounds} give
\begin{align}
M(t)=NA_t+B_t
\le \frac{3Nt}{2K}+\frac{3N}{2K}
=\frac{3N(t+1)}{2K}.
\label{eq:L2-memory-combine}
\end{align}
For $K\in\{4,5\}$, direct substitution in \eqref{eq:L2-memory-decomp} gives
\begin{align}
K=4:\quad
M(0)&=\frac32\le\frac{3N}{8},&
M(1)&=\frac N2+\frac56\le\frac{3N}{4},
\label{eq:L2-memory-K4-check}\\
K=5:\quad
M(0)&=\frac43\le\frac{3N}{10},&
M(1)&=\frac N3+\frac{11}{12}\le\frac{3N}{5},\nonumber\\
&&M(2)&=\frac{11N}{18}+\frac{23}{36}\le\frac{9N}{10}.
\label{eq:L2-memory-K5-check}
\end{align}
Each inequality follows from $N\ge K$ for $K\in\{4,5\}$.
\end{IEEEproof}

\begin{Lem}\label{lem:L2-midpoint}
With $h$ as in \eqref{eq:L2-h},
\begin{align}
M(h)\le N-1,
\qquad
R(h)\le 2.
\label{eq:L2-midpoint}
\end{align}
\end{Lem}
\begin{IEEEproof}
Since \(h+1\ge K/2\), \eqref{eq:R-t} gives
\begin{align}
R(h)=\frac{K}{h+1}\le 2.
\label{eq:L2-midpoint-rate-proof}
\end{align}

Lemma~\ref{lem:L2-memory} also gives
\begin{align}
M(h)&\le\frac{3N(h+1)}{2K}\nonumber\\
&\le
\begin{cases}
3N/4,&K\text{ even},\\[0.3ex]
3N(K+1)/(4K),&K\text{ odd}.
\end{cases}
\label{eq:L2-midpoint-memory-parity-bound}
\end{align}
For even $K\ge4$, $3N/4\le N-1$. For odd $K\ge7$, the second branch is at
most $N-1$ because $N(K-3)/(4K)\ge(K-3)/4\ge1$.  At $K=5$,
\eqref{eq:L2-memory-K5-check} gives
$M(2)=11N/18+23/36\le N-1$ for $N\ge5$.  Together with
\eqref{eq:L2-midpoint-rate-proof}, these bounds give \eqref{eq:L2-midpoint}.
\end{IEEEproof}

\begin{Lem}\label{lem:L2-envelope}
The achievable envelope satisfies
\begin{align}
R_{\mathrm{lift}}(M)\le \frac{3N}{M},
\qquad \frac{K-1}{K-2}\le M\le N-1,
\label{eq:L2-envelope-low}\\
R_{\mathrm{lift}}(M)\le 2,
\qquad M\ge N-1.
\label{eq:L2-envelope-high}
\end{align}
\end{Lem}
\begin{IEEEproof}
For $(K-1)/(K-2)\le M<M(h)$, take $a=3N/(2K)$ and $T=h-1$.
Since $K\ge4$, $a>0$ and $T\ge0$ is an integer.
Theorem~\ref{thm:memoryload} gives the achievable points indexed by
$t\in\{0,\ldots,T\}$ with rate $K/(t+1)$, and
Lemma~\ref{lem:L2-memory} gives $M(t)\le a(t+1)$.
The rate-$K$ endpoint in \eqref{eq:left-endpoint} is available at budget $M$
by memory monotonicity.  The same memory bound gives
\begin{align*}
0<M<M(h)\le\frac{3N(h+1)}{2K}=a(T+2).
\end{align*}
Lemma~\ref{lem:floor-selected-envelope} therefore yields
\begin{align}
R_{\mathrm{lift}}(M)\le\frac{2Ka}{M}=\frac{3N}{M}.
\label{eq:L2-envelope-selected-final}
\end{align}

For $M(h)\le M\le N-1$, Lemma~\ref{lem:L2-midpoint} and memory monotonicity give
\begin{align}
R_{\mathrm{lift}}(M)\le R(h)\le2\le\frac{3N}{M}.
\label{eq:L2-envelope-midpoint-chain}
\end{align}
The two intervals cover the memory range in \eqref{eq:L2-envelope-low}.
For every $M\ge N-1$, Lemma~\ref{lem:L2-midpoint} gives $M\ge M(h)$,
so memory monotonicity yields
$R_{\mathrm{lift}}(M)\le R(h)\le2$, which is \eqref{eq:L2-envelope-high}.
\end{IEEEproof}

\begin{IEEEproof}[Proof of Proposition~\ref{prop:L2-order}]
Lemma~\ref{lem:feasible-memory}, specialized to \(L=2\), gives
\begin{align}
M\ge \frac{K-1}{K-2}
\label{eq:L2-feasible-memory-lower}
\end{align}
for every feasible cache size when \(N\ge K\).  The point in
\eqref{eq:left-endpoint} is therefore available at every feasible memory
budget by memory monotonicity:
\begin{align}
R_{\mathrm{lift}}(M)\le K
\label{eq:L2-trivial}
\end{align}
for every feasible $M$.

\emph{1) Very small memory:} Assume
\begin{align}
\frac{K-1}{K-2}\le M\le \frac{N}{2(K-1)}.
\label{eq:L2-very-small-memory-interval}
\end{align}
The two bounds on $M$ in \eqref{eq:L2-very-small-memory-interval} give
\begin{align}
N\ge 2M(K-1)\ge \frac{2(K-1)^2}{K-2}>2K-2.
\label{eq:L2-very-small-admissibility-chain}
\end{align}
The inequality $(K-1)/(K-2)>1$ gives the strict bound in
\eqref{eq:L2-very-small-admissibility-chain}.  Since $N$ is an integer,
$K-1\le\lfloor N/2\rfloor$, making $s=K-1$ admissible in
\eqref{eq:R-conv}.  Define
\begin{align}
g\triangleq \left\lfloor \frac{N}{K-1}\right\rfloor.
\label{eq:L2-very-small-q-definition}
\end{align}
For $g$ in \eqref{eq:L2-very-small-q-definition},
$\lfloor x\rfloor\ge x-1$ and $M\le N/(2(K-1))$ imply
$g\ge N/(K-1)-1\ge 2M-1$, and therefore
\begin{align}
M\le \frac{g+1}{2}.
\label{eq:L2-very-small-q-half}
\end{align}
The $s=K-1$ term of the converse is
\begin{align}
R_{\mathrm{conv}}(M)
\ge \frac{(K-1)g-1-(K-2)M}{g-1}.
\label{eq:L2-very-small-conv-substitute}
\end{align}
Equation~\eqref{eq:L2-very-small-q-half} yields
\begin{subequations}
\begin{align}
(K-1)g-1-(K-2)M-\frac{K(g-1)}{2}
&=(K-2)\left(\frac{g+1}{2}-M\right)
\label{eq:L2-very-small-numerator-rewrite}\\
&\ge 0.
\label{eq:L2-very-small-numerator-nonnegative}
\end{align}
\end{subequations}
Thus the numerator in \eqref{eq:L2-very-small-conv-substitute}
is at least $K(g-1)/2$, giving
\begin{align}
R_{\mathrm{conv}}(M)\ge \frac{K}{2}.
\label{eq:L2-converse-very-small}
\end{align}
Dividing the achievable rate bound \eqref{eq:L2-trivial} by
\eqref{eq:L2-converse-very-small} gives
\begin{align}
\frac{R_{\mathrm{lift}}(M)}{R_{\mathrm{conv}}(M)}\le 2.
\label{eq:L2-ratio-very-small}
\end{align}

\emph{2) Low memory:} Assume
\begin{align}
\frac{N}{2(K-1)}\le M\le \frac{N}{2}.
\label{eq:L2-low-memory-interval}
\end{align}
Set
\begin{align}
s\triangleq \left\lfloor \frac{N}{2M}\right\rfloor,
\qquad
g\triangleq \left\lfloor \frac{N}{s}\right\rfloor.
\label{eq:L2-low-sq-definition}
\end{align}
The interval endpoints give $1\le s\le K-1$.
The feasible-memory bound \eqref{eq:L2-feasible-memory-lower} gives
$M>1$, so $s<N/2$ and $s\le\lfloor N/2\rfloor$.
Thus $s$ is admissible in \eqref{eq:R-conv},
$\min\{s-1,K-2\}=s-1$, and Lemma~\ref{lem:floor-selected-converse} gives
\begin{align}
g=\left\lfloor\frac Ns\right\rfloor\ge 2M-1,
\qquad\text{and hence}\qquad
M\le\frac{g+1}{2}.
\label{eq:L2-low-q-half}
\end{align}
For the selected $s$ and $g$, the converse term is
\begin{subequations}
\begin{align}
R_{\mathrm{conv}}(M)
&\ge \frac{sg-1-(s-1)M}{g-1}
\label{eq:L2-low-conv-substitute}\\
&=\frac{s(g-M)+M-1}{g-1}.
\label{eq:L2-low-conv-rewrite}
\end{align}
\end{subequations}
Equation~\eqref{eq:floor-converse-general} gives
\begin{align}
R_{\mathrm{conv}}(M)\ge \frac{s}{2}\ge \frac{N}{8M}.
\label{eq:L2-converse-low}
\end{align}
Lemma~\ref{lem:L2-envelope} and \eqref{eq:L2-converse-low} bound the ratio by
\begin{align}
\frac{R_{\mathrm{lift}}(M)}{R_{\mathrm{conv}}(M)}
\le
\frac{3N/M}{N/(8M)}
=24.
\label{eq:L2-ratio-low}
\end{align}

\emph{3) High memory:} First consider
\begin{align}
\frac{N}{2}\le M\le N-1.
\label{eq:L2-intermediate-memory-interval}
\end{align}
The $s=1$ term of \eqref{eq:R-conv} is
\begin{align}
R_{\mathrm{conv}}(M)\ge 1.
\label{eq:L2-intermediate-conv-one}
\end{align}
Lemma~\ref{lem:L2-envelope} gives
\begin{align}
R_{\mathrm{lift}}(M)\le \frac{3N}{M}\le 6.
\label{eq:L2-intermediate-ach-bound}
\end{align}
Dividing \eqref{eq:L2-intermediate-ach-bound} by the unit lower bound
\eqref{eq:L2-intermediate-conv-one} gives
\begin{align}
\frac{R_{\mathrm{lift}}(M)}{R_{\mathrm{conv}}(M)}\le 6.
\label{eq:L2-ratio-intermediate}
\end{align}

For larger cache budgets, consider
\begin{align}
M\ge N-1.
\label{eq:L2-large-memory-interval}
\end{align}
On this interval, $R_{\mathrm{conv}}(M)\ge1$ and
$R_{\mathrm{lift}}(M)\le2$, so
\begin{align}
\frac{R_{\mathrm{lift}}(M)}{R_{\mathrm{conv}}(M)}\le 2.
\label{eq:L2-ratio-large}
\end{align}

Bounds \eqref{eq:L2-ratio-very-small}, \eqref{eq:L2-ratio-low},
\eqref{eq:L2-ratio-intermediate}, and \eqref{eq:L2-ratio-large}
cover every feasible $M$.  Each is at most $24$, proving
\eqref{eq:L2-order-main}.
\end{IEEEproof}

\subsubsection{\texorpdfstring{Algebraic Identities for the $L=2$ Bounds}{Algebraic Identities for the L=2 Bounds}}

The change of variable $j=s-b$ in \eqref{eq:L2-b-sum-def} separates the
sum into ordinary and weighted binomial sums:
\begin{subequations}
\begin{align}
S_s
&=\sum_{j=0}^{s-2}(s-j)\binom{K-s+j-2}{j}
\label{eq:L2-b-sum-reindex}\\
&=s\sum_{j=0}^{s-2}\binom{K-s+j-2}{j}
-\sum_{j=1}^{s-2}j\binom{K-s+j-2}{j}.
\label{eq:L2-b-sum-split}
\end{align}
\end{subequations}
The ordinary and weighted sums follow from the hockey-stick identity:
\begin{align}
\sum_{j=0}^{s-2}\binom{K-s+j-2}{j}
&=\binom{K-3}{s-2},
\label{eq:L2-b-sum-hockey1}\\
j\binom{K-s+j-2}{j}
&=(K-s-1)\binom{K-s+j-2}{j-1},
\label{eq:L2-b-sum-weight-identity}\\
\sum_{j=1}^{s-2}j\binom{K-s+j-2}{j}
&=(K-s-1)\binom{K-3}{s-3}.
\label{eq:L2-b-sum-hockey2}
\end{align}
Equations~\eqref{eq:L2-b-sum-hockey1} and
\eqref{eq:L2-b-sum-hockey2} reduce \eqref{eq:L2-b-sum-split} to
\eqref{eq:L2-b-sum}, which yields \eqref{eq:L2-T-closed}.

For $A_t$, the ratios to $\binom{K-1}{t}$ are
\begin{subequations}
\begin{align}
\frac{\binom{K-3}{t-1}}{\binom{K-1}{t}}
&=\frac{t(K-t-1)}{(K-1)(K-2)},
\label{eq:L2-A-ratio-one}\\
\frac{\binom{K-3}{t-2}}{\binom{K-1}{t}}
&=\frac{t(t-1)}{(K-1)(K-2)},
\label{eq:L2-A-ratio-two}\\
\frac{\binom{K-3}{t-3}}{\binom{K-1}{t}}
&=\frac{t(t-1)(t-2)}{(K-1)(K-2)(K-t)}.
\label{eq:L2-A-ratio-three}
\end{align}
\end{subequations}
Consequently,
\begin{subequations}
\begin{align}
A_t
&=\frac{t(K-t-1)}{(K-2)^2}
+\frac{t^2(t-1)}{2(K-1)(K-2)}
-\frac{(K-t-1)t(t-1)(t-2)}{2(K-1)(K-2)(K-t)}
\label{eq:L2-A-substitution}\\
&=\frac{t(K-t-1)}{(K-2)^2}
+\frac{t(t-1)}{2(K-1)(K-2)}
\left(t-\frac{(K-t-1)(t-2)}{K-t}\right).
\label{eq:L2-A-factor}
\end{align}
\end{subequations}
The bracket simplifies to
\begin{align}
t-\frac{(K-t-1)(t-2)}{K-t}
=\frac{2K-t-2}{K-t}.
\label{eq:L2-A-bracket-identity}
\end{align}
Substituting this bracket into \eqref{eq:L2-A-factor} gives
\eqref{eq:L2-A-closed}.  To evaluate $B_t$, use the ratios
\begin{subequations}
\begin{align}
\frac{\binom{K-3}{t}}{\binom{K-1}{t}}
&=\frac{(K-t-1)(K-t-2)}{(K-1)(K-2)},
\label{eq:L2-B-ratio-one}\\
\frac{\binom{K-3}{t-1}}{\binom{K-1}{t}}
&=\frac{t(K-t-1)}{(K-1)(K-2)},
\label{eq:L2-B-ratio-two}\\
\frac{\binom{K-3}{t-2}}{\binom{K-1}{t}}
&=\frac{t(t-1)}{(K-1)(K-2)}.
\label{eq:L2-B-ratio-three}
\end{align}
\end{subequations}
Substitution of these three ratios into \eqref{eq:L2-B-coefficient-ratio}
gives
\begin{subequations}
\begin{align}
B_t
&=\frac{(K-t-1)(K-t-2)}{(K-2)^2}
+\frac{t(t+1)(K-t-1)}{2(K-1)(K-2)}
-\frac{t(t-1)(K-t-2)}{2(K-1)(K-2)}
\label{eq:L2-B-substitution}\\
&=\frac{(K-t-1)(K-t-2)}{(K-2)^2}
+\frac{t(2K-t-3)}{2(K-1)(K-2)}
\label{eq:L2-B-combine-last-two}\\
&=\frac{2(K-1)(K-t-1)(K-t-2)+t(2K-t-3)(K-2)}
{2(K-2)^2(K-1)}.
\label{eq:L2-B-common-denominator}
\end{align}
\end{subequations}
The two terms with a common denominator combine by the identity
\begin{align}
(t+1)(K-t-1)-(t-1)(K-t-2)=2K-t-3.
\label{eq:L2-B-numerator-identity}
\end{align}
Expanding the numerator gives \eqref{eq:L2-B-expanded}.

The numerator of $B_t$ has derivative
\begin{align}
K(2t-2K+3),
\label{eq:L2-B-numerator-derivative}
\end{align}
which is negative on $0\le t\le h$. Hence
\begin{align}
B_t\le B_0=\frac{K-1}{K-2}\le\frac32\le\frac{3N}{2K}.
\label{eq:L2-B-bound}
\end{align}
For $K\ge6$, the file coefficient satisfies
\begin{align}
A_t\le\frac{3t}{2K},
\qquad 0\le t\le h.
\label{eq:L2-A-target}
\end{align}
For $K\ge10$, set
\begin{align}
u\triangleq K-2t.
\label{eq:L2-u-definition}
\end{align}
Define the polynomial $Q(K,t)$ by
\nopagebreak
\begin{align}
Q(K,t)
&=K^4-K^3t-9K^3-K^2t^2+14K^2t\nonumber\\*
&\quad+16K^2-24Kt-8K+12t.
\label{eq:L2-A-diff-Q}
\end{align}
Bringing the terms of $3t/(2K)-A_t$ to a common denominator gives
\begin{align}
\frac{3t}{2K}-A_t
=\frac{tQ(K,t)}{2K(K-1)(K-2)^2(K-t)}.
\label{eq:L2-A-diff-t}
\end{align}
The polynomial $Q(K,t)$ also has the unexpanded form
\begin{align}
&3(K-1)(K-2)^2(K-t)-2K(K-t-1)(K-1)(K-t)\nonumber\\
&\quad-K(t-1)(2K-t-2)(K-2).
\label{eq:L2-A-diff-Q-unexpanded}
\end{align}
Define the polynomial $P(K,u)$ by
\begin{align}
P(K,u)
&=K^4+4K^3u-8K^3-K^2u^2-28K^2u+16K^2
 +48Ku-8K-24u
\label{eq:L2-P-definition}\\
&=K^2\big((K-4)^2+u(4K-28-u)\big)+24u(2K-1)-8K.
\label{eq:L2-P-rearranged}
\end{align}
The substitution $t=(K-u)/2$ gives
\begin{align}
Q\!\left(K,\frac{K-u}{2}\right)&=\frac{P(K,u)}4,
\label{eq:L2-A-diff-Q-to-P}\\
\frac{3t}{2K}-A_t
&=\frac{tP(K,u)}{8K(K-1)(K-2)^2(K-t)}.
\label{eq:L2-A-diff}
\end{align}
For $K\ge10$, the concave quadratic $u(4K-28-u)$ is nonnegative at
$u=0$ and $u=K$.  It is therefore nonnegative throughout $0\le u\le K$, giving
\begin{align}
P(K,u)&\ge K^2(K-4)^2-8K>0,
\label{eq:L2-P-positive}\\
K^2(K-4)^2-8K
&=K\bigl(K(K-4)^2-8\bigr)
\ge K(10\cdot6^2-8)>0.
\label{eq:L2-P-coarse-positive}
\end{align}
Hence \eqref{eq:L2-A-target} holds for $K\ge10$.
For $6\le K\le9$, the values of $3t/(2K)-A_t$ for $1\le t\le h$ are
\begin{align}
K=6&:\quad 0,\ \frac{1}{40},\nonumber\\
K=7&:\quad \frac{1}{70},\ \frac{22}{525},\ \frac{81}{1400},\nonumber\\
K=8&:\quad \frac{1}{48},\ \frac{25}{504},\ \frac{121}{1680},\nonumber\\
K=9&:\quad \frac{1}{42},\ \frac{31}{588},\ \frac{61}{784},\ \frac{61}{735},
\label{eq:L2-A-small-K-values}
\end{align}
listed in increasing $t$ for each $K$.  Equality holds at $t=0$, and each
listed value is nonnegative.  Thus \eqref{eq:L2-A-target} holds for
$6\le K\le9$.  Combined with the bound for $K\ge10$, this covers every
$K\ge6$.

\section{Proofs for the Exact Largest-Access Tradeoff}
\label{app:largest-access}

\begin{Remark}[Normalized entropy converses]
Fix $(N,K,L)$ and collect the files, caches, and all $N^K$
demand-indexed broadcasts into a finite tuple $V$.
Along an achievable sequence, the normalized entropy vectors satisfy
\[
\frac{H(V)}{F_j}\le N+KM_{F_j}+N^KR_{F_j}.
\]
Every coordinate is bounded, so these normalized entropy vectors have a
coordinatewise convergent subsequence.  The limiting vector satisfies the
Shannon inequalities, file independence, the entropy equalities for recovery
and secrecy, and the resource bounds $H(Z_k)\le M$ and
$H(X_{\mathbf d})\le R$ in normalized units.
The entropy inequalities below and the single-access endpoint converse
of \cite{FangSmallRate} depend only on these relations, so they also hold
for the limiting vector.  In the following rate-one argument, the entropy
notation denotes coordinates of this normalized limiting vector.
Put $C_k=Z_{\mathcal A_k}$.  Marginalizing secrecy for requests 1 and 2 gives
$I(W_{2:N};C_k\mid W_1)=0$ and $I(W_1;C_k\mid W_2)=0$.
File independence implies $I(W_1;C_k)\le I(W_1;C_k,W_2)=0$, so the chain rule
gives $I(W_{[N]};C_k)=0$.  At rate one, recovery from $(C_k,X_{\mathbf d})$
therefore yields
\[
H(X_{\mathbf d}\mid C_k)
=1+H(X_{\mathbf d}\mid C_k,W_{d_k})
\le H(X_{\mathbf d})\le1.
\]
Thus $H(X_{\mathbf d}\mid C_k)=1$ and
$H(X_{\mathbf d}\mid C_k,W_{d_k})=0$, the endpoint identities for the normalized
entropy limit of the induced single-access code sequence.
Each finite-code inequality uses that code's budgets $M_F,R_F$ before taking
the limit.
\end{Remark}

\begin{Lem}[Conditional secrecy implies observation secrecy]
\label{lem:supp-security-equivalence}
Let \(O\) be the observation of a user requesting \(W_d\) under a fixed
public demand vector.  Suppose
\begin{align*}
I(W_{[N]\setminus\{d\}};O\mid W_d)=0.
\end{align*}
For independent files,
\(I(W_{[N]\setminus\{d\}};O)=0\).  The conclusion also holds for
any subtuple of the nonrequested files and any deterministic function
of \(O\).
\end{Lem}

\begin{IEEEproof}
File independence and the chain rule give
\begin{align*}
I(W_{[N]\setminus\{d\}};O)
&\le I(W_{[N]\setminus\{d\}};O,W_d)\\
&=I(W_{[N]\setminus\{d\}};W_d)
+I(W_{[N]\setminus\{d\}};O\mid W_d)=0.
\end{align*}
The claims for file subtuples and deterministic functions of \(O\) follow
by data processing.
\end{IEEEproof}

\subsection{Effective-cache reduction}

\begin{Lem}[Largest-access effective-cache specialization]
For \(L=K-1\), every SMACC code of cache size \(M\) induces a secure
single-access code with the same delivery rate and per-user cache budget at
most \((K-1)M\).  Consequently,
\begin{align}
R^*_{K,K-1,N}(M)
\ge R^*_{\mathrm{SA},K,N}((K-1)M),
\label{eq:supp-lkm1-single-access-reduction}
\end{align}
and every achievable pair with \(R\le1\) satisfies \(M\ge N-1\).
\end{Lem}

\begin{IEEEproof}
In the original cyclic labeling, user \(k\) omits exactly cache
\(\langle k-1\rangle_K\).  Relabel that omitted cache as \(Z_k\) and put
\begin{align}
\mathcal B_k=[K]\setminus\{k\},
\qquad
\widehat Z_k\triangleq Z_{\mathcal B_k}.
\end{align}
At blocklength \(F\), let \(\widehat{\mathcal Z}_k^{(F)}\) be the virtual-cache
alphabet: the Cartesian product of the constituent physical-cache alphabets.
The fixed-length
physical-cache constraints give
\begin{align}
\log_q\bigl|\widehat{\mathcal Z}_k^{(F)}\bigr|
&\le\sum_{i\in\mathcal B_k}\log_q\bigl|\mathcal Z_i^{(F)}\bigr|
\le(K-1)M_F F.
\end{align}
The virtual cache thus fits within budget \((K-1)M_F\).  Since
\(\widehat Z_k\) is exactly the
original access view, reusing the delivery signal preserves
\begin{align}
H(W_{d_k}\mid\widehat Z_k,X_{\mathbf d},\mathbf d)&=0,
\notag\\
I\!\left(W_{[N]\setminus\{d_k\}};
\widehat Z_k,X_{\mathbf d}\middle|W_{d_k},\mathbf d\right)&=0.
\end{align}
Passing to the normalized limsup bounds along an achievable sequence gives a single-access pair
with memory budget \((K-1)M\) and the same rate, giving
\eqref{eq:supp-lkm1-single-access-reduction}.

For every achievable pair with \(R\le1\), applying the rate-one entropy
converse of \cite{FangSmallRate} to the normalized entropy limit vector of the
induced single-access code sequence gives
\((K-1)M\ge(K-1)(N-1)\).
Dividing by \(K-1>0\) yields \(M\ge N-1\).
\end{IEEEproof}

\subsection{The two achievable corners}

At \(t=0\), the run formula gives \(T_{\rm blk}(K,0,K-1)=0\).  For incidence
size one, only \(b=1\) contributes, and with \(K-L=1\),
\begin{align}
T_{\rm blk}(K,1,K-1)
=\frac{K-1}{K-L}\Gamma(K,1,1)=K-1.
\end{align}
The achievable formulas therefore give
\((M(0),R(0))=(K-1,K)\).
The unit-rate formula \(M=(N-1)(K-1)/L\) gives the second endpoint
\((N-1,1)\).

\subsection{Elementary converse}

\begin{IEEEproof}[Proof of Lemma~\ref{lem:lkm1-basic}]
Fix a length-\(F\) code with budgets \(M_F,R_F\).
Fix user \(u\), file \(W_n\), and a demand with \(d_u=n\).  A second demand
assigns user \(u\) a different file, which exists because \(N\ge2\).
Under the second demand, Lemma~\ref{lem:supp-security-equivalence} implies
\(I(W_n;Z_{\mathcal B_u})=0\).  Correctness under the first demand yields
\begin{align}
F=H(W_n\mid Z_{\mathcal B_u})
=I(W_n;X_{\mathbf d}\mid Z_{\mathcal B_u})
\le H(X_{\mathbf d})\le R_F F.
\end{align}
Thus every finite code has \(R_F\ge1\), so every achievable pair has \(R\ge1\).

Set \(D=\min\{N,K\}\).  Choose a demand in which the first
\(D\) users request distinct files \(W_1,\ldots,W_D\).
If \(D=N<K\), the remaining users repeat an already requested file.
Collective decoding gives
\begin{align}
H(W_{1:D}\mid Z_{[K]},X_{\mathbf d})=0.
\label{eq:supp-lkm1-D-decode}
\end{align}
User 1 requests \(W_1\).  By
Lemma~\ref{lem:supp-security-equivalence}, its complete observation is
unconditionally independent of \(W_{2:D}\):
\begin{align}
I(W_{2:D};Z_{\mathcal B_1},X_{\mathbf d})=0.
\label{eq:supp-lkm1-D-security}
\end{align}
The decomposition \(Z_{[K]}=(Z_{\mathcal B_1},Z_1)\) yields
\begin{align}
DF
&=I(W_{1:D};Z_{[K]},X_{\mathbf d})
\notag\\
&=I(W_{1:D};Z_{\mathcal B_1},X_{\mathbf d})
+I(W_{1:D};Z_1\mid Z_{\mathcal B_1},X_{\mathbf d})
\notag\\
&=I(W_1;Z_{\mathcal B_1},X_{\mathbf d}\mid W_{2:D})
+I(W_{1:D};Z_1\mid Z_{\mathcal B_1},X_{\mathbf d})
\notag\\
&\le H(W_1)+H(Z_1)\le(1+M_F)F.
\end{align}
The bound \(DF\le(1+M_F)F\) gives \(M_F\ge D-1\).
Passing to the normalized limsup bound yields
\(M\ge D-1=\min\{N,K\}-1\).
\end{IEEEproof}

\subsection{\texorpdfstring{Affine converse for \(N\ge K+1\)}{Affine converse for N>= K+1}}

\begin{IEEEproof}[Proof of Proposition~\ref{prop:lkm1-line}]
Fix the distinct-demand vector \(\mathbf d=(1,2,\ldots,K)\), abbreviate
\(X=X_{1,\ldots,K}\), and let \(\mathcal B_1=[K]\setminus\{1\}\).

To make placement entropies depend only on the number of files included,
we symmetrize over file labels.  This operation preserves the memory and rate
budgets, correctness, and security.
Starting from a length-\(F\) code with budgets \(M_F,R_F\), take \(N!\)
independent copies indexed by
permutations \(\pi\) of \([N]\), with independent file layers and random seeds.
Each new file \(W_n=(W_n^{[\pi]})_\pi\) has length \(F'=N!F\).
Copy \(\pi\) applies the original maps to the virtual library
\((W_{\pi^{-1}(n)}^{[\pi]})_{n\in[N]}\) and demand
\((\pi(d_1),\ldots,\pi(d_K))\), so its user \(k\) recovers
\(W_{d_k}^{[\pi]}\).  The concatenated cache contents and broadcasts let each
user recover every layer of its requested file.  The file layers and random
seeds are independent across copies, so conditional leakage is zero jointly
across the layers.  The logarithms of the concatenated alphabet sizes add,
preserving the normalized memory and rate budgets.
For \(\mathcal S\subseteq[N]\), the placement entropy
\(H(W_{\mathcal S},Z_{[K]})\) depends only on \(|\mathcal S|\).
Entropies involving broadcasts are invariant under
simultaneous permutations of file and demand labels.  The entropy calculations
below use \(F\) for the new file length and \(M_F,R_F\) for the preserved
normalized budgets of the symmetrized code.

The affine converse follows from the entropy inequality
\begin{align}
(K-1)H(Z_1)+(N-K)H(X)
\ge\bigl((K-1)(N-1)+(N-K)\bigr)F.
\label{eq:supp-lkm1-entropy-target}
\end{align}
The fixed-length constraints bound its left side by
\(((K-1)M_F+(N-K)R_F)F\), and the right coefficient equals \(KN-2K+1\).

All \(K\) requested files are recoverable from \((Z_{[K]},X)\).
Decoding, Lemma~\ref{lem:supp-security-equivalence}, and subadditivity give
\begin{align}
&H(W_{K+1:N},Z_{[K]},X)+H(Z_1)
\notag\\
&=H(W_{[N]},Z_{[K]},X)+H(Z_1)
\notag\\
&\ge H(W_{2:N},Z_{\mathcal B_1},X)+H(Z_1)
\notag\\
&=H(W_{2:N})+H(Z_{\mathcal B_1},X)+H(Z_1)
\notag\\
&\ge(N-1)F+H(Z_{[K]},X).
\label{eq:supp-lkm1-isolation}
\end{align}
Here collective decoding supplies \(W_{1:K}\), and user-1 secrecy makes
\(W_{2:N}\) independent of \((Z_{\mathcal B_1},X)\).

The joint cache and broadcast entropies also satisfy
\begin{align}
&(K-1)H(Z_{[K]},X)+(N-K)H(Z_{\mathcal B_1},X)
\notag\\*
&\quad\ge
(K-1)H(W_{K+1:N},Z_{[K]},X)
+(N-K)H(W_1)
+(N-K)H(Z_{\mathcal B_1}).
\label{eq:supp-lkm1-transfer}
\end{align}
Placement is independent of demands.  Applying
Lemma~\ref{lem:supp-security-equivalence} to a demand in which user 1 requests
another file and taking the cache marginal gives
\begin{align}
I(W_1;Z_{\mathcal B_1})=0.
\end{align}
User 1 decodes \(W_1\) from \((Z_{\mathcal B_1},X)\).
Decodability and \(I(W_1;Z_{\mathcal B_1})=0\) make
\eqref{eq:supp-lkm1-transfer} equivalent to
\begin{align}
(N-K)H(X\mid W_1,Z_{\mathcal B_1})
\ge(K-1)H(W_{K+1:N}\mid Z_{[K]},X).
\label{eq:supp-lkm1-transfer-midgoal}
\end{align}
Correctness and conditioning give
\begin{align}
H(X\mid W_1,Z_{\mathcal B_1})
&\ge H(X\mid W_1,Z_{[K]})
\notag\\
&=H(X,W_{2:K}\mid W_1,Z_{[K]})
\notag\\
&\ge H(W_{2:K}\mid W_1,Z_{[K]}),
\label{eq:supp-lkm1-early-transfer}\\
H(W_{K+1:N}\mid Z_{[K]},X)
&=H(W_{K+1:N}\mid W_{1:K},Z_{[K]},X)
\notag\\
&\le H(W_{K+1:N}\mid W_{1:K},Z_{[K]}).
\label{eq:supp-lkm1-late-transfer}
\end{align}

We use file symmetry to compare the conditional entropies of
\(W_{2:K}\) and \(W_{K+1:N}\) in
\eqref{eq:supp-lkm1-early-transfer} and \eqref{eq:supp-lkm1-late-transfer}.
For each integer \(a\in\{0,\ldots,N\}\), define \(\mathcal E(a)\) below,
with \(W_{1:0}\) denoting the empty tuple:
\begin{align}
\mathcal E(a)=H(W_{1:a},Z_{[K]}),
\qquad0\le a\le N.
\end{align}
Thus \(\mathcal E(0)=H(Z_{[K]})\).  For \(1\le a\le N-1\), file symmetry gives
\(H(W_{2:a+1},Z_{[K]})=\mathcal E(a)\).  Submodularity applied to
\((W_{1:a},Z_{[K]})\) and \((W_{2:a+1},Z_{[K]})\) yields
\begin{align}
2\mathcal E(a)\ge\mathcal E(a-1)+\mathcal E(a+1).
\end{align}
Thus the increments are nonincreasing, and the average of increments
\(2,\ldots,K\) is no smaller than the average of increments
\(K+1,\ldots,N\):
\begin{align}
\frac{\mathcal E(K)-\mathcal E(1)}{K-1}
\ge
\frac{\mathcal E(N)-\mathcal E(K)}{N-K}.
\end{align}
Equivalently,
\begin{align}
(N-K)H(W_{2:K}\mid W_1,Z_{[K]})
\ge(K-1)H(W_{K+1:N}\mid W_{1:K},Z_{[K]}).
\label{eq:supp-lkm1-gap-comparison}
\end{align}
Combining \eqref{eq:supp-lkm1-early-transfer},
\eqref{eq:supp-lkm1-late-transfer}, and
\eqref{eq:supp-lkm1-gap-comparison} proves
\eqref{eq:supp-lkm1-transfer-midgoal}.
Hence \eqref{eq:supp-lkm1-transfer} also holds.

Subadditivity gives
\begin{align}
H(Z_{\mathcal B_1})+H(X)\ge H(Z_{\mathcal B_1},X).
\label{eq:supp-lkm1-final-subadditivity}
\end{align}
Adding \((K-1)\) times \eqref{eq:supp-lkm1-isolation},
\eqref{eq:supp-lkm1-transfer}, and \((N-K)\) times
\eqref{eq:supp-lkm1-final-subadditivity} cancels the joint-entropy terms and
gives \eqref{eq:supp-lkm1-entropy-target}.  Dividing by \(F\) and applying the
fixed-length constraints gives
\((K-1)M_F+(N-K)R_F\ge KN-2K+1\).
Taking limsup along an achievable sequence gives the affine converse.
\end{IEEEproof}

\subsection{Exact tradeoff and the two-user case}

\begin{IEEEproof}[Complete proof of Theorem~\ref{thm:lkm1-boundary}]
The corners \((K-1,K)\) and \((N-1,1)\) are achievable.  Each code remains
valid at any larger cache budget.

If \(N\ge K+1\), the elementary bound excludes \(M<K-1\).  Memory sharing
between the two corners gives, for \(K-1\le M\le N-1\),
\begin{align}
R_{\rm ms}(M)
=\frac{KN-2K+1-(K-1)M}{N-K}.
\label{eq:supp-lkm1-memory-sharing-line}
\end{align}
The affine converse gives the matching lower bound.  For \(M\ge N-1\), the
unit-rate corner is achievable and the elementary converse gives \(R\ge1\).
The achievable rate therefore equals the converse throughout the feasible
range for \(N\ge K+1\).

If \(K\ge N\), including \(K=N\), the elementary converse gives simultaneously
\(M\ge N-1\) and \(R\ge1\).
The unit-rate code attains both bounds and is valid at every \(M\ge N-1\).
Hence \(R^*(M)=1\) throughout this range.
\end{IEEEproof}

\begin{IEEEproof}[Proof of Corollary~\ref{cor:k2-single-access}]
When \(K=2\), the standing constraint forces \(L=1=K-1\).  For \(N\ge3\),
the exact boundary in \eqref{eq:supp-lkm1-memory-sharing-line} specializes to
\begin{align}
R^*_{2,1,N}(M)=\frac{2N-3-M}{N-2},
\qquad1\le M\le N-1.
\end{align}
The exact rate is one for \(M\ge N-1\), and the feasible-memory threshold is
\(M=1\).  For \(N=2\), the exact rate is one throughout the feasible range
\(M\ge1\).
\end{IEEEproof}

\end{document}